\documentclass[review]{elsarticle}

\usepackage{lineno} %% hyperref elsewhere
\modulolinenumbers[5]

\journal{Journal of Approximation Theory}

\usepackage{enumitem} %%enumerate
\usepackage{versions}
\usepackage{graphicx}
\usepackage{fancyvrb}
\usepackage[english]{babel}
\usepackage{amssymb,dsfont,mathtools}
\usepackage{bbm} %% mathbb vs mathbbm?  I use bbm for indicator fns.
\usepackage{amsthm} %% do i want/ use this correctly?
\usepackage{epsfig, subfigure}
\usepackage{natbib}

\usepackage{color}

\usepackage[colorlinks,citecolor=blue,urlcolor=blue]{hyperref}

\usepackage{url}

\usepackage{multirow}
\usepackage{booktabs}

\usepackage{xparse} %% multi-optional-arguments in macros.
\usepackage{etoolbox} %% gives 'ifstrempty' command

\usepackage[usenames,dvipsnames]{xcolor}

\newcommand{\mm}{m} %% Not sure what symbol to use for 'k' in 'k-monotone'.

\newcommand{\lp}{\left(} %% )
  \newcommand{\rp}{\right)}
\newcommand{\lb}{\left\{} %% )
  \newcommand{\rb}{\right\}}
\newcommand{\ls}{\left[} %% )
  \newcommand{\rs}{\right]}
\newcommand{\la}{\left\langle}
\newcommand{\ra}{\right\rangle}
\newcommand{\lv}{\left\vert}
\newcommand{\rv}{\right\vert}

\newcommand{\ve}[1]{\boldsymbol{#1}}

\newcommand{\RR}{\mathbb{R}}

\newcommand{\NN}{\mathbb{N}}

\newcommand{\II}{{\mathbb I}}

\newcommand{\ib}{\boldsymbol{i}}
\newcommand{\jb}{\boldsymbol{j}}
\newcommand{\Gb}{\boldsymbol{\Gamma}}
\newcommand{\bs}[1]{\boldsymbol{#1}}

\newcommand{\one}{\mathbbm{1}}
\newcommand{\mc}[1]{\mathcal {#1}}

\newcommand{\SPAN}{{\mathrm{span}}}
\newcommand{\relbd}{\partial_r}
\newcommand{\bd}{\partial}

\newcommand{\bracketing}{N_{[\,]}} %% deprecated
\newcommand{\Cthree}[3]{\mathcal{C}\left( {#1}, {#2}, {#3} \right)} %%deprecated , \cC
\newcommand{\commentC}[1]{{\noindent  \color{Maroon}#1}}  %% 'out-of-paper' text
\newenvironment{longform}{\color{blue}}{}
\newenvironment{notes}{\color{blue} {\bf Note:}}{}
\newenvironment{old}{\color{Sepia} {\bf Past Version:}}{} %% not ready to
\NewDocumentCommand\Nb{O{}}{
  \ifstrempty{#1}{
    N_{[\, ]}
  }{
    N_{[\, ]}\left( {#1} \right)
  }
}

\NewDocumentCommand\Nc{O{}}{
  \ifstrempty{#1}{
    N
  }{
    N\left( {#1} \right)
  }
}

\NewDocumentCommand\cQ{O{}}{
  \ifstrempty{#1}{
    \mc{Q}
  }{
    \mc Q \left( {#1} \right)
  }
}

\NewDocumentCommand\cC{O{}}{
  \ifstrempty{#1}{
    \mc C
  }{
    \mc C \left( {#1} \right)
  }
}

\NewDocumentCommand\cCz{O{}}{
  \ifstrempty{#1}{
    \mc{C}_0
  }{
    \mc{C}_0 \left( {#1} \right)
  }
}

\NewDocumentCommand\cCdh{O{}}{
  \ifstrempty{#1}{
    \mc{C}_{\delta,h}
  }{
    \mc{C}_{#1}
  }
}

\NewDocumentCommand\cL{O{}}{
  \ifstrempty{#1}{
    \mc L
  }{
    \mc L \left( {#1} \right)
  }
}

\NewDocumentCommand\CLz{O{}}{
  \ifstrempty{#1}{
    \mathcal{C\kern-.11emL}_0
  }{
    \mathcal{C\kern-.11emL}_0  \left( {#1} \right)
  }
}

\NewDocumentCommand\CL{O{}}{
  \ifstrempty{#1}{
    \mathcal{C\kern-.11emL}
  }{
    \mathcal{C\kern-.11emL}  \left( {#1} \right)
  }
}

\NewDocumentCommand\deriv{m O{}}{
  \ifstrempty{#2}{
  \frac{d}{d #1}
  }{
  \frac{d #2}{d #1}
  }
}

\NewDocumentCommand\pderiv{m O{} O{}}{
  \ifstrempty{#2}{
        \ifstrempty{#3}{
          \frac{\partial}{\partial #1}}{      
          \frac{\partial^{#3}}{\partial #1^{#3}}
          }
  }{
        \ifstrempty{#3}{
          \frac{\partial #2}{\partial #1}}{      
          \frac{\partial^{#3} #2}{\partial #1^{#3}}
          }
  }
}
\NewDocumentCommand\derivtwo{m m O{}}{
  \ifstrempty{#3}{
  \frac{\partial^2}{\partial #1 \partial #2}
  }{
  \frac{\partial^2 #3}{\partial #1 \partial #2}
  }
}

\global\long\def\inv#1{\frac{1}{#1}}

\theoremstyle{definition}\newtheorem{problem}{Problem}[section]
\theoremstyle{definition}\newtheorem{definition}[problem]{Definition}
\theoremstyle{definition}\newtheorem{assumption}{Assumption}
\theoremstyle{remark}
\theoremstyle{remark}
\theoremstyle{remark}\newtheorem{remark}[problem]{Remark}
\theoremstyle{definition}\newtheorem{example}[problem]{Example}
\theoremstyle{plain}\newtheorem{theorem}[problem]{Theorem}
\theoremstyle{plain}\newtheorem{lemma}[problem]{Lemma}
\theoremstyle{plain}\newtheorem{proposition}[problem]{Proposition}
\theoremstyle{plain}\newtheorem{corollary}[problem]{Corollary}
\theoremstyle{plain}
\theoremstyle{definition}

\DeclareMathOperator{\vol}{Vol}
\DeclareMathOperator{\Vol}{Vol}
\DeclareMathOperator{\diam}{diam}

\DeclareMathOperator{\dom}{dom}

\def\listofnotation{\begin{tabbing}
  $\Nb(\epsilon, {\cal F}, d) $~~~~~~~~~\=\parbox{4.4in}{
    bracketing number \dotfill  \pageref{symbol:bracketing-number}}\\
  \addnotation L_p: {$L_p(f)$ is $\lp \int f^p(x) dx \rp^{1/p}$ if $1 \le p < \infty$ or $\sup_x |f(x)|$ if $p=\infty$}{symbol:Lp}
  \addnotation \mathcal{C(\cdot, \ldots, \cdot)}: {convex function
    classes}{symbol:convex-fns-B-Gamma}
  \addnotation \Vert \cdot \Vert: {Euclidean distance:  $\Vert z \Vert^2 :=
    \sum_{i=1}^d z_i^2$ for $z \in \RR^d$}{symbol:Euclidean-norm}
  \addnotationtwo D: {$d$-dimensional convex
    set; $d$-dimensional convex polytope}{symbol:D-general}{symbol:D-polytope}
  \addnotation \bd; \relbd: {boundary; relative boundary}{symbol:relbd}
  \addnotation d^+(x,G,v): {}{symbol:d+}
  \addnotation d(x,H,v): {}{symbol:d(x,H,v)}
  \addnotation d(E,H,v): {}{symbol:d(E,H,v)}
  \addnotation N: {number of halfspaces defining $D$}{symbol:N}
  \addnotation E_i: {halfspaces}{symbol:Ei}
  \addnotation H_i: {hyperplanes}{symbol:Hi}
  \addnotation F_i: {facets of convex polytope $D$}{symbol:Fi}
  \addnotation v_i: {(inner) normal vectors for $H_i$}{symbol:vi}
  \addnotation G_{\jb}: {faces of polytope $D$ }{symbol:vi}
  \addnotation G_{\ib,\jb}: {}{symbol:vi}
  \addnotation A_{\jb}, x_{\jb}: {}{symbol:Ajxj}
  \addnotation \gamma_{\jb}: {}{symbol:gammaj}
  \addnotation \delta_i, A, u : {}{symbol:deltai-u}
  \addnotation J_k: {}{symbol:Jk}
  \addnotation I_k: {}{symbol:Ik}
  \addnotationtwo E_{\jb}; e \equiv e_{\ib,\jb}: {See Proposition~\ref{prop:basis-lipschitz}}{symbol:e-axes-def1}{symbol:e-basis}
  \addnotation h(D,x): {support function}{symbol:h}
  \addnotation w(D): {width}{symbol:w}
  \addnotation \II_k: {multi-indices $\ib$ with $|\ib|=k$}{symbol:IIk}
\end{tabbing}

 \clearpage}
\def\addnotation#1: #2#3{$#1$ \> \parbox{4.4in}{#2 \dotfill \pageref{#3}}\\}
\def\addnotationtwo#1: #2#3#4{$#1$ \> \parbox{4.4in}{#2 \dotfill \pageref{#3}, \pageref{#4}}\\}
\def\newnot#1{\label{#1}}

\excludeversion{mylongform}

\excludeversion{mynotes}

\excludeversion{myold}

\begin{document}

\begin{frontmatter}

\title{Bracketing numbers of convex and $m$-monotone functions on polytopes}
% \title{Elsevier \LaTeX\ template\tnoteref{mytitlenote}}
% \tnotetext[mytitlenote]{Fully documented templates are available in the elsarticle package on \href{http://www.ctan.org/tex-archive/macros/latex/contrib/elsarticle}{CTAN}.}

%% \titlerunning{Bracketing Convex Functions}        % if too long for running head

%%%% For affiliation, note: can group by affiliation (use \address then)
%%%% or can just by footnote

%% Group authors per affiliation :
\author{Charles R. Doss\fnref{crd-affiliation}}

\address{School of Statistics\\
  University of Minnesota \\
  224 Church St SE \#313\\
  Minneapolis, MN 55455
}
\fntext[crd-affiliation]{Supported by NSF grant DMS-$1712664$}

% %% or include affiliations in footnotes:
% \author[mymainaddress,mysecondaryaddress]{Elsevier Inc}

% \author[mysecondaryaddress]{Global Customer Service\corref{mycorrespondingauthor}}
% \cortext[mycorrespondingauthor]{Corresponding author}

% \address[mymainaddress]{1600 John F Kennedy Boulevard, Philadelphia}
% \address[mysecondaryaddress]{360 Park Avenue South, New York}

\begin{abstract}
We study bracketing covering numbers for spaces of bounded convex functions in the $L_p$ norms.  Bracketing numbers are crucial quantities for understanding asymptotic behavior for many statistical nonparametric estimators.  Bracketing number upper bounds in the supremum distance are known for bounded classes that also have a fixed Lipschitz constraint.  However, in most settings of interest, the classes that arise do not include Lipschitz constraints, and so standard techniques based on known bracketing numbers cannot be used.  In this paper, we find upper bounds for bracketing numbers of classes of convex functions without Lipschitz constraints on arbitrary polytopes.  Our results are of particular interest in many multidimensional estimation problems based on convexity shape constraints.

Additionally, we show other applications of our proof methods; in particular we define a new class of multivariate functions, the so-called $m$-monotone functions.  Such functions have been considered mathematically and statistically in the univariate case but never in the multivariate case.  We show how our proof for convex bracketing upper bounds also applies to the $m$-monotone case.
\end{abstract}


\end{frontmatter}

\linenumbers

%% Search \highlightC
%% prefer notes envir to \commentC

%% Use \citep rather than (\cite{}) to allow to go between bracketed citations and name-year

\begin{mylongform}
  \begin{longform}
    The bound \eqref{eq:vol-Gij} can be improved.  It is a hyperrectangle
    area instead of a parallelotope area; the latter could be much smaller.
  \end{longform}
\end{mylongform}

\section{Introduction and Motivation}
\label{sec:intr}

To quantify the size of an infinite dimensional set, the pioneering work of \cite{Kolmogorov:1961tl} studied the so-called metric entropy of the set, which is the logarithm of the metric covering number of the set.  In this paper, we are interested in a related quantity, the bracketing entropy for a class of functions, which serves a similar purpose as metric entropy.  Metric or bracketing entropies quantify the amount of information it takes to approximate any element of a set with a given accuracy $\epsilon > 0$.  This quantity is important in many areas of statistics and information theory; in particular, the asymptotic behavior of empirical processes and thus of many statistical estimators is fundamentally tied to the entropy of related classes of functions under consideration \citep{Dudley:1978tl}.

Let $\cal F$ be a set of
functions on some space $\mc X$ and let $\rho$ be a metric on $\cal F$.
Given a pair of functions $l, u$ on $\mc X$, a {\em bracket} $[l,u]$ is the set of all functions $f \colon {\mc X} \to \RR$ with $l \le f \le u$ pointwise.
%% (We do not actually force $l,u \in \mc{F}$.)
For $\epsilon > 0$, we say $[l,u]$ is an $\epsilon$-bracket (for $\rho$) if $\rho(l,u) \le \epsilon$.
Then  the  $\epsilon$-bracketing number of $\cal F$, denoted
$\bracketing(\epsilon, {\cal F}, \rho)$, is the smallest integer  $N$ such that there
exist $\epsilon$-brackets $[l_i,u_i]$, $i=1,\ldots, N$, such that
for all $f \in {\cal F}$,
$f \in [l_i,u_i]$ for some $i$.
%% there exists $i$ with $l_i(x) \le f(x) \le u_i(x)$ for all $ x \in {\mc X}$.
(We do not actually force $l_i, u_i \in {\mc{F}}$.)
% exist brackets $[l_i,u_i]$, $i=1,\ldots, N$, such that $d(l_i,u_i) \le
% \epsilon$, and for all $ f \in {\cal F}$, there exists $i$ with $l_i(x) \le
% f(x) \le u_i(x)$ for all $x \in \mc X$.
\newnot{symbol:bracketing-number} The
bracketing entropy is the logarithm of the bracketing number.  Like metric
entropies, bracketing entropies are fundamentally tied to rates of
convergence of certain estimators (see e.g., \cite{Birge:1993by},
\cite{vanderVaart:1996tf}, \cite{vandeGeer:2009tk}).  In this paper, we study
the bracketing entropy of classes of convex functions. Our interest is
motivated by the study of nonparametric estimation of functions satisfying
convexity restrictions, such as the least-squares estimator of a convex or
concave regression function on $\RR^d$ (e.g., \cite{Seijo:2011ko},
\cite{Guntuboyina:2013tl}), possibly in the high dimensional setting
\citep{Xu:2014uc}, or estimators of a log-concave or $s$-concave density
(e.g., \cite{MR2766867}, \cite{MR2722462}, \cite{Kim:2014wa},
%% \cite{Doss-Wellner:2013, Doss-Wellner:2013Mode},
\cite{Doss-Wellner:2013,Doss:2016vqa,Doss:2016ux},
among others).  Entropy
bounds, of the metric or bracketing type, are directly relevant for studying
asymptotic behavior of estimators in these contexts.

Fix the dimension $d \in \lb 2, 3, \ldots \rb$. Let $D \subset \RR^d$ be a convex set,
let $v_1, \ldots, v_d \in \RR^d$, be
linearly independent vectors, let $B, \Gamma_1, \ldots, \Gamma_d$ be positive
reals, and let ${\ve v} = (v_1, \ldots, v_d)$ and ${\ve \Gamma} = (\Gamma_1,
\ldots, \Gamma_d)$.
For $f: D \to \RR$, let $L_{p,D}(f) \equiv L_p(f) = \lp \int_D f(x)^p \, dx\rp^{1/p}$ for $1 \le p < \infty$, and let $L_\infty(f) = \sup_{x \in D} |f(x)|$.\newnot{symbol:Lp}
We will let $\cC$ with various arguments denote different classes of convex functions.
We let $\cC \equiv \cC_d$ be the class of convex functions on $\RR^d$, where
we consider all convex functions $f$ to be defined on all of $\RR^d$ and to take the value $\infty$ off of its {\em effective domain} $\dom(f) := \lb x \in \RR^d : f(x) < \infty \rb$
\citep{Rockafellar:1997ww}.  (This approach does not affect bracketing numbers.)
For a function $f$ and a set $D \subset \RR^d$, we will use the notation $f \colon D \to \RR$ to mean that $\dom(f) = D$ and we let $\mc{C}_d(D) \equiv \cC[D]$ be the class of convex functions on $\RR^d$ with  $\dom(f) = D$.
%% $\cC[D]$ helps with my "colon problem" in set notation with functions ...
Then we let
%% $\Cthree D B {{\ve \Gamma}, {\ve v}} $ be
\begin{equation}
  \label{eq:10}
  \Cthree D B {{\ve \Gamma}, {\ve v}} : =
  \lb
  %% f : D \to \RR   \; \vert
  f \in \cC[D]   :  %% solves double colon
  %% f \in \cC,
  L_\infty(f) \le B ,
  | f(x + \lambda v_i) - f(x) | \le \Gamma_i |\lambda|
  \mbox{ if } x, x+\lambda v_i \in D
  \rb
\end{equation}
be the class of  convex functions on $D$ satisfying uniform boundedness and Lipschitz constraints given by $B$ and $\ve \Gamma$.
When $\lb v_1,\ldots,v_n \rb$ is the standard basis of $\RR^d$, we just write $\cC[D, B, \Gb ]$ for this class.  If $D$ is the hyperrectangle $\prod_{i=1}^d [a_i,b_i]$ (with $a_i < b_i$),
%% $B>0$, and ${\ve \Gamma} = (\Gamma_1, \ldots, \Gamma_d) \in \RR^d$ is a vector of positive reals,
then \cite{Bronshtein:1976wy} and \cite{MR876079} (chapter 8) show  that  if $0 < \epsilon \le \epsilon_0$ (for some $\epsilon_0 > 0$) then
\begin{equation}
  \label{eq:6}
  %% \log \Nc[ \epsilon, \cC[ \prod_{i=1}^d [a_i,b_i], B, {\ve \Gamma}], L_\infty]
  \log N \bigg(
  \epsilon,
  {\mc C} \bigg( \prod_{i=1}^d [a_i,b_i],   B, {\ve \Gamma} \bigg),
  L_\infty \bigg)
  \le C \epsilon^{-d/2}
\end{equation}
for a constant $C \equiv C_{D,B, \ve \Gamma}$.  Here, $\Nc[ \epsilon, {\cal F}, \rho]$ is the $\epsilon$-covering number of $\cal F$ in the metric $\rho$, which is defined to be the smallest number of balls of $\rho$-radius $\epsilon$ that cover ${\cal F}$, and $\log \Nc[ \epsilon, {\cal F}, \rho]$ is the corresponding metric entropy of $\cal F$, discussed in the first paragraph of this paper.
% The asymptotic behavior as $\epsilon \searrow 0$ of $\log \Nc[ \epsilon,
% \cC[D,B, \ve \Gamma], L_\infty]$
% is the same as that of the class
% Since convex
% functions are Lebesgue almost everywhere two-times differentiable
% %% aleksandrov
% their entropies correspond to the entropies for compact subsets of twice
% differentiable function classes, namely $\epsilon^{-d/2}$.

One would like to use \eqref{eq:6} in the study of  asymptotic properties of the statistical estimators discussed above.
% Bracketing entropies govern the suprema of corresponding empirical processes
% and thus govern the global rates of convergence of certain statistical
% estimators.
Unfortunately, the function classes that arise in those problems generally do not include Lipschitz constraints, and so the class $\cC[D, B, {\ve \Gamma}]$ is not of immediate use.  Furthermore, it turns out that without Lipschitz constraints, the $L_\infty$ covering or bracketing numbers are not bounded.  Thus, instead of of using the $L_\infty$ distance, we may consider using the $L_p$ distances, $1 \le p < \infty$.  Let $\cC\lp D, B \rp$ be the class of convex functions on $D$ with uniform bound $B$ (and no Lipschitz constraints).\newnot{symbol:convex-fns-B-Gamma}\newnot{symbol:convex-fns-B} Then \cite{MR2519658} %%dryanov
and \cite{Guntuboyina:2013jb} %
found bounds when $d=1$ and $d > 1$, respectively, for metric entropies of $ \cC[D, B]$: they showed that $\log \Nc[ \epsilon, \cC[D, B], L_p] \lesssim \epsilon^{-d/2}$, again with $D$ a hyperrectangle and $1 \le p < \infty$.  Here $\lesssim$ means $\le$ up to a constant which does not depend on $\epsilon$ (but does depend on $D$, $B$, and $p$).  The $d=1$ case (from \cite{MR2519658}) %%dryanov \cite{Doss-Wellner:2013}.
was the fundamental building block in computing global rates of convergence of the univariate log-concave and $s$-concave MLEs in \cite{Doss-Wellner:2013}.
In the corresponding statistical problems when $d > 1$, the domain of the
functions under consideration is not restricted to be a hyperrectangle but
rather may be an arbitrary convex set $D$.  Thus the results of
\cite{Guntuboyina:2013jb} %
are not immediately applicable, and there is need for results on more
general convex domains $D$ with a more complicated boundary and no Lipschitz constraints.
% However, in many problems the function domains may not be a hyperrectangle,
% so there is need to know the bracketing entropies for more general convex
% domains, $D$.  The approach of \cite{Guntuboyina:2013jb} does not immediately
% generalize to studying other convex domains $D$, since their results rely on
% repeated application of the bound of \cite{Bronshtein:1976wy}, after
% partitioning the domain in such a way that certain Lipschitz constraints
% hold.  This construction relies largely on the nature of the boundary of the
% hyperrectangle, which is much simpler than the boundary of more arbitrary
% convex sets $D$.

In this paper we are indeed able to generalize the results of
\cite{Guntuboyina:2013jb} considerably by finding bracketing entropy upper bounds
% to generalize the arguments of
% \cite{Guntuboyina:2013jb} to
for all (convex) polytopes $D$, attaining the bound
\begin{equation}
  \label{eq:1}
  \log \Nb[\epsilon, \cC[D, B ], L_p] \lesssim \epsilon^{-d/2}
\end{equation}
with $1 \le p < \infty$, $D$ a polytope,
%% satisfying Assumption~\ref{assm:simple-polytope},
and $0 < B < \infty$; this result is given in Theorem~\ref{thm:main-thm}.  Note that we work with bracketing entropy rather than
metric entropy.  Bracketing entropies are larger than metric entropies for
the $L_p$  metrics,
\begin{equation}
  \label{eq:81}
  N(\epsilon, \mc F, L_p) \le \Nb( 2\epsilon, \mc F, L_p),
  \quad \text{for } 1 \le p \le \infty,
  \text{ and }
  N(\epsilon, \mc F, L_\infty) = \Nb( 2\epsilon, \mc F, L_\infty),
\end{equation}
\citep[p. 84]{vanderVaart:1996tf}, %% here here here
so our bracketing entropy bounds imply metric entropy bounds of the same
order. Along the way, we also generalize the results of
\cite{Bronshtein:1976wy} to bound the $L_\infty$ bracketing numbers of
$\cC[D, B, {\ve \Gamma}]$ when $D$ is arbitrary.  One of the benefits of our
method is its constructive nature.  We initially study only simple polytopes
(defined in Subsection~\ref{subsec:defns-assms}) and in that case
we pay careful attention to
how the constants depend on $D$.

In Section~\ref{sec:further-applications}, we consider two further applications of our methods and ideas.
% In Subsection~\ref{sec:m-monotone}, we show that the method of proof used in
% Theorem~\ref{thm:main-thm} can be used in other contexts:
% our Theorem~\ref{thm:main-extension-m-mono} gives a bracketing upper bound
% for any class satisfying a condition,  \eqref{eq:conj-k-mono-multivariate} (given below). 
In Subsection~\ref{sec:m-monotone}
we define a new class of functions, the so-called multivariate $m$-monotone functions.  In the univariate setting $m$-monotone functions have been studied mathematically (\cite{Williamson:1955uy,Williamson:2010wy}, and references therein) and statistically \citep{Balabdaoui:2007jj,Balabdaoui:2010do,Gao:2009hf}, but to the best of our knowledge there has been no consideration or even definition of $m$-monotone functions in the multivariate case.  We define a class and show that our proof for the bracketing upper bound for convex functions applies to the case of $m$-monotone functions.  This is given in Theorem~\ref{thm:main-extension-m-mono}.

% \commentC{The result relies on Theorem~\ref{conj:m-mono-sup-bracket} which assumes that Property~\ref{property:bracketing-FTOC} holds; unfortunately, we have not proved that Property~\ref{property:bracketing-FTOC} holds, and thus the conclusion \eqref{eq:66} of Theorem~\ref{thm:main-extension-m-mono} is not entirely proven yet.  Nonetheless the proof of Theorem~\ref{thm:main-extension-m-mono} is sound in that it applies to any class for which the conclusion of Theorem~\ref{conj:m-mono-sup-bracket} holds.  Also, in Theorem~\ref{thm:k-mono-entropy}, we show that the conclusion of Theorem~\ref{conj:m-mono-sup-bracket}  indeed does hold in the univariate case.  (Theorem~\ref{conj:m-mono-sup-bracket} is a multivariate bracketing upper bound.  Thus
% Theorem~\ref{thm:k-mono-entropy} provides some evidence that Theorem~\ref{conj:m-mono-sup-bracket}  holds, and is also of interest in its own right.) }

In Subsection~\ref{sec:entr-level-set} we consider level set estimation (where the $\lambda$-level set of a function $f$
%% at value $-\infty < \lambda < \infty$
is %% the set
$\lb x : f(x) = \lambda \rb$).
Nonparametric level set estimation has gained increasing attention in recent
years, since it can capture very complex dependencies in a distribution or dataset.  In
Bayesian analysis, the level set of the posterior distribution is commonly
used to form a credible set, and this level set often has to be estimated based on samples generated from the Markov chain Monte Carlo method.
% Since posterior densities in most realistic
% problems are not available in closed form, the level set (or ``highest
% posterior density region'') often has to be estimated, generally based on samples generated from the Markov chain
% Monte Carlo method.
There are a large number of other settings where level
set estimation arises; see, for instance, the introduction of
\cite{Doss:2018kv}.  Here, we consider convex level set estimation.  For a
recent review of convexity-based methods in set estimation, see
\cite{Brunel:2018dl}.
In Subsection~\ref{sec:entr-level-set}, 
we present upper bounds for the so-called local entropy of level sets of convex functions.
These upper bounds  are an important step in proving that
fast rates of convergence may be achievable when one is estimating a
polytopal level set of a convex function.

During the course of the development of this paper, we became aware of the related work  \cite{Gao-Wellner-2017-arxiv-v3}, which was developed simultaneously and separately  from our paper.
In \cite{Gao-Wellner-2017-arxiv-v3}, the authors demonstrate in their Theorem~1.6 that if $D$ is a sphere then
\eqref{eq:1} fails when $p(d-1) > d$.
This shows that if $D$ is not a polytope the situation may be more complicated than when $D$ is a polytope.
%% \cite{Gao-Wellner-2017-arxiv-v3}
They
also find upper bounds of order $\epsilon^{-d/2}$ when $D$ is a polytope. Their methods are quite different than ours and in particular they do not explicitly construct their bracketing set but rather rely on an algebraic relation (see their function $g(\cdot,\cdot)$ in their Section 2.5); our method on the other hand is explicitly constructive.
% in that we construct a cover of $D$ involving parallelotopes near the boundary, and then apply
% Theorem~\ref{thm:GS-extension} below to the elements of the cover.
%% We provide some discussion of how our constants can differ from those of \cite{Gao-Wellner-2017-arxiv-v3} in Example~\ref{ex:hyperrectangle}. %%
Our constants differ from those of \cite{Gao-Wellner-2017-arxiv-v3}.  Our constants
depend on the volume (measured in the appropriate dimension) of the faces of the polytope $D$, which is perhaps  an interesting phenomenon
(and is (distantly) reminiscent of the Minkowski-Steiner formula \citep{Federer:1969uh}).
Besides the fact that our constants differ from those of \cite{Gao-Wellner-2017-arxiv-v3} and reflect the geometry of $D$,
the constructive nature of our approach enables consideration of other problems, not considered by  \cite{Gao-Wellner-2017-arxiv-v3}, which we do in Section~\ref{sec:further-applications} (as described above).

This paper is organized as follows.  %%
% In Section~\ref{sec:notation} we provide definitions and notation for the
% paper. %%
In Section~\ref{sec:Lipschitz-bracketing} we prove bounds for bracketing
entropy of classes of convex functions with Lipschitz bounds, using the
$L_\infty$ metric. %%
We use these to prove our main result,
Theorem~\ref{thm:main-thm}, for the bracketing entropy of classes
of convex functions without Lipschitz bounds in the $L_p$ metrics, $1 \le p <
\infty$, which we do in Section~\ref{sec:bracketing-uniform-nolipschitz}. %%
We defer some of the details of the proofs to
Section~\ref{sec:technical-details}. %
In Section~\ref{sec:further-applications} we study two more problems.  In Subsection~\ref{sec:m-monotone} we consider bracketing numbers related to univariate and multivariate $m$-monotone function classes.  In Subsection~\ref{sec:entr-level-set} we consider local entropies related to level set estimation.
%% Section~\ref{sec:notation} provides an index of notation. %\
There is a notation index at the end of the document.

\begin{mynotes}
  \begin{notes}

    %% \section{Notation}
    %% \label{sec:notation}

    $D$, convex set\\
    $E$, half space \\
    $F$, facet \\
    $G$ , face \\
    $H$ hyperplane \\
    $I$ interior ? \\
    J
    K
    L
    $N$ number facets or  hyperplanes

    All our vector spaces are finite dimensional.

    We use the symbol ``$:=$'' or ``$=:$'' to assign a definition to a variable
    (on the side of the colon), and the symbol ``$\equiv$'' to indicate when we
    suppress dependence on subscripts for ease of notation, as in ${\ve e} \equiv
    {\ve e}_{\ib,\jb}$.

  \end{notes}
\end{mynotes}

\section{Bracketing with Lipschitz Constraints}
\label{sec:Lipschitz-bracketing}

If we have sets $D_i \subset \RR^d$, $i = 1, \ldots, M$, for $M \in \NN$, and
$D \subseteq \cup_{i=1}^M D_i$ then for $\epsilon_i > 0$, $0 < p < \infty$, and any class of functions $\mc F$, 
\begin{equation}
  \label{eq:bracketing-cover}
  %% \Nb[ \lp \sum^M_{i=1} \epsilon_i^p \rp^{1/p}, \cC[D,1],  L_p ]
  N_{[\,]} \bigg(  \bigg( \sum^M_{i=1} \epsilon_i^p  \bigg)^{1/p}, \mc F,
  %% \cC[D,1],
  L_p \bigg)
  \le \prod_{i=1}^M \Nb[ \epsilon_i, \mc{F}\vert_{D_i}, L_p ],
\end{equation}
where, for  %% a class of functions $\cal F$ and
a set $G$, we let $\cal F|_G$
denote the class $\lb f|_G : f \in {\cal F} \rb$ where $f|_G$ is the
restriction of $f$ to the set $G$. 
We will apply \eqref{eq:bracketing-cover}
to a cover of $D$ by sets $G$ with the property that
$\cC[D,1]|_G \subseteq \cC[G, 1, \ve \Gamma]$
% \begin{equation*}
%   \cC[D,1]|_G \subseteq \cC[G, 1, \ve \Gamma]
% \end{equation*}
for some bounded vector $\ve \Gamma$, so that we can apply bracketing results for classes of convex functions with Lipschitz bounds.  Thus, in this section, we develop the needed bracketing results for such Lipschitz classes, for arbitrary (bounded) convex domains $D$. \newnot{symbol:D-general}  Recall the definition of $\cC[D, B, \Gb, {\ve v} ]$ and $\cC[D, B, \Gb ]$ from \eqref{eq:10}.  When we have Lipschitz constraints on convex functions, we will see that the situation for forming brackets for $\cC[D, 1, \Gb ]$ with $D \subseteq [0,1]^d$ is essentially the same as for forming brackets for $\cC[[0,1]^d, 1, \Gb ]$.
% For $y,z \in \RR^d$ let $\la y,z \ra := \sum_{i=1}^d y_iz_i$,
% let $\Vert z \Vert^2 := \la z,z \ra$,\newnot{symbol:Euclidean-norm} and for two sets $C, D \subset \RR^d$, define the Hausdorff distance between them by
% \begin{equation*}
%   l_H(C,D) := \max \lp \sup_{x \in D} \inf_{y \in C}  \Vert x - y \Vert,
%   \sup_{y \in C} \inf_{x \in D}  \Vert x - y \Vert \rp.
% \end{equation*}
% Let $B_d(z,R) \equiv B(z, R) := \lb x \in \RR^d : \Vert x-z \Vert \le R \rb$.
% , and let
% \begin{equation*}
%   {\cal K}^{d}(R)
%   = \lb D: D \text{ is a closed, convex, nonempty set}, D \subseteq B_{d}(0,R) \rb.
% \end{equation*}
\begin{mynotes}
  \begin{notes}
    I added ``nonempty'' to definition of ${\cal K}^d$; \cite{Bronshtein:1976wy} and I think others do not include this, although it seems quite necessary
  \end{notes}
\end{mynotes}
Theorem~3.2 from
\citep{Guntuboyina:2013jb} gives the result of the below Theorem~\ref{thm:GS-extension}
when $D = \prod_{i=1}^d [a_i,b_i]$; we now extend it in Theorem~\ref{thm:GS-extension} to the case of a general
$D$.  When we consider convex functions without Lipschitz constraints, we
will partition $D$ into sets that are contained in parallelotopes and apply
Theorem~\ref{thm:GS-extension} to those sets.

\begin{theorem}
  \label{thm:GS-extension}
  Let $a_i < b_i$ and let $D \subset \prod_{i=1}^d [a_i,b_i]$ be a convex
  set. Let $\Gb = (\Gamma_1, \ldots, \Gamma_d)$ and
  $0< B, \Gamma_1, \ldots, \Gamma_d < \infty$.  Then there exists
  a positive constant $c\equiv c_d$
  %% and $\epsilon_0 \equiv \epsilon_{0,d}$
  such that
  \begin{align}
    \log \Nb[\epsilon \vol_d (D)^{1/p}, \cC[D, B, \Gb], L_p ]
    & \le \log \Nb[\epsilon, \cC[D, B, \Gb], L_\infty ] \label{eq:Lip-bound-1}\\
    & \le c
      \epsilon^{-d/2}
      \bigg( B + \sum_{i=1}^d \Gamma_i (b_i-a_i) \bigg)^{d/2} \label{eq:Lip-bound-2}
      %% \\ & \le c d \lp \frac{B + \max_{i \in \lb 1,\ldots, d\rb} \Gamma_i
      %% (b_i-a_i)}{\epsilon} \rp^{d/2}
  \end{align}
  % \begin{equation*}
  %   \begin{split}
  %     \log \Nb[\epsilon \vol_d (D)^{1/p}, \cC[D, B, \Gb], L_p ]
  %     & \le \log \Nb[\epsilon, \cC[D, B, \Gb], L_\infty ] \\
  %     & \le c \lp \frac{B + \sum_{i=1}^d \Gamma_i (b_i-a_i)}{\epsilon} \rp^{d/2}
  % %%     \\      & \le c d \lp \frac{B + \max_{i \in \lb 1,\ldots, d\rb} \Gamma_i (b_i-a_i)}{\epsilon} \rp^{d/2}    \end{split}
  % \end{equation*}
  for $\epsilon > 0$ and
  $p \ge 1$.
\end{theorem}

\noindent  Here, $\vol_d(D)$ is $d$-dimensional volume (Lebesgue measure) of the set $D$.  
The proof is given in  \cite{DBLP:journals/corr/Doss15}; we leave it out here due to space limitations.

\section{Bracketing without Lipschitz Constraints}
\label{sec:bracketing-uniform-nolipschitz}

In the previous section we bounded bracketing entropy for classes of
functions with Lipschitz constraints.  In this section we remove those
Lipschitz constraints.  With Lipschitz constraints we could consider arbitrary domains $D$, but without the Lipschitz constraints we need more restrictions: now we will take $D$ to be a {\em simple polytope} (defined below).
We now define notation and assumptions we will use for the remainder of the document.

\subsection{Notation and Terminology}

For $y,z \in \RR^d$ let $\la y,z \ra := \sum_{i=1}^d y_iz_i$,
let $\Vert z \Vert^2 := \la z,z \ra$,\newnot{symbol:Euclidean-norm} and for two sets $C, D \subset \RR^d$, define the Hausdorff distance between them by
\begin{equation*}
  l_H(C,D) := \max \lp \sup_{x \in D} \inf_{y \in C}  \Vert x - y \Vert,
  \sup_{y \in C} \inf_{x \in D}  \Vert x - y \Vert \rp.
\end{equation*}
Let $B_d(z,R) \equiv B(z, R) := \lb x \in \RR^d : \Vert x-z \Vert \le R \rb$.

We will consider only the case $d \ge 2$ since the result  when $d=1$ is given in \cite{MR2519658}. %%dryanov \cite{Doss-Wellner:2013}.
Recall that for a convex set $G$, a set $F \subset G$ is a {\em face} of $G$ if $F$ is either $\emptyset$ (the empty set), $G$, or if $F = G \cap H$ for some supporting hyperplane $H$ \citep{Rockafellar:1997ww} %%
of $G$.
A set $F \subset G$ is a {\em facet} of $G$ if $F$ is a $(d-1)$-dimensional face
(see e.g., \cite{Grunbaum:1967vq}). %% page 17 for face; 26, 31 for facet
We will focus on {\em simple} polytopes first
(see Assumption~\ref{assm:simple-polytope}).
A simple polytope is one in which all $(d-k)$-dimensional faces (abbreviated ``$(d-k)$-faces'') of $D$ have exactly $k$ incident facets for $k \in \lb 0, \ldots, d \rb$.
The simple polytopes are dense in the class of all polytopes in the
Hausdorff distance (page 82 of \cite{Grunbaum:1967vq}).
Any convex polytope
can be triangulated into
%% $O(n^{\lfloor d/2 \rfloor})$
$O(n^{\lceil d/2 \rceil})$  %% CEIL OR FLOOR
simplices (which are
simple polytopes) if the polytope has $n$ vertices (see e.g.\
\cite{Dey:1998cu}), and so we can translate our theorem into a  result for a general polytope
$D$; see Corollary~\ref{cor:general-convex-polytope}.
\begin{mynotes}
  \begin{notes}
    This definition of simple is dual to the definition of 'general', given
    as simplicial polytopes.  A simplicial polytope has all vertices in
    general position, whereas this definition has all hyperplanes in general
    position.  See pages 57 and 58 (and perhaps page 310) of Grunbaum (1967).
  \end{notes}
\end{mynotes}
For two sets $A$ and $B$ %% (one of which may be a singleton)
%% as in \eqref{eq:FritzJohn}), %% to 'define' the lack of brackets around $x_\jb$
let $A + B := \lb a + b : a\in A ,b \in B \rb$.  For a vector $v \in \RR^d$,  we
let $[0,v]:= \lb
\lambda v : \lambda \in [0,1] \rb$.
For a set $G$, let %
$d^+(x, G, e ) := \inf \lb K \ge 0: (x +K e) \cap G \ne \emptyset \rb$ %
(which may in general be infinite). \newnot{symbol:d+}
For a point $x$, a set $H$, and a unit vector $v$, let \newnot{symbol:d(x,H,v)}
\begin{equation*}
  \label{eq:defn:distance-set-direction}
  d(x,H,v) := \inf
  \lb |k| : x + kv \in H \rb
  = \min\lp d^+(x,H, v), d^+(x,H,-v) \rp
\end{equation*}
be the distance from $x$ to $H$ along the vector  $v$, and for a set $E$,
let $d(E,H,v) := \inf_{x \in E} d(x,H,v)$.\newnot{symbol:d(E,H,v)}
\begin{myold}
  \begin{old}
    For $\ib = (i_1, \ldots, i_k) \in
    I_k,$ and $\jb = (j_1,\ldots, j_k) \in J_k$ let
    \begin{equation}
      \label{eq:defn:Gi}
      G_{\ib,\jb}
      %% G_{\ib}
      := \lb x \in D : \delta_{i_\alpha} \le d(x, H_{j_\alpha}, v_{j_\alpha}) <
      \delta_{i_\alpha + 1} \mbox{ for } \alpha = 1, \ldots, N \rb,
    \end{equation}
    where for $\alpha > k $ we let ${i_\alpha } = {A+1}$.
  \end{old}
\end{myold}
We let $\partial G$ be the boundary of $G$ in $\RR^d$ and we let $\relbd G$ be the {\em relative boundary} of $G$, the set difference between the closure of $G$ and the relative interior of $G$ (e.g., page 44 of \cite{Rockafellar:1997ww}). \newnot{symbol:relbd}  %
Let
$\vol_{d-k}(G)$ be the $(d-k)$-dimensional volume of $G$ (and, in
particular, $\vol_0(G)$ is the number of elements in $G$).\footnote{In general, $\vol_{d-k}$ can be defined rigorously using the so-called $(d-k)$-dimensional Hausdorff measure. We will only need the $(d-k)$-dimensional volume of polytopes contained in affine spaces, and in such cases the definition is straightforward (and only requires  Lebesgue measure).} %
For $a,b \in \RR$, we let $a \vee b $ be the maximum of $a$ and $b$, and $a \wedge b$ be the minimum of $a$ and $b$.
For two vectors $e,v \in \RR^d$ and a linear subspace $V$ of $\RR^d$, we write $e \perp v$ if $\la e, v \ra = 0$,
we write $e \perp V$ if $e \perp v$ for all $v \in V$,
and we let $V^\perp$ be the orthogonal complement linear subspace of $V$ in $\RR^d$.

\subsection{Definitions and Assumptions}
\label{subsec:defns-assms}
\bigskip

In what follows, we will  assume that $D$ is a polytope, meaning that for some $N \in \NN$,
\newnot{symbol:D-polytope}  \newnot{symbol:N}
%% What is $N$ ; here here here
$D = \cap_{j=1}^N E_j$ where $E_j:= \lb x \in \RR^d : \la v_j, x\ra
\ge p_j \rb$ are halfspaces \newnot{symbol:Ei} with inner normal
unit % added 2017 ... ?
vectors $v_j$ such that $v_i \ne  v_j$  if $i \ne j$, \newnot{symbol:vi} %
and where $p_j
\in \RR$, for $j=1,\ldots, N$.
\begin{mynotes}
  \begin{notes}
    Note: We allow $v_j = - v_i$ if $i \ne j$, as in a cube.
  \end{notes}
\end{mynotes}
Let $H_j := \lb x \in \RR^d : \la x, v_j \ra
= p_j \rb$ be the corresponding hyperplanes \newnot{symbol:Hi}
and let $F_j := H_j \cap D$ be the corresponding facets of $D$. \newnot{symbol:Fi}
For $k \in \lb 0,\ldots ,d \rb$, we will define $J_k$ to index the $(d-k)$-faces of $D$.
% Let $J_0 = \lb 1 \rb$,
% and for $k \in \lb 1, \ldots, d \rb$,
% let %
% $J_k := \lb (j_1,\ldots, j_k) \in \lb 1, \ldots, N\rb^k
% : j_1 < \cdots < j_k \rb$. \newnot{symbol:Jk}    %
% For $\jb \in J_k$, let
% % \begin{equation*}
% %   G_{\jb} = \cap_{\alpha=1}^k H_{j_\alpha} \cap D.
% % \end{equation*}
% \begin{equation*}
%   G_{\jb} = \cap_{\alpha=1}^k H_{j_\alpha} \cap D \,
%   \mbox{ if } k \ne 0,
%   \qquad
%   \mbox{ and let }
%   \qquad
%   %%   G_{1} = D
%   G_{\jb} = D \,
%   \mbox{ if } k = 0.
% \end{equation*}
% The face $G_{\jb}$, $\jb \in J_k$, is $(d-k)$-dimensional if it is not empty, and $H_{j_1}\cap D, \ldots, H_{j_k} \cap D$ are the only facets of $D$ containing $G_{\jb}$,
% by Theorem 12.14 of \cite{Brondsted:1983vt}. \newnot{symbol:Gi}
\begin{mynotes}
  \begin{notes}
    Use of $J_k$ may need to be replaced by $\tilde{J}_k$ in places.  see e.g.\ page~\pageref{tmp:loc:Jk}, ``ignore terms with ...''.
  \end{notes}
\end{mynotes}
First let $\tilde{J}_k :=
\lb (j_1,\ldots, j_k) \in \lb 1, \ldots, N\rb^k
: j_1 < \cdots < j_k \rb$,
and for $\jb \in \tilde{J}_k$,
let
\begin{equation*}
  G_{\jb} = \cap_{\alpha=1}^k H_{j_\alpha} \cap D \,
  \mbox{ if } k \ne 0,
  \qquad
  \mbox{ and let }
  \qquad
  G_{\jb} = D \,
  \mbox{ if } k = 0.
\end{equation*}
Now let $J_0 = \lb 1 \rb$,
and for $k \in \lb 1, \ldots, d \rb$,
let %
$J_k := \lb \jb \in \tilde{J}_k : G_{\jb} \ne \emptyset \rb $.
% be the subset of
% $\tilde{J}_k$
% such that $G_{\jb}$ is not an empty set for all $\jb \in  J_k$.
\newnot{symbol:Jk}    %
The face $G_{\jb}$, $\jb \in J_k$, is $(d-k)$-dimensional
% if it is not empty,
and $H_{j_1}\cap D, \ldots, H_{j_k} \cap D$ are the only facets of $D$ containing $G_{\jb}$,
by Theorem 12.14 of \cite{Brondsted:1983vt}. \newnot{symbol:Gi}
\begin{mynotes}
  \begin{notes}
    ((Unneeded:
    (and each $(d-k)$-face of a $d$-polytope is equal to the intersection of the facets containing it by 3.1.7, page 35, \cite{Grunbaum:1967vq}). ))
  \end{notes}
\end{mynotes}
Thus, by
%% Fritz
John's theorem,
Theorem~\ref{thm:FritzJohn} (\cite{John:1948vs}, see also \cite{Ball:1992bp} or \cite{Ball:1997ud}), there exists $x_{\jb } \in G_{\jb}$ such that $G_{\jb} - x_{\jb}$ contains a $(d-k)$-dimensional ellipsoid $A_{\jb} - x_{\jb}$ of maximal $(d-k)$-dimensional volume and such that \newnot{symbol:Ajxj}
\begin{equation}
  \label{eq:FritzJohn}
  A_{\jb} -x_{\jb} \subset G_{\jb} -x_{\jb} \subset d (A_{\jb}-x_{\jb}).
\end{equation}
Let $\gamma_{\jb, \alpha} / 2 := d^+(x_{\jb}, \relbd A_{\jb}, e_\alpha)$ \label{defn:gamma-j} be the radius of $A_{\jb}$ in the direction $e_{\jb, \alpha}$, where
% \begin{equation}
%   \label{eq:3}
%   e_{\jb, k+1},\ldots, e_{\jb, d}
%   \text{ are the orthonormal unit vectors given by the axes of the ellipsoid }
%   A_{\jb}-x_{\jb}.
% \end{equation}
$ e_{\jb, k+1},\ldots, e_{\jb, d} $\label{loc:defn:e-alpha}\newnot{symbol:e-axes-def1}
are the orthonormal unit vectors given by the axes of the ellipsoid $A_{\jb}-x_{\jb}$.  Let $E_{\jb} := \SPAN \lb e_{\jb, k+1}, \ldots, e_{\jb,d} \rb$ \label{loc:defn:Ej} be the linear space containing $G_{\jb} - x_{\jb}$.
\newnot{symbol:gammaj}
\begin{myold}
  \begin{old}
    The vectors $u_{\jb,\alpha}$ were previously called $e_{\jb, \alpha}$ but
    these conflicted, I believe, with my later basis vectors $e_{\jb}$.
  \end{old}
\end{myold}
% , meaning that $x_{\jb} \pm \gamma_{\jb,\alpha}
% e_\alpha / 2$ lies in the boundary of $A_{\jb}$.
Let
$A$ be an integer
and $u$ a positive real number, and let \newnot{symbol:deltai-u}
\begin{equation}
  \label{eq:5}
  0 = \delta_0 < \delta_1 < \cdots < \delta_A < u <  \delta_{A+1} < \delta_{A+2}=\infty
\end{equation}
be a sequence.  This sequence as well as $A$ and $u$ will be specified in
greater detail later.  For $k \in \lb 1, \ldots, d\rb$, let
$I_k := \lb 0, \ldots, A \rb^k$, and let $I_0 := \lb A \rb$.
\newnot{symbol:Ik}
%%%% Turns out we don't need \linaff ?
% Let $\linaff P $ be the translated affine span of $P$,
% i.e.\ the space of all linear combinations of elements of $ (P-x) $, for any
% $x \in P$.
% \footnote{Note that
% $\operatorname{lin} P$ is commonly used to refer to the linear span of $P$
% rather than of $P-x$, and thus to distinguish from this case, we use the
% notation ``$\Lin$'' rather than ``$\operatorname{lin}$.''}
For $k \in \lb 1, \ldots, d \rb$, $\ib = (i_1, \ldots, i_k) \in
I_k,$ and $\jb = (j_1,\ldots, j_k) \in J_k$ let
\begin{equation}
  \label{eq:defn:Gi}
  G_{\ib,\jb}
  %% G_{\ib}
  := \lb x \in D : \delta_{i_\alpha} \le d(x, H_{j_\alpha}, v_{j_\alpha})
  \le %% <
  \delta_{i_\alpha + 1} \mbox{ for } \alpha = 1, \ldots, N \rb
\end{equation}
where in the previous display for $ \alpha > k$ we let $i_\alpha = A+1$ and $j_\alpha$ take on the values in $\lb 1, \ldots, N \rb \setminus \lb j_1, \ldots, j_k \rb$ (in any order). \newnot{symbol:Gij}  For the $k=0$ case, let $G_{A,1} := \lb x \in D : d(x, \partial D ) \ge u \rb.$ These sets are not parallelotopes, since for $\alpha > k$, $\delta_{i_\alpha+1} = \infty$.  However, for any $x \in G_{\jb}$, $(G_{\ib,\jb} -x) \cap \SPAN \lb v_{j_1}, \ldots, v_{j_\beta} \rb$, for $\beta \le k$, is contained in a $\beta$-dimensional parallelotope by construction; this will be used to understand the volume of $G_{\ib,\jb}$.
We will eventually define $u$ such that
$D \subset \bigcup_{k=0}^d \bigcup_{\jb \in J_k, \ib \in I_k} G_{\ib, \jb}$
(see Lemma~\ref{lem:D-subset-cup-Gij}).

The setup for our first main results is summarized in the following assumption.
\begin{assumption}
  \label{assm:simple-polytope}
  Let $d \ge 2$, let the definitions of the above
  Subsection~\ref{subsec:defns-assms} hold, and let $D \subset \RR^{d}$ be a
  simple convex polytope.
\end{assumption}

Additionally, define the support function for a convex set $D$ to be, for
$x \in \RR^d$ with $\|x\|=1$,
% \begin{equation*}
$  h(D, x) := \max_{d \in D} \la d,x \ra.$
% \end{equation*}
Then the width function is, for $\Vert u \Vert = 1$, \newnot{symbol:h} \newnot{symbol:w}
% \begin{equation*}
$w(D, u) := h(D, u) + h(D,-u),$
% \end{equation*}
which gives the distance between supporting hyperplanes of $D$ with inner
normal vectors $u$ and $-u$, respectively, and let
% \begin{equation*}
$ \diam(D) := \sup_{ \Vert u \Vert = 1} w(D,u)$
% \end{equation*}
be the diameter of $D$.

\subsection{Main Results}
We want to bound the slope of functions $f \in \cC[D,1]|_{G_{\ib,\jb}}$, so that we can apply bracketing bounds on convex function classes with Lipschitz bounds.  Note that each $G_{\ib,\jb}$ is distance $\delta_{i_\alpha}$ in the direction of $v_{j_\alpha}$ from $H_{j_\alpha}$, which means that if $f \in \cC[D,1]|_{G_{\ib,\jb}}$ then $f$ has Lipschitz constant bounded by $2 / \delta_{i_\alpha}$ along the direction $v_{j_\alpha}$.  However, the vectors $v_{j_\alpha}$ are not orthonormal, so the distance from $G_{\ib,\jb}$ along $v_{j_\alpha}$ to a hyperplane other than $H_{j_\alpha}$ may be smaller than $\delta_{i_\alpha}$.

Note that if $P \subset R \subset \RR^d$ where $R$ is a hyperrectangle and $P$ is a parallelotope defined by vectors $v_1,\ldots, v_d$, then if $A$ is a linear map with $v_1,\ldots, v_d$ as its eigenvectors (thus rescaling $P$), then $AR$ will not necessarily still be a hyperrectangle, i.e.\ its axes may no longer be orthogonal. Thus, we cannot argue by simple scaling arguments that bracketing numbers for $P$ scale with the lengths along the vectors $v_i$.

For each $G_{\ib,\jb}$ we will find an orthonormal basis such that
$G_{\ib,\jb}$ is contained in a rectangle $R$ whose axes are given by the
basis and whose lengths along those axes (i.e., widths) are bounded by a
constant times the width of one of the normal vectors $v_{j_\alpha}$.
Furthermore, the distance from $R$ along each basis vector to $\partial D$
will be bounded by the distance from $G_{\ib,\jb}$ along $v_{j_\alpha}$ to
$H_{j_\alpha}$. This will give us control of both the Lipschitz parameters
and the widths corresponding to the basis, and thus control of the bracketing number for classes of convex functions. We rely on the following basic lemma.
\begin{lemma}
  \label{lem:basic-lemma_conv-lipschitz}
  If 
  $ f \in \cC(D, B),$
  %% \text{ for }
  $B > 0,$
  and $x \in D$  is such that
  $d(x, \bd D, e_\alpha) \ge \delta > 0$
  then
  \begin{equation}
    \label{eq:79}
    \lv \pderiv{x_i}f(x)\rv \le \frac{2B}{\delta}
  \end{equation}
  where the derivative stands for both the right and left derivative of $f$.
\end{lemma}
\begin{proof}
  % $f \in \cC[ D, B]|_{G_{\ib,\jb}}$,  then for any $x \in G_{\ib,\jb}$,
  Let
  $z_1 = x - \gamma_1 e_\alpha $ and $z_2 = x + \gamma_2 e_\alpha$,
  $\gamma_1,\gamma_2 > 0$, both be
  elements of $\partial D$, so that
  by convexity we have for any $h \in [-\gamma_1, \gamma_2]$,
  \begin{equation*}
    \label{eq:lipschitz-constraint}
    \frac{-2B}{\delta }
    \le \frac{f(z_1) - f( z_1 + \delta e_\alpha ) }{\delta }
    \le  \frac{ f(x+ h e_\alpha)-f(x)}{h}
    \le \frac{f(z_2 ) - f( z_2 - \delta e_\alpha) }{\delta  }
    \le \frac{2 B } { \delta }.
  \end{equation*}
  Thus, $f$ satisfies a Lipschitz constraint in the direction of $e_\alpha$.
\end{proof}

The following proposition constructs a basis and gives control for the basis elements in $\SPAN \lb G_{\jb} \rb$.
For the basis elements perpendicular to $\SPAN \lb G_{\jb} \rb$,
control is given
by Lemma~\ref{lem:embed-parallelotope-hyperrectangle}
and Lemma~\ref{lem:Rij-Lipschitz-diameter-bound}
in Section~\ref{sec:technical-details}.

%% I have an earlier form of this proof where the notation is in terms of the
%% actual j_\beta indices, but it gets a bit messy (and the proof wasn't
%% quite completed in that notation )
\begin{proposition}
  \label{prop:basis-lipschitz} \newnot{symbol:e-basis}
  Let Assumption~\ref{assm:simple-polytope} hold for a convex polytope $D$.    For each $k \in \lb 0,\ldots, d \rb$, $\ib \in I_k, \jb \in J_k$, and each $G_{\ib,\jb}$, there is an orthornormal basis ${\ve e}_{\ib,\jb} \equiv {\ve e} := (e_1,\ldots, e_d) $ of $\RR^d$ such that
  % \begin{equation*}
  %   d(G_{\ib,\jb}, \partial D, e_\alpha) \ge \delta_{i_\alpha} \text{ for all
  % } \alpha \in \lb 1,\ldots, d \rb,
  % \end{equation*}
  % and thus,
  for any $f \in \cC[
  D,B]|_{G_{\ib,\jb}}$, $f$ has Lipschitz constant $2B/\delta_{i_\alpha}$ in
  the direction $e_{\alpha}$, where $\delta_{i_\alpha} = \delta_{A+1}$ if
  $k+1 \le \alpha \le d$.  Furthermore, there exists a permutation $\pi$ of $\lp 1,\ldots, k \rp$ such that for $\alpha=1,\ldots, k$,
  $e_{\ib,\jb,\alpha} \equiv e_{\alpha}$ satisfies
  \begin{equation}
    \label{eq:11}
    e_\alpha \in \SPAN  \lb v_{j_{\pi(1)}},\ldots, v_{j_{\pi(\alpha)}}\rb, \;
    e_\alpha \perp \SPAN \lb v_{j_{\pi(1)}}, \ldots, v_{j_{\pi(\alpha-1)}} \rb, \;
    \mbox{ and } \la e_\alpha, v_{j_{\pi(\alpha)}} \ra > 0,
    %% \Vert   e_\alpha \Vert = 1,
  \end{equation}
  % \begin{equation}
  %   \label{eq:11}
  %   e_\alpha \in \SPAN  \lb v_{j_1},\ldots, v_{j_\alpha}\rb, \;
  %   e_\alpha \perp \SPAN \lb v_{j_1}, \ldots, v_{j_{\alpha-1}} \rb, \;
  %   \mbox{ and } \la e_\alpha, v_{j_\alpha} \ra > 0,
  %   %%   \Vert   e_\alpha \Vert = 1,
  % \end{equation}
  and
  and for $\alpha \in \lb k+1,\ldots, d\rb$,
  % $e_\alpha \perp \SPAN \lb  v_{j_1}, \ldots, v_{j_k} \rb$.
  $e_\alpha \perp \SPAN \lb  v_{j_{\pi(1)}}, \ldots, v_{j_{\pi(k)}} \rb =: V$.
  In particular, we may take $e_{ k+1}, \ldots, e_{ d}$ to be the orthonormal unit axis vectors of $A_{\jb} - x_{\jb}$ as defined
  on page~\pageref{loc:defn:e-alpha}.
  Thus it is immediate that neither $V$ nor $V^\perp$ depend on $\ib$.
\end{proposition}
\begin{proof}
  Without loss of generality, for ease of notation we assume in this proof
  that
  % \begin{equation*}
  $j_\beta = \beta \text{ for } \beta=1,\ldots, k,$
  % \end{equation*}
  and then that
  %% those $k$ hyperplanes $H_1, \ldots, H_k$ and assume
  \begin{equation*}
    \delta_{i_1} \le \delta_{i_2} \le \cdots \le \delta_{i_k} \le
    \delta_{i_{k+1}} = \cdots  = \delta_{i_N},
  \end{equation*}
  where we let $i_\alpha = A+1$ for $k < \alpha \le N$.
  That is, we assume that
  $H_{1}, \ldots, H_{k}$ are the nearest hyperplanes to $G_{\ib,\jb}$, in
  order of increasing distance;  we then take $\pi$ to be the identity. To define the orthonormal basis vectors, we
  will use a Gram-Schmidt orthonormalization, proceeding
  according to increasing distances from $G_{\ib,\jb}$ to the
  hyperplanes $H_{j}$.  Define $e_1 := v_{1}$ and for $1 < j \le
  k$, define $e_j$ inductively by
  \begin{equation*}
    e_j \in \SPAN  \lb v_{1},\ldots, v_{j}\rb, \;
    e_j \perp \SPAN \lb v_{1}, \ldots, v_{j-1} \rb, \;
    \la e_j, v_{j} \ra > 0, \mbox{ and }
    \Vert   e_j \Vert = 1.
  \end{equation*}
  Let $e_{k+1}, \ldots, e_{d}$ be orthonormal unit vectors given by the axes of
  the ellipsoid $A_{\jb}-x_{\jb}$.  Note that these vectors form an
  orthonormal basis of $\SPAN \lb v_1, \ldots, v_k \rb^\perp$ because
  $\SPAN \lb e_{k+1}, \ldots, e_{d} \rb = \SPAN
  (G_{\jb}-x_{\jb}) $ is perpendicular
  to $\SPAN \lb v_1, \ldots, v_k \rb$ by
  definition.  % and % let $\lb e_j \rb_{j=k+1}^d $ be any orthonormal basis of
  % $\SPAN \lb v_{1}, % \ldots, v_{k} \rb^{\perp}$.
  % Recall, for a point $x$, a set $H$, and a unit vector $v$, that
  % % \begin{equation*}
  % %%   \label{eq:defn:distance-set-direction}
  %   $d(x,H,v) := \inf
  %   \lb |k| : x + kv \in H \rb$
  % % \end{equation*}
  % is the distance from $x$ to $H$ in direction $v$, and for a set $E$,
  % $d(E,H,v) := \inf_{x \in E} d(x,H,v)$.
  For $\alpha \in \lb 1, \ldots, k
  \rb$, for any $x \in G_{\ib,\jb}$, since $d(x,H_{\alpha}, v)$ achieves its minimum
  when $v$ is $v_{\alpha}$,
  \begin{align*}
    & d(x, H_{\alpha}, e_\alpha) \ge d(x, H_{\alpha}, v_{\alpha}) \ge \delta_{i_\alpha}, \\
    & d(x, H_{j}, e_\alpha) \ge d(x, H_{j}, v_{j}) \ge \delta_{i_j} \ge
      \delta_{i_\alpha},
      \mbox{ for all } N \ge j >  \alpha,    \mbox{   and }  \\
    & d(x,H_{j}, e_\alpha) = \infty > \delta_{i_\alpha} \mbox{ for } j < \alpha,
  \end{align*}
  since $e_\alpha \perp \SPAN \lb v_{1}, \ldots, v_{\alpha-1} \rb$.
  Similarly, for $\alpha \in \lb k+1, \ldots, d \rb$,
  \begin{align*}
    %% & d(x, H_\alpha, e_\alpha) \ge d(x, H_\alpha, v_\alpha) \ge \delta_{i_\alpha}, \\
       & d(x, H_j, e_\alpha) \ge d(x, H_j, v_j) \ge \delta_{A+1},
         \mbox{ for all } N \ge   j \ge k+1,    \mbox{   and }  \\
       & d(x,H_j, e_\alpha) = \infty > \delta_{A+1} \mbox{ for } j \le k,
  \end{align*}
  since $e_\alpha \perp \SPAN \lb v_1, \ldots, v_k \rb$.  Thus, we have
  % \begin{equation*}
  $d(G_{\ib, \jb}, H_j, e_\alpha) \ge \delta_{i_\alpha}$
  % \end{equation*}
  for $\alpha \in \lb 1, \ldots, d \rb$ and for $j \in \lb 1, \ldots , N
  \rb$.
  % where we let $\delta_{i_\alpha} = \delta_{A+1} $ for $\alpha \in \lb
  % k+1,\ldots, d\rb$.
  That is, we have shown
  \begin{equation}
    \label{eq:dist-G-to-boundary-D}
    d(G_{\ib, \jb}, \partial D, e_\alpha) \ge \delta_{i_\alpha} \mbox{ for all }
    \alpha \in \lb 1,  \ldots , d \rb.
  \end{equation}
  Thus by \eqref{eq:79}, $f$ has Lipschitz bound $2 B / \delta_{i_\alpha}$ in the direction $e_\alpha$.
  % Thus, if $f \in \cC[ D, B]|_{G_{\ib,\jb}}$,  then for any $x \in G_{\ib,\jb}$,  let
  % $z_1 = x - \gamma_1 e_\alpha $ and $z_2 = x + \gamma_2 e_\alpha$,
  % $\gamma_1,\gamma_2 > 0$, both be
  % elements of $\partial G_{\ib,\jb}$, so that by convexity we have
  % \begin{equation*}
  %   \label{eq:lipschitz-constraint}
  %   \frac{-2B}{\delta_{i_\alpha}}
  %   \le \frac{f(z_1) - f( z_1 -  \delta_{i_\alpha}e_\alpha) }{\delta_{i_\alpha}}
  %   \le  \frac{ f(x+ k e_\alpha)-f(x)}{k}
  %   \le \frac{f(z_2  +\delta_{i_\alpha} e_\alpha) - f( z_2) }{\delta_{i_\alpha} }
  %   \le \frac{2 B } { \delta_{i_\alpha}},
  % \end{equation*}
  % using \eqref{eq:dist-G-to-boundary-D}.  Thus, $f$ satisfies a Lipschitz
  % constraint in the direction of $e_\alpha$.
\end{proof}

\smallskip

The next  lemma is necessary for us to be able to apply
\eqref{eq:bracketing-cover}.  To state it, we first define some constants.
For $k \in \lb 1, \ldots, d\rb$, let $d_{i,j,k} := d(E_i, F_j)$ where $E_i,$ $i=1,\ldots, N_k$, is a $(d-k)$-face
%% (a face of dimension $d-k$)
and $F_j$, $j=1,\ldots, N$, is a facet.  Then let
\begin{equation}
  \label{eq:21}
  r_D :=
  \min \lb d_{i,j,k} : d_{i,j,k} \ne 0, k \in \lb 1,\ldots, d \rb \rb > 0.
\end{equation}
% Let
% \begin{equation}
%   \label{eq:defn:u}
%   u \equiv u_D :=
%   %%   \exp \lp -2(p+1)^2  (p+2) \log 2 \rp
%   r_D/2 \wedge
%   2^{  -2(p+1)^2  (p+2) }
%   \wedge
%   %%   \min_{k \in \lb 1,\ldots, d-1 \rb} \min_{ \jb \in J_k} w_r(G_{\jb})
%   \min_{k \in \lb 1,\ldots, d-1 \rb} \min_{ \jb \in J_k} \rho_{\jb, \alpha}
% \end{equation}
% where $\rho_{\jb,\alpha} := w(G_{\jb},e_\alpha)$ (recall the definition of
% the width function $w$ in Subsection~\ref{subsec:defns-assms}).
%%%%% %% where $E_{\jb}$ is defined on page~\pageref{loc:defn:Ej}.
Let
\begin{equation}
  \label{eq:defn:u}
  u \equiv u_D :=
  %% \exp \lp -2(p+1)^2  (p+2) \log 2 \rp
  r_D/2 \wedge
  2^{  -2(p+1)^2  (p+2) }
  \wedge
  \min_{k \in \lb 1,\ldots, d-1 \rb} \min_{ \jb \in J_k,
    e \in E_{\jb}
    %% \SPAN\lb  e_{\jb, k+1}, \ldots, e_{\jb,d} \rb
  }
  \frac{d^+(x_{\jb}, \relbd  G_{\jb}, e) }{ L_{k,2}}
\end{equation}
where %% $L_{k,2}$
%% and
for $k \in \lb 1, \ldots, d-1 \rb$,
\begin{equation}
  \label{eq:defn:Lk2}
  L_{k,2}:= 1
  \vee
  \max_{\jb \in J_k}
  %% \max_{\beta > k } %%
  \max_{i \in \lb 1, \ldots, N \rb \setminus \jb }
  \sum_{\gamma=1}^k \frac{ \la \tilde f_{\jb, \gamma} ,
    v_{i} \ra }{\la \tilde f_{\jb,\gamma} , v_{j_\gamma} \ra},
\end{equation}
where $\tilde f_{\jb,\gamma}$ are defined in
Proposition~\ref{prop:width-upperbound},
and  $E_{\jb}$ is defined on page~\pageref{loc:defn:Ej}.
\begin{mynotes}
  \begin{notes}
    Definition of $u$ involving $L_{k,2}$ is used in \eqref{eq100}.
  \end{notes}
\end{mynotes}
\begin{mynotes}
  \begin{notes}
    I have changed defn of $u$; remember to check its use carefully!
  \end{notes}
\end{mynotes}
\begin{mylongform}
  \begin{longform}
    The only property of $u$ we need for the following lemma is that $u \le
    r_D/2$.
  \end{longform}
\end{mylongform}
\begin{lemma}
  \label{lem:D-subset-cup-Gij}
  Under Assumption~\ref{assm:simple-polytope}, with $u$ given in \eqref{eq:defn:u}, we have
  $$D \subset \bigcup_{k=0}^d \bigcup_{\jb \in J_k, \ib \in I_k} G_{\ib, \jb}.$$
\end{lemma}
\begin{proof}
  Fix $x \in D$.
  % We need to show that there are no more than $d$ facets $F$ such that $d(x,F) < u$.
  % If $d(x, D) \ge u$ then $x \in G_{A,1}$ (corresponding to $k=0$), so we assume $d(x, D) < u$.  Then let $k_x := \min \lb k \in \lb 0, \ldots, d-1 \rb : d(x,G) < u, \text{ some } k\text{-face } G\rb$ and let $G_x$ be the $k_x$-face such that the minimum is attained.  Now for any facet $F$, if $d(x,F) < u$ then we also have $d(G_x,F) < 2 u \le r_D$.  But this contradicts the definition of $r_D$ unless $d(G_x,F)=0$.
  We need to show that there are no more than $d$ facets $F$ such that $d(x,F) < u$.  If $d(x, \bd D) \ge u$ then $x \in G_{A,1}$ (corresponding to $k=0$), so we assume $d(x, \bd D) < u$.  Then let $k_x := \max \lb k \in \lb 1, \ldots, d \rb : d(x,G) < u, \text{ some } (d-k)\text{-face } G\rb$ and let $G_x$ be any $(d-k_x)$-face such that the minimum is attained.  Now for any facet $F$, if $d(x,F) < u$ then we also have $d(G_x,F) < 2 u \le r_D$.  But this contradicts the definition of $r_D$ unless $d(G_x,F)=0$.  Because $G_x$ is nonempty, $G_x = G_{\jb}$ for some $\jb \in J_{k_x}$ (rather than $\jb \in \tilde{J}_{k_x} \setminus J_{k_x}$).  The distance from $x$ to the boundary of $G_x$ is no smaller than $u$, because otherwise we would contradict the maximality %minimality
  defining $k_x$ since the boundary is given by $(d-(k_x+1))$-faces.  Thus the distance from $x$ to any facet intersecting but not containing $G_x$ is no smaller than $u$.  Furthermore because $D$ is simple, there are exactly $k_x \le d$ facets containing $G_x$; and we have shown that the distance to every facet excluding these $k_x$ is no smaller than $u$.  Thus, $G_x$ is unique and $x$ lies in $G_{\ib, \jb}$ for some $\ib \in I_{k_x}$.
\end{proof}

The next lemma combines Lemma~\ref{lem:embed-parallelotope-hyperrectangle} and Lemma~\ref{lem:Rij-Lipschitz-diameter-bound} with Theorem~\ref{thm:GS-extension}.  The statement depends on the constants $L_{k,1}$, $k\in \lb 1,\ldots,d \rb$, and $L_{\jb,4}$, $\jb \in J_k$.  These depend only on $D$ and are defined in \eqref{eq:defn:Lk1} and \eqref{eq:defn:Lj4}.  %%in Section~\ref{sec:technical-details}.
\begin{mynotes}
  \begin{notes}
    The constants are used in the proposition / lemma right after their definition, respectively; and in order to read the definition of $L_{k,1}$ one has to go into section 4 anyway to see the definition of the $f$'s.  Since they are constants, their exact definition is not needed at this point for first reading.
  \end{notes}
\end{mynotes}
\begin{lemma}
  \label{lem:Pij-bracketing-bound}
  Let Assumption~\ref{assm:simple-polytope} hold.  Fix $k \in \lb 1, \ldots, d\rb$, $\ib \in I_k, \jb \in J_k$. Then for any $p \ge 1$ and  for $\epsilon > 0 $,
  % Let
  % $P_{\ib,\jb} : = \sum_{\alpha=1}^k [0, f_{\alpha}]$ where $f_\alpha$ are
  % defined in Proposition~\ref{lem:parallelotope-vertex-representation}.
  \begin{equation}
    \label{eq:bracketing-card--1}
    \begin{split}
      \MoveEqLeft
      \log \Nb[\epsilon \vol_d (G_{\ib,\jb})^{1/p} ,
      \cC[D,1]|_{G_{\ib,\jb}}, L_p ]
      % \le \sum_{\ib \in I_k} \log \Nb[\alpha_{\ib} \vol (G_{\ib,\jb})^{1/p} ,
      % \cC[G_{\ib,\jb},1, \Gb_{\ib}], L_p ]
      \\
      & \le  c_d \epsilon^{-d/2}\lp 1 +
      \frac{2 d^2}{L_{\jb,4}}  \max_{\alpha=1,\ldots, k}
      \frac{\delta_{i_\alpha+1}}{\delta_{i_\alpha}} +
      \sum_{\alpha=k+1}^d \frac{8 L_{k,1} \rho_{\jb,\alpha} }{u}
      \rp^{d/2}.
    \end{split}
  \end{equation}
  % \begin{equation}
  %   \label{eq:27}
  %   \begin{split}
  %     \MoveEqLeft
  %     \log \Nb[ \epsilon
  %     ,
  %     \cC[G_{\ib,\jb},1, \Gb_{\ib}],
  %     L_\infty  % L_p
  %     ]
  %     \le \\
  %     &      c_d \lp \frac{1 +
  %     \sum_{\alpha=1}^k \frac{  2  (\delta_{i_\alpha+1}-\delta_{i_\alpha})}{\delta_{i_\alpha}} +
  %     \sum_{\alpha=k+1}^d \frac{8 L_{k,1} \rho_{\jb,\alpha} }{u} } {a_{\ib}}
  %     \rp^{d/2}
  %   \end{split}
  % \end{equation}
\end{lemma}
\begin{proof}
  Let
  \begin{equation}
    \label{eq:45}
    \Gb_{\ib} := \lp \frac{2}{d(G_{\ib,\jb}, \bd D, e_1)}, \ldots,
    \frac{2}{d(G_{\ib,\jb}, \bd D, e_k)}, \frac{2}{u}, \cdots, \frac{2}{u} \rp
  \end{equation}
  where $e_{\ib,\jb, \alpha} \equiv e_\alpha$, $\alpha=1,\ldots,d$, is given
  by Proposition~\ref{prop:basis-lipschitz}.
  Then
  \begin{equation}
    \label{eq:C-inclusion-2}
    \cC[D,1]|_{G_{\ib,\jb}}
    \subset \cC[G_{\ib,\jb}, 1, \Gb_{\ib}, {\ve e}]
  \end{equation}
  where ${\ve e} = \lp e_1, \ldots, e_d \rp$.
  Let $\tilde{f}_{j_\gamma} $ be given by
  Lemma~\ref{lem:parallelotope-vertex-representation} applied to the $k$
  linearly independent unit normal vectors $v_{j_1}, \ldots, v_{j_k}$, and
  (as in that lemma, with ``$d_\beta$'' given by $
  (\delta_{i_\gamma+1}-\delta_{i_\gamma})$), let
  \begin{equation}
    \label{eq:defn:fjgamma}
    f_{\ib,\jb,j_\gamma}
    \equiv f_{j_\gamma} := (\delta_{i_\gamma+1}-\delta_{i_\gamma})
    \tilde{f}_{j_\gamma} / \la \tilde{f}_{j_\gamma}, v_{j_\gamma} \ra.
  \end{equation}
  Let $P_{\ib,\jb} := \sum_{\gamma=1}^k [0,f_{j_\gamma}]$, where $[0,v]:= \lb \lambda v : \lambda \in [0,1] \rb$.  By Lemma~\ref{lem:embed-parallelotope-hyperrectangle}, $P_{\ib,\jb} \subset \sum_{\alpha=1}^k [0, \gamma_\alpha e_\alpha]$ where $\gamma_\alpha$ are given by the lemma.  Thus by \eqref{eq:Gij-subset-Rij}, for some $x \in G_{\ib,\jb}$,
  \begin{equation}
    \label{eq:39}
    G_{\ib,\jb}
    \subset x + \sum_{\alpha=1}^k  [0, \gamma_\alpha e_\alpha]
    + \sum_{\alpha=k+1}^d \ls - 2 L_{k,1} \rho_{\jb,\alpha} e_\alpha,
    2 L_{k,1} \rho_{\jb,\alpha} e_\alpha \rs.
  \end{equation}
  Now, using \eqref{eq:C-inclusion-2}, we apply
  Theorem~\ref{thm:GS-extension} to see
  \begin{equation}
    \label{eq:44}
    \begin{split}
      \MoveEqLeft \log \Nb[\epsilon \Vol_d(G_{\ib,\jb})^{1/p} ,
      \cC[D,1]|_{G_{\ib,\jb}}, L_p] \\
      & \le c_d \epsilon^{-d/2} \lp 1 + \sum_{\alpha=1}^k \frac{2        \gamma_\alpha }{d(G_{\ib,\jb}, \bd D, e_\alpha)} +
      \sum_{\alpha=k+1}^d \frac{8 L_{k,1} \rho_{\jb,\alpha} }{u} \rp^{d/2}
    \end{split}
  \end{equation}
  Now by applying \eqref{eq:40},  \eqref{eq:36}, and \eqref{eq:37}
  with $v = e_\alpha$, we see that
  \begin{equation}
    \label{eq:26}
    \frac{2        \gamma_\alpha }{d(G_{\ib,\jb}, \bd D, e_\alpha)}
    \le \frac{2 d \diam(G_{\ib,\jb}, e_\alpha)}{d(G_{\ib,\jb}, \bd D, e_\alpha)}
    \le \frac{2d \min_{\substack{\beta = 1, \ldots , k  }}
      \frac{ \delta_{i_\beta+1}}{\lv \la  e_\alpha, v_{j_\beta} \ra \rv} }{% L_{\jb,4}
      \max_{\substack{\beta = 1, \ldots , k }}
      \frac{ \delta_{i_\beta}}{\lv \la  e_\alpha, v_{j_\beta} \ra \rv }}
    \le \frac{2d}{L_{\jb,4}} \max_{\beta = 1, \ldots, k} \frac{\delta_{i_\beta +1}}{\delta_{i_\beta}}
  \end{equation}
  \begin{mynotes}
    \begin{notes}
      (In some version of the paper I had a 'min' in the numerator of third fraction rather than  a max -- seems obviously wrong so i am changing, but making a note here in case i'm missing something complicated ... )

      \commentC{(( March 2019: Changing the numerator of third fraction back to 'min', which is what \eqref{eq:36} gives!?  and without which, the final inequality seems wrong.  (The 'min' makes sense because the inner product $< e_\alpha, v_{j_\beta}>$ could be very small or $0$).  I have also changed the denominator 'min' to a 'max'. again to line up with \eqref{eq:37}.   Everything that was written here is very confusing to me.  I'm not sure if there is a hidden error somewhere, or what.  But))}
      
      (Add justification for final step?)

      NOTE       the strict inequality in definition of $L_{\jb,4}$ is because the term
      \begin{equation*}
        \frac{2d \min_{\substack{\beta = 1, \ldots , k  }}
          \frac{ \delta_{i_\beta+1}}{\lv \la  e_\alpha, v_{j_\beta} \ra \rv} }{% L_{\jb,4}
          \max_{\substack{\beta = 1, \ldots , k }}
          \frac{ \delta_{i_\beta}}{\lv \la  e_\alpha, v_{j_\beta} \ra \rv }}
      \end{equation*}
      is finite.  Thus the $\la e_\alpha, v_{j_\beta} \ra$ terms must be finite!
      So we get
      \begin{equation*}
        L_{\jb,4}^{-1} := \lp {\min_{e_\alpha}  \min_{v_{j_\beta} : \la v_{j_\beta}, e_\alpha \ra > 0}
          \lv \la e_\alpha v_{j_\beta} \ra \rv} \rp^{-1}.
      \end{equation*}
    \end{notes}
  \end{mynotes}
  where
  \begin{equation}
    \label{eq:defn:Lj4}
    % L_{\jb,4}^{-1} := \lp {\min_{v_{j_\beta} : \la v_{j_\beta}, e_\alpha \ra > 0}
    % \lv \la e_\alpha v_{j_\beta} \ra \rv} \rp^{-1}.
    L_{\jb,4} :=  {\min_{e_1,\ldots, e_d} \min_{v_{j_\beta} : \la v_{j_\beta}, e_\alpha \ra > 0}
      \lv \la e_\alpha v_{j_\beta} \ra \rv}.
  \end{equation}
  (We can restrict to $v_{j_\beta}$ such that $\la v_{j_\beta}, e_\alpha \ra > 0$ in the definition of $L_{\jb,4}$ because the numerator in \eqref{eq:26} is finite.)  Thus \eqref{eq:44} is bounded above by
  \begin{equation*}
    c_d \epsilon^{-d/2} \lp 1 +
    \frac{2 d^2}{L_{\jb,4}} \max_{\beta = 1, \ldots, k} \frac{\delta_{i_\beta +1}}{\delta_{i_\beta}}
    +
    \sum_{\alpha=k+1}^d \frac{8 L_{k,1} \rho_{\jb,\alpha} }{u} \rp^{d/2}.
  \end{equation*}
\end{proof}

Now we present our main theorem.  It gives a bracketing entropy of order $\epsilon^{-d/2}$ when $D$ is a fixed simple polytope.  Its proof relies on embedding $G_{\ib,\jb}$ in a set $R_{\ib,\jb}$ (defined in \eqref{eq:defn:R}) which is a set-sum of a parallelotope and a hyperrectangle with axes given by Proposition~\ref{prop:basis-lipschitz}.  We need to control the distance of $G_{\ib,\jb}$ to $\partial D$, and we need to control the size of $R_{\ib,\jb}$ in terms of the widths along its axes.  Then we can use the results of Section~\ref{sec:Lipschitz-bracketing} on $R_{\ib,\jb}$ and thus on $G_{\ib,\jb}$.  We defer some statements and proofs of needed facts about $G_{\ib,\jb}$ and $R_{\ib,\jb}$ until Section~\ref{sec:technical-details}.
\begin{mynotes}
  \begin{notes}
    If $u_k$ is bounded (above and) below by $d(G_{\jb}, e_\alpha)$ ,
    $\alpha=k+1,\ldots, d$, then the $1/u_k^{d-k}$ will cancel with the
    volume term here, leaving just a count of the number of $k$-faces times
    $B_u^k$ here.
  \end{notes}
\end{mynotes}
%% Before stating the theorem, let us define some constants.
% \begin{mynotes}
%   \begin{notes}
%     If the form of $S$ is not given in the theorem then these constants do not need to be given here?
%   \end{notes}
% \end{mynotes}
% Let
% \begin{equation}
%   \label{eq:defn:Lk1}
%   L_{k,1} :=
%   1 \vee
%   \max_{\jb \in J_k}
%   \max_{\substack{\|e\| = 1\\ e \in E_{\jb}}}
%   \max_{\substack{j \in \lb 1,\ldots, N\rb \setminus \jb ; \,  \la e, v_j \ra < 0 \\  \la v_i,v_j \ra \ge 0 , \text{ some } i \in \jb }}
%   \la -e, v_j \ra^{-1},
% \end{equation}
% where $E_{\jb} := \SPAN \lb e_{\jb, k+1}, \ldots, e_{\jb, d}\rb$ from Proposition~\ref{prop:basis-lipschitz}, and we abuse notation as convenient to treat $\jb$ as if it were a set rather than a vector.
\begin{mynotes}
  \begin{notes}
    It seems like there are several issues here.
    \begin{enumerate}
    \item Why is the second part of the definition of $u$ needed?  is it operative?  Specifically, dividing by $L_{k,2}$?  See at least \eqref{eq100} where it is used but I am unconvinced it is really needed (could have $1+u$ in place of $2$?)
    \end{enumerate}

    Another note.  Seems like the defn of $u$ may be the issue
    governing the difference between strict curvature and polytope.  (and may
    suggest answer to question of how much curvature matters ).  $u$ is chosen to govern the ratio of  width of $G_{i,j}$  to  $G_j$.   Basically as $G_j$ gets smaller in width $u$ has to get smaller.  ``Curvature'' means $G_j$ has width $0$ so $u$ must be $0$.  Or rather, the width of the approximating pieces, $G_{i,j}$ are infinitely larger than that of $G_j$. changes paradigm.

  \end{notes}
\end{mynotes}
% Let
% \begin{equation}
%   \label{eq:12}
%   L_{\jb,3} := \max_{\alpha \in \lb 1,\ldots,k\rb} 1 / \la \tilde
%   f_\alpha, v_{j_\alpha} \ra.
% \end{equation}

\begin{theorem}
  \label{thm:main-thm}
  Let Assumption~\ref{assm:simple-polytope} hold for a convex polytope $D
  \subseteq \prod_{i=1}^d [a_i,b_i]$.  Fix $p \ge
  1$.  Then
  % for some $\epsilon_0>0$ and for $0 < \epsilon \le \epsilon_0 B
  % \lp \prod_{i=1}^d b_i-a_i \rp^{1/p} $, %% what is upper bound
  for all $\epsilon > 0$,
  \begin{equation}
    \label{eq:20}
    \log \Nb[\epsilon , \cC(D,B), L_p]
    \le
    % \lb \tilde c_d \sum_{k=0}^d \frac{ B_u^k }{ u^{(d-k)d/2}} \sum_{\jb \in
    % J_k^D} \vol_{d-k}(G_{\jb})^{d/2} \rb
    S
    \epsilon^{-d/2}
    \bigg( B  \bigg( \prod_{i=1}^d (b_i-a_i) \bigg)^{1/p} \bigg)^{d/2},
  \end{equation}
  where
  % $S := \lp \sum_{k=0}^d \gamma_k^D A_u^k \sum_{\jb \in J_k^D}
  % \vol_{d-k}(G_{\jb}) \rp $,
  % $S := \lp \sum_{k=0}^d
  % (2 L_{k,1}) ^{d-k}  A_u^k
  % \sum_{j \in J_k} \vol_{d-k} (G_{\jb}) L_{\jb,3}^k   \rp ^{1/p} $,
  $S$ is a constant depending only on $d$ and $D$.
\end{theorem}
The form of the constant $S$ is given in the proof of the theorem.
\begin{proof}
  Fix $\epsilon > 0$.
  First, we will reduce to the case where $D \subset [0,1]^d$ and $B = 1$ by
  a scaling argument.   Let $C$ be an affine map from $\prod_{i=1}^d
  [a_i,b_i]$ to $[0,1]$, where  $\tilde D$ is the image of $D$,
  and assume we have a bracketing cover $[\tilde l_1, \tilde u_1],
  \ldots, [\tilde l_N, \tilde u_N]$ of $\cC[ \tilde D, 1]$.  Let $l_i
  := B \, \tilde{l}_i \circ C$ and similarly for $u_i$, so that
  $[l_1,u_1],\ldots, [l_N,u_N]$ form brackets for $\cC[ D, B]$.  Their
  $L_p^p$ size
  is
  \begin{equation*}
    \int_D \lp u_i(x) - l_i(x) \rp^p dx
    = B^p \int_{\tilde D} (\tilde u_i(x) - \tilde l_i(x))^p \prod^d (b_i-a_i) dx.
  \end{equation*}
  Thus,
  \begin{equation*}
    \bracketing \bigg( \epsilon B \bigg( \prod^d b_i-a_i\bigg)^{1/p} , \cC[D, B], L_p \bigg)
    \le \Nb[ \epsilon  , \cC[ \tilde D, 1], L_p],
  \end{equation*}
  so apply the theorem with  $\eta =  \epsilon / B \lp \prod^d
  b_i-a_i\rp^{1/p}$ for $\epsilon$.  Note that the
  %% constants $S_1$ and $S_2$  depend
  constant $S$ depends %
  on $\tilde D$, the version of $D$ normalized to lie in $[0,1]^d$.

  We now assume $D \subset [0,1]^d$ and $B=1$.  We specify the sequence in \eqref{eq:5} and $a_{\ib,k} \equiv a_{\ib} > 0$, which will govern the $L_p$-sizes of our brackets on $G_{\ib,\jb}$, as follows.
  \begin{mylongform}
    \begin{longform}
      Note the sizes $a_{\ib}$ are constant over $\jb$.
    \end{longform}
  \end{mylongform}
  Let
  \begin{equation}
    \label{eq:15}
    \delta_i := \exp\lb p \lp \frac{p+1}{p+2}\rp^{i-1} \log \epsilon \rb
    \quad \mbox{ for } i = 1, \ldots, A,
    \qquad \mbox{ and } \qquad
    \delta_0 = 0.
  \end{equation}
  %% and $\delta_0 = 0$.
  Note that this implicitly defines $A$, by \eqref{eq:5} and \eqref{eq:defn:u}.
  \begin{mynotes}
    \begin{notes}
      $a_{\ib, k}$ will be the bracket size (the $\epsilon$ for each small piece).
    \end{notes}
  \end{mynotes}
  For $k \in \lb 1, \ldots d \rb$ and $\ib \in
  I_k$, we will let $ a_{(i_1,\ldots,i_k)} = 2$ if $i_\alpha = 0$ for any
  $\alpha \in \lb 1,\ldots ,k \rb$, and otherwise we let
  \begin{align*}
    a_{(i_1,\ldots,i_k)} & :=
                           \prod_{\beta=1}^k a_{i_\beta} :=
                           \prod_{\beta=1}^k \epsilon^{1/k} \exp \lb - p
                           \frac{(p+1)^{i_\beta-2}}{(p+2)^{i_\beta-1}}
                           \log \epsilon \rb.
  \end{align*}
  For the $k = 0$ case, let $a_{A} := \epsilon / u$. %% Is this clear?
  \begin{mynotes}
    \begin{notes}
      If you think through the guntuboyina and sen induction: an
      inductive argument essentially plugs $\epsilon_1$ in for $\epsilon_2$,
      which is to say to get $a_{(i_1,\ldots, i_k)}$, i should be COMPOSING the
      formula for G\&S's $\alpha_m$ on itself.  This basically reduces to the
      above.
    \end{notes}
  \end{mynotes}
  Let
  \begin{equation}
    \label{eq:13}
    a = \bigg( \sum_{k=0}^d \sum_{\jb \in J_k, \ib \in I_k} a_{\ib}^p
    \vol_{d}(G_{\ib,\jb}) \bigg)^{1/p}.
  \end{equation}
  Then %% by Assumption~\ref{assm:simple-polytope} and
  since $D \subset \cup_{k=0}^d \cup_{\jb \in J_k, \ib \in I_k} G_{\ib, \jb}$ by Lemma~\ref{lem:D-subset-cup-Gij}, as in \eqref{eq:bracketing-cover}, \label{tmp:loc:Jk}
  \begin{mylongform}
    \begin{longform}
      \begin{equation*}
        % \label{eq:bracket-number}
        \Nb[ a  , \cC[D, 1], L_p ]
        \le \prod_{k=0}^d \prod_{\jb \in J_k, \ib \in I_k}
        \Nb[a_{\ib} \vol_{d}(G_{\ib, \jb}) ^{1/p}, \cC[D,1]|_{G_{\ib,\jb}}, L_p],
      \end{equation*}
    \end{longform}
  \end{mylongform}
  \begin{myold}
    \begin{old}
      Now by Lemma~\ref{lem:Gij-cover-D}, we can ignore all terms with
      $\jb \in J_k \setminus J_k^D$, where
      $J_k^D := \lb \jb \in J_k : \cap_{\alpha=1}^k H_{j_\alpha} \cap D
      \text{ is a } (d-k)\text{-face of } D \rb$.
    \end{old}
  \end{myold}
  \begin{equation}
    \label{eq:bracket-number}
    \log \Nb[ a  , \cC[D, 1], L_p ]
    \le \sum_{k=0}^d \sum_{\jb \in J_k}
    \sum_{ \ib \in I_k} \log \Nb[a_{\ib} \vol_{d}(G_{\ib,\jb})^{1/p} , \cC[D,1]|_{G_{\ib,\jb}}, L_p].
  \end{equation}
  First, consider the case $k \in \lb 1, \ldots , d \rb$ and compute the sum over $I_k$ for a fixed $\jb \in J_k$.
  \begin{myold}
    \begin{old}
      Let $\rho_{\jb,\alpha} := w(G_{\jb},e_\alpha)$
      (recall the definition of the width function $w$ in
      Subsection~\ref{subsec:defns-assms}).
      %% By Proposition~\ref{prop:basis-lipschitz},
      Let $V_{\ib,\jb,1} \equiv V_1$ be $\SPAN \lb e_1, \ldots,e_k \rb$ and
      $V_{\ib,\jb,2} \equiv V_2, \ldots, V_{\ib,\jb,d-k+1} \equiv V_{d-k+1}$ be
      $\SPAN \lb e_{k+1} \rb, \ldots, \SPAN \lb e_{d} \rb$, respectively.  Let
      $R_1 := \max_{\alpha \in \lb 1, \ldots, k \rb}
      2(\delta_{i_\alpha+1}-\delta_{i_\alpha}) / \delta_{i_\alpha}$, %
      and %
      let
      $ \ve{R}_{\ib,\jb} := \lp R_1, \frac{8 L_{k,1} \rho_{\jb, k+1}}{u}, \ldots,
      \frac{8 L_{k,1} \rho_{\jb, d}}{u} \rp$.  Then by
      Lemma~\ref{lem:Rij-Lipschitz-diameter-bound},
      \begin{equation}
        \label{eq:C-inclusion}
        \cC[D,1]|_{G_{\ib,\jb}}
        \subset
        \cC[G_{\ib,\jb}, 1, \ve{R}_{\ib,\jb}, V_1, \ldots, V_{d-k+1} ].
      \end{equation}
      % \begin{equation}
      %   \label{eq:C-inclusion}
      %   \cC[D,1]|_{G_{\ib,\jb}}
      %   \subset \cC[G_{\ib,\jb}, 1, \Gb, {\ve e}]
      % \end{equation}
      % where $\ve{R}_{\ib,\jb}$ and $V_1, \ldots, V_{d-k+1}$ are defined in the lemma.
      % where $\Gb_{\ib} = (2/\delta_{i_1}, \ldots, 2/\delta_{i_k}, 2/u, \ldots,
      % 2/u)$.
      We have
      \begin{equation}
        \label{eq:C-inclusion}
        \cC[D,1]|_{G_{\ib,\jb}}
        \subset \cC[G_{\ib,\jb}, 1, \Gb_{\ib}, {\ve e}],
      \end{equation}
      % where $\Gb \in \RR^d$ has $2/\delta_{i_\alpha}$ in the $j_{\alpha}$
      % component, for $\alpha = 1, \ldots, k$, and is $2/ u$ in the remaining
      % components.
      where $\Gb_{\ib}$ is given by \eqref{eq:45} and ${\ve e} = (e_1, \ldots,
      e_d)$ is the basis given by Proposition~\ref{prop:basis-lipschitz}.
      % We will need the following constant (depending on $D$).
      % For $k \in \lb 1, \ldots, d -1 \rb$,  let
      % \begin{equation}
      %   \label{eq:defn:Lk1}
      %   L_{k,1} :=
      %   1 \vee
      %   \max_{\jb \in J_k}
      %   \max_{\substack{\|e\| = 1\\ e \in E_{\jb}}}
      %   \max_{\substack{j \in \lb 1,\ldots, N\rb \setminus \jb ; \,  \la e, v_j \ra < 0 \\  \la v_i,v_j \ra \ge 0 , \text{ some } i \in \jb }}
      %   \la -e, v_j \ra^{-1},
      % \end{equation}
      % where $E_{\jb} := \SPAN \lb e_{\jb, k+1}, \ldots, e_{\jb, d}\rb$ from
      % Proposition~\ref{prop:basis-lipschitz}, and we abuse notation as convenient
      % to treat $\jb$ as if it were a set rather than a vector.
      % We also (arbitrarily) define $L_{d,1}:= 1$, for ease of presentation later on.
      Then by \eqref{eq:C-inclusion} and the first inequality of Theorem~\ref{thm:GS-extension} we have
      \begin{equation}
        \label{eq:bracketing-card--1}
        \sum_{\ib \in I_k} \log \Nb[a_{\ib} \vol_d (G_{\ib,\jb})^{1/p} ,
        \cC[D,1]|_{G_{\ib,\jb}}, L_p ]
        % \le \sum_{\ib \in I_k} \log \Nb[\alpha_{\ib} \vol (G_{\ib,\jb})^{1/p} ,
        % \cC[G_{\ib,\jb},1, \Gb_{\ib}], L_p ]
        \le \sum_{\ib \in I_k}
        \log \Nb[a_{\ib}
        ,
        \cC[G_{\ib,\jb},1, \Gb_{\ib}],
        L_\infty  % L_p
        ].
      \end{equation}
    \end{old}
  \end{myold}
  We use the trivial bracket $[-1,1]$ for any $G_{\ib,\jb}$ where $i_\alpha=0$ for any $\alpha \in \lb 1, \ldots, k\rb$.
  Otherwise apply Lemma~\ref{lem:Pij-bracketing-bound}
  %% Theorem~\ref{thm:GS-extension}, %%
  which shows us that the sum over the remaining terms
  \begin{mylongform}
    \begin{longform}
      in the sum over $\ib in I_k$
    \end{longform}
  \end{mylongform}
  in
  %% \eqref{eq:bracketing-card--1}
  \eqref{eq:bracket-number} is bounded by
  \begin{align}
    \label{eq:bracketing-card-00}
    % \MoveEqLeft
    \sum_{i_1=1}^A \cdots \sum_{i_k=1}^A c_d  a_{\ib}^{-d/2}\lp 1 +
    \frac{ 2d^2}{L_{\jb,4}} \max_{\alpha=1,\ldots, k} \frac{ \delta_{i_\alpha +1}}{\delta_{i_\alpha}}
    +
    \sum_{\alpha=k+1}^d \frac{8 L_{k,1} \rho_{\jb,\alpha} }{u}
    \rp^{d/2}. % \\
  \end{align}
  % To apply
  % Theorem~\ref{thm:GS-extension-v2_directional-version}, we first apply
  % Lemma~\ref{lem:Rij-Lipschitz-diameter-bound} to $R_{\ib,\jb}$.
  Since $L_{k,1} \ge 1$ %
  %% $L_{k,2} \ge 1$, %%!
  and %
  %% $$ u \le \rho_{\jb,\alpha} / L_{k,2}$$ %
  $ u \le \rho_{\jb,\alpha} $ %
  by
  \eqref{eq:defn:u} for all $k, \ib, \jb$ and $\alpha = k+1,\ldots, d$,
  we have
  $\sum_{\alpha = k+1}^d \frac{ 8 \rho_{\jb,\alpha} L_{k,1}}{u} = 4 L_{k,1}
  \sum_{\alpha=k+1}^d \frac{ 2 \rho_{\jb,\alpha}}{u} \le 4 L_{k,1}
  \prod_{\alpha=k+1}^d \frac{2 \rho_{\jb,\alpha}}{u}$ (using the fact that
  for $a,b\ge 2$, $ab \ge a+b$).  We also bound $\max_{\alpha=1,\ldots,k }
  2\delta_{i_\alpha+1}  / \delta_{i_\alpha} \le \prod_{\alpha=1}^k 2
  \delta_{i_\alpha+1} / \delta_{i_\alpha}$
  % $\sum_{\alpha=1}^k 2 (\delta_{i_\alpha+1} - \delta_{i_\alpha}) / \delta_{i_\alpha}
  % \le \prod_{\alpha=1}^k 2 \delta_{i_\alpha+1} / \delta_{i_\alpha} $
  since $2 \delta_{i_\alpha+1}/\delta_{i_\alpha} > 2$.  Thus
  \eqref{eq:bracketing-card-00} is bounded above by
  \begin{align}
    \label{eq:bracketing-card-2}
    c_d %%  (d!)^{d/2}
    d^2L_{\jb,4}^{-1}
    \lp 1 + 2^{d-k+2} L_{k,1} \prod_{\alpha=k+1}^d \frac{\rho_{\jb,\alpha}}{u}  \rp^{d/2}
    \sum_{i_1=1}^A \cdots \sum_{i_k=1}^A a_{\ib}^{-d/2} \prod_{\alpha=1}^k \lp
    \frac{2 \delta_{i_\alpha+1}}{\delta_{i_\alpha}} \rp^{d/2},
  \end{align}
  %% Now \eqref{eq:bracketing-card-2} is
  which is
  \begin{equation}
    \label{eq:bracketing-card-}
    c_d %%  (d!)^{d/2}
    d^2L_{\jb,4}^{-1}
    \lp 1 + 2^{d-k+2} L_{k,1} \prod_{\alpha=k+1}^d \frac{\rho_{\jb\alpha}}{u} \rp^{d/2}
    \sum_{i_1=1}^A \cdots \sum_{i_k=1}^A \prod_{\beta=1}^k \lp \frac{2
      \delta_{i_\beta+1}}{\delta_{i_\beta}  a_{i_\beta} } \rp^{d/2}.
  \end{equation}
  Note that when $k=d$ we take the product over an empty set to be $1$.
  \begin{mylongform}
    \begin{longform}
      In fact, $ \lp 1 + 2^{d-k+2} L_{k,1} \prod_{\alpha=k+1}^d
      \frac{\rho_{\jb\alpha}}{u} \rp^{d/2}$ could be taken as $1$ when $k=d$,
      although we do not deal with this notational issue.
    \end{longform}
  \end{mylongform}
  For $i=1,\ldots, A$, let
  \begin{equation}
    \label{eq:23}
    \zeta_i \equiv \zeta_{i,k} := \sqrt{ \epsilon^{1/k} \delta_{i+1} / (\delta_i
      a_i) },
  \end{equation}
  so that $ \sum_{i_1=1}^A \cdots \sum_{i_k=1}^A
  \prod_{\beta=1}^k \lp \frac{2 \delta_{i_\beta+1}}{\delta_{i_\beta}
    a_{i_\beta} } \rp^{d/2}$ equals
  \begin{align*}
    \sum_{i_1=1}^A \cdots \sum_{i_k=1}^A 2^{kd/2} \epsilon^{- d/2}
    \prod_{\beta = 1}^k \zeta_{i_\beta}^d
    & =  2^{kd/2}\epsilon^{-  d/2} \sum_{i_1=1}^A \zeta_{i_1}^d \sum_{i_2=1}^A \zeta_{i_2}^d
      \cdots \sum_{i_k=1}^A \zeta_{i_k}^d \\
    &  = \epsilon^{-d/2} 2^{kd/2} B_u^k
  \end{align*}
  where, for $0 < \epsilon \le 1$
  \begin{equation}
    \label{eq:4}
    B_u := \sum_{i=1}^A \zeta_{i}^d
    %% \le    2 u^{ d/ (2(p+1)^2)},
    \le    2 u^{ d/ (2(p+1)(p+2))},
  \end{equation}
  by Lemma~\ref{lem:zetasum}.

  Next, we will relate the term $ \lp 1 + 2^{d-k+2} L_{k,1} \prod_{\alpha=k+1}^d \frac{\rho_{\jb\alpha}}{u} \rp^{d/2}$ to $\vol_{d-k}(G_{\jb})$. Recall that $A_{\jb}$ is the ellipsoid defined in \eqref{eq:FritzJohn} which has diameter in the $e_\alpha$ direction given by $\gamma_{\jb,\alpha}$.  By \eqref{eq:FritzJohn}, $\rho_{\jb,\alpha} \le d \gamma_{\jb,\alpha}.$
  % By
  % Proposition~\ref{prop:width-upperbound} and \eqref{eq:defn:u}, and by
  % \eqref{eq:FritzJohn},
  % \begin{equation*}
  %   w_{\ib,\jb,\alpha} \le w(G_{\jb}, e_\alpha) + M_{k,D} u
  %   \le 2 w(G_{\jb},e_\alpha)
  %   \le   2 d \rho_{\jb, \alpha}.
  % \end{equation*}
  % $\prod_{\alpha = k+1}^d w_{\ib,\jb,\alpha} \le (2d)^{d-k}
  % \prod_{\alpha=k+1}^d \gamma_{\jb,\alpha}$, and
  The volume of $A_{\jb}$ is $\vol_{d-k}(A_{\jb}) = \lp \prod_{\alpha=k+1}^d
  \gamma_{\jb,\alpha}/2 \rp \pi^{(d-k)/2}/\Gamma((d-k)/2+1)$.
  Thus,
  letting $C_d :=  \frac{(2d)^{d-k} \Gamma((d-k)/2+1)}{\pi^{(d-k)/2}}$, we have
  \begin{equation*}
    \prod_{\alpha = k+1}^d \rho_{\jb,\alpha}
    \le C_d \vol_{d-k}(A_{\jb})
    \le  C_d \vol_{d-k} ( G_{\jb}).
  \end{equation*}
  Thus we
  have shown that
  \eqref{eq:bracketing-card-} is bounded above by
  % \begin{equation}
  %   \label{eq:bracketing-card-sum}
  %   c_d (d!)^{d/2} 2^{kd/2} \lp 1 + \lp \frac 8 u \rp^{d-k} C_d  \vol_{d-k} (G_{\jb})
  %   L_{k,1}^{d-k}
  %   \rp^{d/2}
  %   B_u^k
  %   \cdot
  %   \epsilon^{- d/2}.
  % \end{equation}
  \begin{equation}
    \label{eq:bracketing-card-sum}
    c_d %%  (d!)^{d/2}
    d^2L_{\jb,4}^{-1}
    2^{kd/2}
    \lp 1 + 2^{d-k+2} L_{k,1} u^{-(d-k)} C_d \vol_{d-k}(G_{\jb})  \rp^{d/2}
    B_u^k
    \cdot
    \epsilon^{- d/2}.
  \end{equation}
  Therefore,
  \begin{mylongform}
    \begin{longform}
      using \eqref{eq:4},
    \end{longform}
  \end{mylongform}
  letting $\tilde{c}_{d,k} := c_d d^2
  2^k 2^{kd/2}$, we have shown that
  \begin{equation}
    \label{eq:16}
    \begin{split}
      \MoveEqLeft \sum_{\ib \in I_k} \log \Nb[a_{\ib} \vol_d (G_{\ib,\jb})^{1/p} ,
      \cC[D,1]|_{G_{\ib,\jb}}, L_p ] \\
      %% & \le \tilde{c}_{d,k} u^{kd / 2(p+1)^2}
      & \le L_{\jb,4}^{-1} \tilde{c}_{d,k} u^{kd / 2(p+1)(p+2)}
      %% \lp 1 + \lp \frac 8 u \rp^{d-k} C_d  \vol_{d-k} (G_{\jb})      L_{k,1}^{d-k}      \rp^{d/2}
      \lp 1 + 2^{d-k+2} L_{k,1} u^{-(d-k)} C_d \vol_{d-k}(G_{\jb})  \rp^{d/2}
      \epsilon^{-d/2}.
    \end{split}
  \end{equation}
  \begin{mynotes}
    \begin{notes}
      Can we tie together $u^{-(d-k)}$ and $\vol_{d-k}(G_{\jb})$?
    \end{notes}
  \end{mynotes}
  \begin{myold}
    \begin{old}
      Then, gathering the constants together into $\tilde c_d$,
      %% \eqref{eq:bracketing-card-sum}
      %% is bounded above by
      we have shown
      \begin{equation}
        \label{eq:14}
        \begin{split}
          \MoveEqLeft    \sum_{\ib \in I_k} \log \Nb[a_{\ib} \vol (G_{\ib,\jb})^{1/p} ,
          \cC[D,1]|_{G_{\ib,\jb}}, L_p ] \\
          & \le
          \epsilon ^{- d/2} \tilde c_d \lp  \frac{ \vol_{d-k}(G_{\jb})
            %% \prod_{\alpha=k+1}^d L_{k,1}
            L_{k,1}^{d-k}
          }{ u^{d-k}}
          \rp^{d/2}   % B_u^{k}.
          u^{kd/ (2(p+1)^2)}.
        \end{split}
      \end{equation}
    \end{old}
  \end{myold}
  Display \eqref{eq:16} holds for $k \in \lb 1,\ldots, d \rb$.  When $k=0$,
  recalling $a_A = \epsilon/u$, we have
  \begin{equation}
    \label{eq:17}
    \begin{split}
      %% \MoveEqLeft
      \log \Nb[ a_A \vol_d(G_{A,1})^{1/p},
      \cC[D,1]|_{G_{A,1}}, L_p]
      & \le
      c_d \lp u + 2d \rp^{d/2} \epsilon^{-d/2}
    \end{split}
  \end{equation}
  by Theorem~\ref{thm:GS-extension}
  %% since $\vol_d(G_{A,1}) \le \vol_d(D)$
  since $\cC[ D, 1]|_{G_{A,1}} \subset \cC[ G_{A,1}, 1, \frac{2}{u} \ve{1}]$ where
  $\ve{1} \in \RR^d$
  is the vector of all $1$'s.
  \begin{mynotes}
    \begin{notes}
      ( Do we want to use that $L_{0,1} = 1$ and $\vol_d(G_{A,1}) \le
      \vol_d(D) \le 1$ to align the formula for $k=0$ with the formula for $k
      \ge 1$?)
    \end{notes}
  \end{mynotes}
  Then, combining \eqref{eq:17} and \eqref{eq:16}, the cardinality of the collection
  of brackets covering the entire domain $D$ is given by summing over $\jb
  \in J_k$ and $k \in \lb 0,\ldots,d \rb$.

  \medskip

  We have computed the cardinality of the brackets.  Now we bound their size.
  Let $I^0_k$ be the subset of $ \ib \in I_k$ such that
  some $i_\alpha$ is $0$, and let $I^+_k := I_k \setminus I^0_k$.
  We have
  % \begin{equation}
  %   \label{eq:ap}
  %   a^p \le
  %   a_A^p \vol_d(D) +
  %   \sum_{k=1}^d (2L_{k,1})^{d-k} \sum_{\jb \in J_k} \vol_{d-k} (G_{\jb})
  %   \sum_{\ib \in I_k} a_{\ib}^p \prod_{\alpha=1}^k \frac{\delta_{i_\alpha + 1} -
  %   \delta_{i_\alpha}}{\la \tilde f_\alpha, v_{j_\alpha} \ra}
  % \end{equation}
  \begin{equation}
    \label{eq:ap}
    \begin{split}
      a^p  & \le
      a_A^p \vol_d(D)
      + \sum_{k=0}^d \sum_{\jb\in J_k, \ib \in I^0_k} 2^p \vol_{d}(G_{\ib,\jb}) \\
      & \quad +
      \sum_{k=1}^d (2L_{k,1})^{d-k} \sum_{\jb \in J_k} \vol_{d-k} (G_{\jb})
      \sum_{\ib \in I_k^+} a_{\ib}^p \prod_{\alpha=1}^k \frac{\delta_{i_\alpha + 1} -
        \delta_{i_\alpha}}{\la \tilde f_\alpha, v_{j_\alpha} \ra}
    \end{split}
  \end{equation}
  by Proposition~\ref{prop:width-upperbound} with $\tilde{f}_\alpha \equiv
  \tilde{f}_{\jb,\alpha}$ defined there.
  % (Note that by choice of $a_{\ib} = 2$ when any of the $i_\alpha = 0$, the above inequality includes such terms.) %% True. but it's the next equality that's the issue.
  \begin{mynotes}
    \begin{notes}
      For the cases where at least one $i_\alpha=0$.  Note $\delta_1 = \epsilon^p$.  Thus, the entire ``inner boundary'' of $D$ (a ring of width $\delta_1$) is covered by one $[-1,1]$ bracket of $L_p^p$ size $\epsilon^p \Vol_{d-1}(\partial D)$.  Different ways to present this idea...

      NOTE: I still have a hunch/belief that the $\Vol_{d-1} \bd D$ term can be removed.  It is unfortunate in that it is a constant that grows with $d$ for the hyperrectangle and so I guess it does not need to be here. This may take an actual sort of induction proof as in Guntu and Sen.
    \end{notes}
  \end{mynotes}
  Recalling $\delta_1 = \epsilon^p$, note that
  \begin{equation}
    \label{eq:74}
    \sum_{k=0}^d \sum_{\jb\in J_k, \ib \in I^0_k} 2^p \vol_{d}(G_{\ib,\jb})
    \le 2^p \epsilon^p \vol_{d-1}(\bd D).
  \end{equation}
  \begin{mylongform}
    \begin{longform}
      (This is because $G_{\ib,\jb}$ are disjoint sets, and the sets in question are all within $\delta_1$ distance from $\bd D$.)
    \end{longform}
  \end{mylongform}
  Fixing $k
  \in \lb 1, \ldots d \rb$, we have
  \begin{align*}
    %% \MoveEqLeft
    \sum_{\jb \in J_k} \vol_{d-k} (G_{\jb})
    \sum_{\ib \in I_k^+} a_{\ib}^p \prod_{\alpha=1}^k \frac{ \delta_{i_\alpha + 1} -
    \delta_{i_\alpha}}{\la \tilde f_\alpha, v_{j_\alpha} \ra}
  %% \\
    & \le   \sum_{\jb \in J_k} \vol_{d-k} (G_{\jb}) L_{\jb,3}^k
      \sum_{i_1=1}^A \cdots \sum_{i_k=1}^A \prod_{\alpha=1}^k a_{i_{\alpha}}^p
      \delta_{i_\alpha + 1}  \\
    & \le   \sum_{\jb \in J_k} \vol_{d-k} (G_{\jb}) L_{\jb,3}^k
      \sum_{i_1=1}^A a_{i_1}^p \delta_{i_1+1} \cdots \sum_{i_k=1}^A  a_{i_{k}}^p
      \delta_{i_k + 1} .
  \end{align*}
  %% where $L_{\jb,3}$ is defined in \eqref{eq:12}.
  where $L_{\jb,3} := \max_{\alpha \in \lb 1,\ldots,k\rb} 1 / \la \tilde
  f_{\jb,\alpha}, v_{j_\alpha} \ra$.
  We have
  \begin{equation}
    \label{eq:defn:Au}
    \sum_{\alpha = 1}^A a_\alpha^p \delta_{\alpha+1}
    =    \epsilon^{p/k} \sum_{\alpha=1}^A \frac{ \epsilon^{1/k} \delta_{\alpha+1}}{
      \delta_\alpha a_\alpha}
    =    \epsilon^{p/k}  \sum_{\alpha=1}^A \zeta_\alpha^2
    =: \epsilon^{p/k} A_u,
  \end{equation}
  where $A_u \le  2 u^{1 / (p+1)^2}$ by
  Lemma~\ref{lem:zetasum}, below.
  \begin{mylongform}
    \begin{longform}
      To check the equality
      \begin{equation}
        \label{eq:63-tmp}
        \sum_{\alpha = 1}^A a_\alpha^p \delta_{\alpha+1}
        =    \epsilon^{p/k} \sum_{\alpha=1}^A \frac{ \epsilon^{1/k} \delta_{\alpha+1}}{
          \delta_\alpha a_\alpha}
      \end{equation}
      we do the following algebra.
      Note
      \begin{equation}
        \label{eq:abc-algebra}
        \frac{a_i^p}{\epsilon^{p/k}} =
        \exp \lb - p^2 \frac{(p+1)^{i-2}}{(p+2)^{i-1}} \log \epsilon \rb.
      \end{equation}
      And on the other hand
      \begin{align}
        \inv{\delta_i a_i}
        & = \exp \lb - p \frac{(p+1)^{i-1}}{(p+2)^{i-1}} \log \epsilon \rb
          \epsilon^{-1/k}
          \exp \lb p \frac{ (p+1)^{i-2}}{(p+2)^{i-1}} \log \epsilon \rb \\
        & = \epsilon^{-1/k}
          \exp \lb -p^2 \log \epsilon \frac{(p+1)^{i-2}}{ (p+2)^{i-1}} \rb.
          \label{eq:delta-a-algebra}
      \end{align}
      Together, the equality of \eqref{eq:abc-algebra} and \eqref{eq:delta-a-algebra} shows
      \eqref{eq:63-tmp}.

      It is also helpful (for thinking through the above computations sometimes) to remove the epsilon's.  Define $b_i := a_i / \epsilon^{1/k}$.  Then
      \eqref{eq:abc-algebra} equals $b_i^p$
      and 
      \eqref{eq:delta-a-algebra}
      shows that
      \begin{equation*}
        \inv{ \delta_i b_i} =           \exp \lb -p^2 \log \epsilon
        \frac{(p+1)^{i-2}}{ (p+2)^{i-1}} \rb,
      \end{equation*}
      so we have 
      \begin{equation}
        \label{eq:abc-b-algebra}
        b_\alpha^p = \inv{\delta_\alpha b_\alpha}.
      \end{equation}

    \end{longform}
  \end{mylongform}
  Thus
  % \eqref{eq:ap}  is bounded above by
  \begin{equation*}
    \sum_{\jb \in J_k} \vol_{d-k} (G_{\jb}) L_{\jb,3}^k
    \lp \sum_{i_1=0}^A a_{i_1}^p \delta_{i_1+1} \rp
    \cdots \lp \sum_{i_k=0}^A  a_{i_{k}}^p
    \delta_{i_k + 1} \rp
    \le
    \epsilon^{p} A_u^k \sum_{\jb \in J_k}  \vol_{d-k} ( G_{\jb} ) L_{\jb,3}^k,
  \end{equation*}
  so by \eqref{eq:ap}
  $a \le S_{D,s}^{2/d} \epsilon $
  %% where $S_k^D = \sum_{j \in J_k} \vol_{d-k} (G_{\jb}).$
  where
  \begin{equation}
    \label{eq:43}
    S_{D,s}^{2/d}
    := \Big(
    \frac{\vol_d(D) }{u^p}
    +  2^p \vol_{d-1} (\bd D)
    + \sum_{k=1}^d (2 L_{k,1}) ^{d-k} A_u^k \sum_{\jb \in J_k}
    \vol_{d-k} (G_{\jb}) L_{\jb,3}^k
    \Big)^{1/p}.
  \end{equation}

  We have thus bounded the bracketing entropy when $D \subset [0,1]^d$ and
  $B=1$.  Thus, by the scaling at the beginning of the proof, for any convex
  polytope $D \subset \prod_{i=1}^d [a_i,b_i]$ and any $B > 0$, we have shown for $0 < \epsilon \le  B \lp \prod_{i=1}^d b_i-a_i \rp^{1/p}$ that
  \begin{equation*}
    \log \Nb[\epsilon S_{\tilde{D},s}^{2/d}, \cC[D,B], L_p]
    \le S_{\tilde{D},c}
    \epsilon^{-d/2}
    \Bigg( B  \lp \prod_{i=1}^d (b_i-a_i) \rp^{1/p} \Bigg)^{d/2}
  \end{equation*}
  where $ S_{\tilde{D},c} $ equals
  \begin{equation}
    \label{eq:51}
    c_d ( u + 2d )^{d/2}
    %% + \sum_{k=1}^d \sum_{\jb \in J_k} \tilde{c}_{d,k} u^{kd/2(p+1)^2}
    + \sum_{k=1}^d \sum_{\jb \in J_k} L_{\jb,4}^{-1} \tilde{c}_{d,k} u^{kd/2(p+1)(p+2)}
    %% \lp 1 + \lp \frac{8}{u} \rp^{d-k} C_d \vol_{d-k}(G_{\jb}) L_{k,1}^{d-k} \rp^{d/2}.
    \lp 1 + 2^{d-k+2} L_{k,1} u^{-(d-k)} C_d \vol_{d-k}(G_{\jb})  \rp^{d/2}.
  \end{equation}
  Letting $\delta := S_{\tilde D,s}^{2/d} \epsilon$, we have shown that
  \begin{equation}
    \label{eq:22}
    \log \Nb[\delta, \cC[D,B], L_p]
    \le S_{\tilde D,c} S_{\tilde{D},s}
    \delta^{-d/2}
    \Bigg( B  \lp \prod_{i=1}^d (b_i-a_i) \rp^{1/p} \Bigg)^{d/2}
  \end{equation}
  for $0 < \delta \le S_{\tilde D,s}^{2/d} B \prod^d_{i=1} (b_i-a_i)^{1/p}$. 
  \begin{myold}
    \begin{old}
      We are thus done, except that to show \eqref{eq:4} we applied
      Lemma~\ref{lem:zetasum} which required $\epsilon \le 1$.  The result
      \eqref{eq:20} can be extended to allow any $\epsilon > 0$ just as in
      the proof of Theorem~\ref{thm:bronshtein-1976-convex-sets-cover}, at
      the cost of increasing the constant in \eqref{eq:20}.
    \end{old}
  \end{myold}

  Finally, we can extend from requiring
  $\delta \le S_{D,s} B \prod^d_{i=1} (b_i-a_i)^{1/p}$ to allowing any
  $\delta > 0$, just as in the proof of
  %% Theorem~\ref{thm:bronshtein-1976-convex-sets-cover},
  Theorem~2.1, 
  in \cite{DBLP:journals/corr/Doss15}
  (at the slight cost of
  increasing the constant on the right hand side of \eqref{eq:22}).
\end{proof}

Note that the constants $S_{\tilde{D},s}$ and $S_{\tilde{D}, c}$ should be
calculated using the rescaling of $D$ that lies in $[0,1]^d$, $\tilde{D}$.
The following lemma was used above.
\begin{mylongform}
  \begin{longform}

    The following summarizes the results together with the constants
    involved.

    We have shown that
    \begin{equation}
      \log \Nb[\delta, \cC[D,B], L_p]
      \le S_{\tilde D,c} S_{\tilde{D},s}
      \lp \frac{ B  \lp \prod_{i=1}^d (b_i-a_i) \rp^{1/p} }{    \delta }\rp^{d/2}
    \end{equation}
    for $0 < \delta \le S_{\tilde D,s}^{2/d} B \prod^d_{i=1} (b_i-a_i)^{1/p}$
    where
    \begin{equation*}
      S_{\tilde{D},c} :=
      xxx
    \end{equation*}
    and
    \begin{equation*}
      S_{D,s} :=
      xxx
    \end{equation*}

    For $k \in \lb 0, \ldots, d-1\rb$, let $d_{i,j,k} := d(E_i, F_j)$ where $E_i,$ $i=1,\ldots, N_k$, is a $k$-face (a face of dimension $k$) and $F_j$, $j=1,\ldots, N$, is a facet.  Then let
    \begin{equation}
      r_D :=
      \min \lb d_{i,j,k} : d_{i,j,k} \ne 0, k \in \lb 0,\ldots, d-1\rb \rb > 0.
    \end{equation}
    Let
    \begin{equation}
      \label{eq:defn:u-old}
      u \equiv u_D :=
      %% \exp \lp -2(p+1)^2  (p+2) \log 2 \rp
      r_D/2 \wedge
      2^{  -2(p+1)^2  (p+2) }
      \wedge
      %% \min_{k \in \lb 1,\ldots, d-1 \rb} \min_{ \jb \in J_k} w_r(G_{\jb})
      \min_{k \in \lb 1,\ldots, d-1 \rb} \min_{ \jb \in J_k} \rho_{\jb, \alpha}
    \end{equation}
    where $\rho_{\jb,\alpha} := w(G_{\jb},e_\alpha)$ (recall the definition
    of the width function $w$ in Subsection~\ref{subsec:defns-assms}).  Let
    $L_{d,1} = 1$ and For $k \in \lb 1, \ldots, d -1 \rb$, let
    \begin{equation}
      %% \label{eq:defn:Lk1}
      L_{k,1} :=
      1 \vee
      \max_{\jb \in J_k}
      \max_{\substack{\|e\| = 1\\ e \in E_{\jb}}}
      \max_{\substack{j \in \lb 1,\ldots, N\rb \setminus \jb ; \,  \la e, v_j \ra < 0 \\
          \la v_i,v_j \ra \ge 0 , \text{ some } i \in \jb }}
      \la -e, v_j \ra^{-1},
    \end{equation}
    and $L_{\jb,3} := \max_{\alpha \in \lb 1,\ldots,k\rb} 1 / \la \tilde
    f_{\jb,\alpha}, v_{j_\alpha} \ra$
    where
    $\tilde f_{\jb,\gamma}$ are defined in
    Proposition~\ref{prop:width-upperbound}.

  \end{longform}
\end{mylongform}

\begin{lemma}
  \label{lem:zetasum}
  For any $\gamma \ge 1$,
  %% $\epsilon \le \epsilon_0 \le 1$,
  $0 < \epsilon \le 1$,
  with $\zeta_i$
  %% $\delta_i$, and $a_i$
  given in \eqref{eq:23},
  %% as in the previous proof,
  %% $\zeta_i := \sqrt{ \epsilon^{1/k} \delta_{i+1} / (\delta_i a_i)}$,
  and  with $A$ and $u$ given by \eqref{eq:5} and \eqref{eq:15}, we have
  \begin{equation*}
    \sum_{\alpha = 1}^A \zeta_\alpha^\gamma
    \le     2 u^{\gamma / (2(p+1)^2)}.
  \end{equation*}
\end{lemma}
\begin{proof}
  Straightforward algebra shows
  \begin{equation}
    \label{eq:zeta-formula}
    \zeta_\alpha^2
    = \exp \lb p \frac{(p+1)^{\alpha-2}}{(p+2)^{\alpha}}
    \log \epsilon \rb.
  \end{equation}
  \begin{mylongform}
    \begin{longform}
      Here is the algebra for \eqref{eq:zeta-formula}. Recall $ \delta_i := \exp\lb p \lp \frac{p+1}{p+2}\rp^{i-1} \log \epsilon \rb.$ We have (see \eqref{eq:delta-a-algebra})
      \begin{align*}
        \zeta_\alpha^2
        = \epsilon^{1/k} \delta_{\alpha+1}/ \delta_\alpha a_\alpha
        &  =  \delta_{\alpha+1}
          \exp \lb -p^2 \frac{(p+1)^{\alpha-2}}{(p+2)^{\alpha-1}}
          \log \epsilon \rb \\
        & =\exp \lb
          \frac{ (p+1)^{\alpha-2} }{ (p+2)^{\alpha-2} }
          \lp p \frac{(p+1)^2}{(p+2)^2} - \frac{p^2}{p+2} \rp
          \log \epsilon \rb \\
        & =  \exp \lb
          \frac{ (p+1)^{\alpha-2} }{ (p+2)^{\alpha-2} }
          p
          \lp \frac{(p+1)^2 - p(p+2)}{(p+2)^2}  \rp
          \log \epsilon \rb \\
        & =  \exp \lb
          \frac{ (p+1)^{\alpha-2} }{ (p+2)^{\alpha-2} }
          p
          \frac{1}{(p+2)^2}
          \log \epsilon \rb.
      \end{align*}
    \end{longform}
  \end{mylongform}
  We have, for $\alpha = 1, \ldots, A-1$,
  \begin{equation*}
    \frac{\zeta_\alpha}{\zeta_{\alpha+1}}
    = \exp \lb \frac{ p \log \epsilon}{2
      (p+1)^2 (p+2)} \lp\frac{p+1}{p+2}\rp^{\alpha}\rb,
  \end{equation*}
  which is bounded above by
  % \begin{align*}
  %   \frac{\zeta_\alpha}{\zeta_{\alpha+1}}
  %   & = \exp \lb \frac{-p \log \epsilon}{2 (p+1)^2 (p+2)}
  %   \lp\frac{p+1}{p+2}\rp^{\alpha-1}\rb \\
  %   & \ge \exp \lb  \frac{ -p \log \epsilon} {2 (p+1)^2 (p+2) }
  %   \lp\frac{p+1}{p+2}\rp^{A-1} \rb \\
  %   & \ge \exp \lb  \frac{ -  \log u} {2 (p+1)^2 (p+2) }\rb
  %   %   =: R_u    .
  %   =: R    .
  % \end{align*}
  \begin{align*}
    \exp \lb  \frac{ p \log \epsilon} {2 (p+1)^2 (p+2) }
    \lp\frac{p+1}{p+2}\rp^{A-1} \rb 
    \le \exp \lb  \frac{   \log u} {2 (p+1)^2 (p+2) }\rb
    =: R^{-1}    .
  \end{align*}
  \begin{mylongform}
    \begin{longform}
      The algebra for the previous equality and inequality are as follows.  We have
      \begin{align*}
        \frac{\zeta_\alpha^2}{\zeta_{\alpha+1}^2}
        & = \exp \lb \lp \frac{(p+1)^{\alpha-2}}{(p+2)^\alpha} -
          \frac{(p+1)^{\alpha-1}}{(p+2)^{\alpha+1}} \rp p \log \epsilon \rb \\
        & = \exp \lb \lp 1 -  \frac{p+1}{p+2} \rp
          \frac{(p+1)^{\alpha-2}}{(p+2)^{\alpha}}
          p \log \epsilon \rb \\
        & = \exp \lb  p \frac{ (p+1)^{\alpha-2}}{ (p+2)^{\alpha+1}} \log \epsilon \rb.
      \end{align*}
    \end{longform}
  \end{mylongform}
  Then, $\zeta^\gamma_\alpha (R^\gamma - 1) \le \zeta_\alpha^\gamma R^\gamma - (R \zeta_{\alpha-1})^\gamma $ so $\zeta_\alpha^\gamma \le \lp R^\gamma / (R^\gamma - 1) \rp ( \zeta_\alpha^\gamma - \zeta_{\alpha-1}^\gamma )$ and thus
  \begin{equation}
    \label{eq:69}
    \sum_{\alpha = 1}^A \zeta_\alpha^\gamma
    \le \zeta_1^\gamma + \frac{R^\gamma}{R^\gamma-1} \sum_{\alpha=2}^A
    (\zeta_\alpha^\gamma - \zeta_{\alpha-1}^\gamma )
    = \zeta_1^\gamma + \frac{R^\gamma}{R^\gamma-1} (\zeta_A^\gamma -
    \zeta_1^\gamma)
    \le \frac{R^\gamma}{R^\gamma-1} \zeta_A^\gamma
  \end{equation}
  % Then $\zeta_{\alpha-1}^\gamma (R^\gamma - 1) \le \zeta_{\alpha-1}^\gamma R^\gamma - \zeta_\alpha^{\gamma} R^\gamma$ so $\zeta_{\alpha-1}^\gamma \le (\zeta_{\alpha-1}^\gamma - \zeta_\alpha^\gamma) R^\gamma / (R^\gamma - 1)$ and thus
  % \begin{equation*}
  %   \sum_{\alpha = 1}^A \zeta_\alpha^\gamma
  %   \le \zeta_1^\gamma + \frac{R^\gamma}{R^\gamma-1} \sum_{\alpha=2}^A
  %   (\zeta_{\alpha-1}^\gamma - \zeta_{\alpha}^\gamma )
  %   = \zeta_1^\gamma + \frac{R^\gamma}{R^\gamma-1} (\zeta_1^\gamma -
  %   \zeta_A^\gamma)
  %   \le \frac{ 2 R^\gamma-1 }{R^\gamma-1} \zeta_1^\gamma.
  % \end{equation*}
  and $\zeta_A^\gamma \le u^{\gamma / (2(p+1)(p+2))}$.  Since $u \le \exp\lp -2 (p+1)^2 (p+2) \log 2\rp$ by its definition \eqref{eq:defn:u}, $R \ge 2$ so
  $R^{\gamma } / (R^\gamma - 1) \le 2$ for any $\gamma \ge 1$.
  % $2 R^{\gamma } / (R^\gamma - 1) \le 1$ for any $\gamma \ge 1$.  Since $\zeta_\alpha \le 1$ for all $\alpha$ (by \eqref{eq:zeta-formula}), %% ie including $\alpha = 1$
  % we conclude that $\sum_{\alpha=1}^A \gamma_\alpha^\gamma \le 1$ for any $\gamma \ge 1$.
\end{proof}

For any convex $D$ and convex subset $\tilde{D} \subset D$, note that
$\cC[D,1]|_{\tilde{D}} \subset \cC[\tilde{D}, 1]$.
Thus by covering any convex polytope $D$ by simple polytopes $D_i \subset D$,
we can bound $\Nb[\epsilon, \cC[D,B], L_p]$ by applying
Theorem~\ref{thm:main-thm} repeatedly to $\cC[D_i,1]$ and using   \eqref{eq:bracketing-cover}.
A cover of $D$ can be attained by, for instance, subdividing $D$ into simple polytopes \citep{Lee:1997ww}, such as simplices.
The constant in the bound then
depends on the subdivision of $D$.
\begin{mynotes}
  \begin{notes}
    (Triangulation enforces more conditions than we need though; we do not
    need the nonempty intersection of $k$-faces of the simplices to be
    $k$-faces; we do not even need the interiors of the simplices to be
    disjoint.  (And we do not need simplices, any simple polytopes will do.)
  \end{notes}
\end{mynotes}
\begin{myold}
  \begin{old}
    By subdividing \citep{Lee:1997ww} any convex polytope $D$ into simple
    polytopes (e.g., triangulating into simplices, which are simple), we can
    extend our theorem to any polytope $D$.  The constant in the bound then
    depends on the subdivision of $D$.
  \end{old}
\end{myold}
\begin{corollary}
  \label{cor:general-convex-polytope}
  Fix $d\ge 1$ and $p \ge 1$.  Let $D \subseteq \prod_{i=1}^d [a_i,b_i]$ be
  any convex polytope. Then for
  %% some $\epsilon_0>0$ and for
  %% $0 < \epsilon \le   \epsilon_0 B \lp \prod_{i=1}^d b_i-a_i \rp^{1/p} $, %% what is upper bound
  $\epsilon > 0$,
  \begin{equation*}
    \log \Nb[\epsilon , \cC[D,B], L_p]
    \le
    C_{d,D}
    \epsilon^{-d/2}
    \Bigg(  B  \Bigg( \prod_{i=1}^d (b_i-a_i) \Bigg)^{1/p} \Bigg)^{d/2}.
  \end{equation*}
\end{corollary}
\begin{proof}
  By the same scaling argument as in the proof of Theorem~\ref{thm:main-thm}
  we may assume $[a_i,b_i] = [0,1]$ and $B=1$.  The $d=1$ case is given by
  \cite{MR2519658}. %%dryanov
  Any convex polytope $D$ can be triangulated into $d$-dimensional simplices
  (see e.g.\ \cite{Dey:1998cu}, \cite{Rothschild:1985jo}). We are done by
  applying Theorem~\ref{thm:main-thm} to each of those simplices, by
  \eqref{eq:bracketing-cover}.
\end{proof}

\section{Properties of $G_{\ib,\jb}$}
\label{sec:technical-details}

%%% \subsection{Inscribing  $G_{\ib,\jb}$ in a Hyperrectangle}

% Theorem~\ref{thm:GS-extension} shows that the bracketing entropy of
% $\cC[D,B,\Gb]$ depends on the diameters of the hyperrectangle $\prod_{i=1}^d
% [a_i,b_i]$ circumscribing $D$.
% This is part of why bounding entropies on hyperrectangular domains is
% more straightforward than on non-hyperrectangular domains.
In this section we
show how to embed the domains $G_{\ib,\jb}$, which partition $D$, into
hyperrectangles.  We used this in the proof of Theorem~\ref{thm:main-thm}
so we could apply Theorem~\ref{thm:GS-extension}.
Theorem~\ref{thm:GS-extension} says that the bracketing entropy of convex
functions on domain $D$ with Lipschitz constraints along directions $e_1,
\ldots, e_k$ depends on $w(D,e_i)$ (since that gives the maximum ``rise'' in
``rise over run'').  In our proof of Theorem~\ref{thm:main-thm} we
partitioned $D$ into sets related to parallelotopes.  Thus we will study these
parallelotopes.
\begin{mylongform}
  \begin{longform}
    Although it is not true that a (translate of a) vector of length
    $w(D,e_i)$ in the direction $e_i$ is contained in $D$, it is true that
    there are points $x,y \in D$ such that one can decompose $y - x$ as $\sum
    a_j + b_j $ where each $a_j = k_j e_i$ and each partial sum is still in
    $D$ (so that the Lipschitz constraints apply).
  \end{longform}
\end{mylongform}
We know the width of $G_{\ib,\jb}$ in the directions $v_{j_\alpha}$, which are
$\delta_{i_\alpha+1}-\delta_{i_\alpha}$, by definition.

\begin{mynotes}
  \begin{notes}
    The diameter of a set $D$ is $d(D) = \sup_{x,y \in D} \Vert x- y \Vert$.
    It is elementary that $w(D) \le d(D)$, since if $H_1$ and $H_2$ are
    parallel supporting hyperplanes of $D$ at the points $h_1$ and $h_2$,
    respectively, and orthogonal to $v$, then
    \begin{equation*}
      d(D)
      \ge \Vert h_1 - h_2 \Vert
      = \Vert h_1 - P_{H_1}(h_2) + P_{H_1}(h_2) - h_2 \Vert
      \ge \Vert P_{H_1}(h_2) - h_2 \Vert
      = w(D,v),
    \end{equation*}
    where $P_{H_1}(h_2)$ is the projection of $h_2$ onto $H_1$.

    On the other hand for any unit vector $e$,
    $w(D, e) \ge d(D,e) := \sup_{x \in D} |x + \SPAN \lb e\rb \cap D |$.
    Taking supremum over $e$ shows that $w(D) \ge d(D)$.

    Note that in general, one defines width $\tilde{w}(D)$ to be the {\em
      infimum} over unit vectors $e$!

    In which case we just have $\tilde{w}(D) \le d(D)$.

    False.
  \end{notes}
\end{mynotes}

\begin{myold}
  \begin{old}

    For an integer $\alpha$, let $\alpha!$ be the factorial of $\alpha$, $\alpha!
    := \alpha \cdot (\alpha-1) \cdots 2 \cdot 1$.  Recall from
    Subsection~\ref{subsec:defns-assms}) that $\diam(P)$ is the diameter of $P$.
    \begin{lemma}
      \label{lem:parallelotope-widths}
      Let $j $ be a positive integer.  Let $v_1, \ldots, v_j \in \RR^j$ be
      linearly independent. Let $d_i > 0$ for $i=1,\ldots, j$, and let $P$ be the
      parallelotope defined by having $w(P, v_i)=d_i$.  Then $P$ satisfies
      \begin{equation*}
        \diam(P)
        %% w(P)
        %% \le d(P)
        \le j! \max_{1 \le i \le j} d_i.
      \end{equation*}
    \end{lemma}
    \begin{proof}
      The proof is by induction.  The case $j = 1$ is trivial.  Now assume the
      statement holds for $j -1$ and we want to show it for $j$.  For any $x, y
      \in \partial P$ we can find a path $x= x_0, x_1, \ldots, x_n = y$ from $x$
      to $y$ such that $x_i$ and $x_{i+1}$ are elements of (the boundary of) the
      same facet of $P$. $P$ has $2j$ facets; if $n > j$, then we can find a path
      through the complementary $2j-n$ facets, so that we may assume $n \le j$.
      By the induction hypothesis, $\Vert x_{i+1} - x_i \Vert \le (j-1)! \max_{1
        \le i \le j} d_i$, since any $(j-1)$-dimensional facet is a parallelotope
      lying in a hyperplane with normal vector $v_i$ and widths $d_1, \ldots,
      d_{i-1}, d_{i+1}, \ldots d_j$.  Thus
      \begin{equation*}
        \Vert x - y \Vert
        \le \sum_{i=1}^n \Vert x_i - x_{i-1} \Vert
        \le n (j-1)! \max_{1 \le i \le j} d_i \le  j! \max_{1 \le i \le j } d_i,
      \end{equation*}
      as desired.
    \end{proof}

    This gives a bound on the width of $G_{\ib,\jb}$ in the direction of each
    basis vector $e_\alpha$, $\alpha=1,\ldots, k$, from
    Proposition~\ref{prop:basis-lipschitz}.
    \begin{mynotes}
      \begin{notes}
        Could write bound of $w(G_{\ib,\jb}, e_\alpha) \le \alpha!
        \delta_{i_\alpha +1}$ in the below, because to apply lemma 4.1 widths
        need to be in increasing order and I dont want write out verification of
        that step?
      \end{notes}
    \end{mynotes}

    \begin{mynotes}
      \begin{notes}
        The induction hypothesis here is incorrectly applied.  The faces do
        not satisfy the induction hypothesis (with the same $d_i$'s, from
        Lemma~\ref{lem:parallelotope-widths}).
      \end{notes}
    \end{mynotes}
    \begin{proposition}
      \label{prop:basis-widths}
      Let Assumption~\ref{assm:simple-polytope} hold for a convex polytope $D$.
      % Fix $k \in \lb 0,\ldots, d \rb$,
      Fix $k \in \lb 1,\ldots, d \rb$,
      $\ib \in I_k, \jb \in J_k$, and let
      $G_{\ib,\jb}$ be as in \eqref{eq:defn:Gi}.  Let ${\ve e}_{\ib,\jb} \equiv
      {\ve e} := (e_1,\ldots, e_d)$, with $e_\alpha \in \RR^d$, be the
      orthornormal basis from Proposition~\ref{prop:basis-lipschitz}.  Then
      \begin{equation*}
        w(G_{\ib,\jb}, e_\alpha) \le \alpha! (\delta_{i_\alpha+1} -
        \delta_{i_{\alpha}}).
      \end{equation*}
      for $\alpha = 1,\ldots, k$.
      % , and, trivially,
      % \begin{equation*}
      %   w(G_{\ib,\jb}, e_\alpha) \le w(D)
      % \end{equation*}
      % for $\alpha = k+1,\ldots, d$.
    \end{proposition}
    \begin{proof}
      % This shows that $w(G_{\ib,\jb}, e_\alpha) \le \alpha! \delta_{i_\alpha}$
      % (since
      % $\delta_{i_\alpha} = \max_{1 \le \beta \le \alpha} \delta_{i_\beta}$), as
      % follows.
      Let $\alpha \in \lb 1, \ldots, k \rb$ and let $w(G_{\ib,\jb}, e_\alpha)$ be
      given by the distance between the parallel supporting hyperplanes $H_1$ and
      $H_2$.  The distance between $H_1$ and $H_2$ is equal to the distance
      between $H_1 \cap A$ and $H_2 \cap A$ where $A$ is any linear subspace
      containing the normal vector of $H_1$ and $H_2$.  Thus, let $A = \SPAN \lb
      v_{j_1},\ldots, v_{j_\alpha} \rb \ni e_\alpha$.  $G_{\ib,\jb}$ is contained
      in a
      parallelotope, $G_{\ib,\jb} \subseteq \cap_{\beta=1}^k \tilde H_{j_\beta}$
      where $\tilde H_{j_\beta} = \lb x \in \RR^d : \delta_{i_\beta} \le \la x,
      v_{j_\beta} \ra \le \delta_{i_\beta+1} \rb$.  Let $P =
      \cap_{\beta=1}^\alpha \tilde H_{j_\beta} \cap \SPAN \lb v_{j_1},\ldots,
      v_{j_\alpha} \rb$.  Then $P$ is a parallelotope contained in the
      $\alpha$-dimensional vector space $V = \SPAN \lb v_{j_1},\ldots,
      v_{j_\alpha} \rb$ with widths $w(P, v_{j_\beta}) =
      \delta_{i_\beta+1}-\delta_{i_\beta}$, for $\beta=1,\ldots, \alpha$.  Thus
      we can apply Lemma~\ref{lem:parallelotope-widths} and conclude that
      \begin{equation*}
        w(G_{\ib,\jb}, e_\alpha)
        \le w(P, e_\alpha)
        \le \alpha! ( \delta_{i_\alpha+1} - \delta_{i_\alpha})
        \mbox{ for } \alpha\in \lb 1, \ldots, k\rb.
      \end{equation*}
      %% did i check that i got \max d_i correct here ... (ie because of the difference)
      For the first inequality, we use the fact that $w \lp \cap_{\beta =
        1}^\alpha \tilde H_{j_\beta}, e_\alpha \rp \ge w \lp \cap_{\beta=1}^k
      \tilde H_{j_\beta}, e_\alpha \rp$, and that $w \lp \cap_{\beta = 1}^\alpha
      \tilde H_{j_\beta}, e_\alpha \rp = w(P, e_\alpha)$ since the distance
      between any two supporting hyperplanes $H_1$ and $H_2$ of $
      \cap_{\beta=1}^k \tilde H_{j_\beta} $ is equal to the distance between $H_1
      \cap A$ and $H_2 \cap A$ where $A$ is any linear subspace containing the
      normal vector of $H_1$ and $H_2$.
      % Thus we can embed $G_{\ib,\jb} $ in a
      % rectangle with orthogonal axes $e_1, \ldots, e_d$ and with widths along
      % $e_1,\ldots, e_k$ bounded by $1!(\delta_{i_1+1}-\delta_{i_1}), \ldots,
      % k!(\delta_{i_k+1}-\delta_{i_k})$, and widths in the remaining directions
      % bounded by the width of $D$.
    \end{proof}
  \end{old}
\end{myold}

A polytope $P$ is a $d$-parallelotope if $P = \sum_{i=1}^d [a_i,b_i]$ for vectors $a_i,b_i \in \RR^d$, where for all $i$, $[a_i,b_i]$  is not parallel to  the affine hull of $[a_j,b_j]$ for any $j \ne i$ (\cite{Grunbaum:1967vq} page 56).
We will rely on the following representation for a $k$-dimensional parallelotope.
%% For sets $A$ and $B$, let $A+B = \lb a+b : a \in A, b \in B \rb.$
\begin{myold}
  \begin{old}
    \begin{lemma}
      \label{lem:parallelotope-vertex-representation}
      Let $k$ be a positive integer and let $P := \cap_{\beta = 1 }^k \tilde
      E_\beta$ be a parallelotope where $\tilde E_\beta := \lb x \in \RR^k :
      0 \le \la x, v_\beta \ra \le d_\beta \rb$ for $k$ linearly independent
      normal unit vectors $v_\beta$.
      Let $\tilde H_\beta^+ := \lb x \in \RR^k: \la x,
      v_\beta \ra = d_\beta\rb$.
      Let $ \tilde f_\beta$ be the unit vector lying in $\cap_{\gamma=1,
        \gamma \ne \beta}^k \tilde H_{\beta}^+$ with $\la \tilde f_\beta ,
      v_\beta \ra > 0$, for $\beta = 1, \ldots, k$.  Then $0$ is a vertex of
      $P$ and we can write
      \begin{equation*}
        P = \sum_{\beta=1}^k [0, f_\beta]
      \end{equation*}
      where $f_\beta := d_\beta \tilde f_\beta / \la \tilde f_\beta, v_\beta \ra$,
      $[0, f_\beta] = \lb \lambda f_\beta : \lambda \in [0,1] \rb$.
    \end{lemma}
  \end{old}
\end{myold}
\begin{lemma}
  \label{lem:parallelotope-vertex-representation}
  Let $k$ be a positive integer and let $P := \cap_{\beta = 1 }^k \tilde
  E_\beta$ be a parallelotope where $\tilde E_\beta := \lb x \in \RR^k : 0
  \le \la x, v_\beta \ra \le d_\beta \rb$ for $k$ linearly independent normal
  unit vectors $v_\beta$.  Let $H_{\beta}^0 := \lb x \in \RR^k : \la x,
  v_\beta \ra = 0 \rb$.
  % Let $\tilde H_\beta^+ := \lb x \in \RR^k: \la x,
  % v_\beta \ra = d_\beta\rb$.
  Let $ \tilde f_\beta$ be the unit vector lying in $\cap_{\gamma=1, \gamma
    \ne \beta}^k \tilde H_{\beta}^0$ with $\la \tilde f_\beta , v_\beta \ra >
  0$, for $\beta = 1, \ldots, k$.  Then $0$ is a vertex of $P$ and we can
  write
  \begin{equation*}
    P = \sum_{\beta=1}^k [0, f_\beta]
  \end{equation*}
  where $f_\beta := d_\beta \tilde f_\beta / \la \tilde f_\beta, v_\beta \ra$,
  $[0, f_\beta] = \lb \lambda f_\beta : \lambda \in [0,1] \rb$.
\end{lemma}

\begin{proof}
  %% Let $H_{\beta}^0 := \lb x \in \RR^k : \la x, v_\beta \ra = 0 \rb$.
  Since
  the vectors $v_\beta$ are unique, $\cap_{\beta = 1}^k H_\beta^0 = 0$ and
  the intersection of any $k-1$ of the hyperplanes $ H_\beta^0$ gives a
  $1$-dimensional space, $\SPAN \lb \tilde f_\beta \rb $.  A $k$-dimensional
  parallelotope can be written as the set-sum of the $k$ intervals emanating
  from the vertex, each given by the intersection of $k-1$ of the hyperplanes
  $H_\beta^0$.  See page 56 of \cite{Grunbaum:1967vq}.  Note that $f_\beta$
  satisfy $\la f_\beta, v_\beta \ra = d_\beta$ so that $f_\beta \in
  \tilde H_\beta^+ := \lb x \in \RR^k: \la x, v_\beta \ra = d_\beta\rb$; thus
  the $k$ intervals are given by $[0, f_\beta]$, $\beta = 1, \ldots, k$.
\end{proof}

Note the vector $\tilde{f}_\beta$ can be written as $(I-Q)v_\beta$ where $I$ is the identity projection in $\RR^k$ and $Q$ is the projection onto $\SPAN \lb v_1, \ldots, v_{\beta-1}, v_{\beta+1}, \ldots, v_k \rb$.
\begin{mylongform}
  \begin{longform}
    Thus $Q = X (X'X)^{-1}X'$ where $X = \lp v_1 | \ldots |  v_{\beta-1} |  v_{\beta+1} |
    \ldots | v_k \rp$.
  \end{longform}
\end{mylongform}
\begin{mylongform}
  \begin{longform}
    The bound \eqref{eq:vol-Gij} can be improved.  It is a hyperrectangle
    area instead of a parallelotope area; the latter could be much smaller.
  \end{longform}
\end{mylongform}
\begin{mynotes}
  \begin{notes}
    This is different than the original material on embedding a parallelotope
    into a hyperrectangle in that the orientation here can be arbitrary.
    Originally it was a fixed orientation.
  \end{notes}
\end{mynotes}
The next proposition uses Lemma~\ref{lem:parallelotope-vertex-representation}
to bound the widths of $G_{\ib,\jb}$, in certain directions, in terms of the width of $G_{\jb}$ in those directions.
We will need the following constant (depending on $D$).  For $k \in \lb 1, \ldots, d -1 \rb$, let
\begin{equation}
  \label{eq:defn:Lk1}
  L_{k,1} :=
  1 \vee
  \max_{\jb \in J_k} \,
  \max_{\substack{\|e\| = 1\\ e \in E_{\jb}}}  \
  \max_{\substack{j \in \lb 1,\ldots, N\rb \setminus \jb ; \,  \la e, v_j \ra < 0 \\  \la v_i,v_j \ra > 0 , \text{ some } i \in \jb }}
  \la -e, v_j \ra^{-1},
\end{equation}
where $E_{\jb} := \SPAN \lb e_{\jb, k+1}, \ldots, e_{\jb, d}\rb$ from Proposition~\ref{prop:basis-lipschitz}, and we abuse notation as convenient to treat $\jb$ as if it were a set rather than a vector.  We also (arbitrarily) define $L_{d,1}:= 1$, for ease of presentation later on.
\begin{proposition}
  \label{prop:width-upperbound}
  For each $k \in \lb 1,\ldots, d-1 \rb$,
  %% For each $k \in \lb 1,\ldots, d \rb$,
  $\ib \in I_k, \jb \in J_k$, and each $G_{\ib,\jb}$, %% defined by $u$
  %% in \eqref{eq:5},
  and the basis ${\ve e} \equiv {\ve e}_{\ib,\jb}$ from Proposition~\ref{prop:basis-lipschitz}, for
  $\alpha \in \lb k+1, \ldots, d \rb$, we have
  \begin{equation}
    \label{eq:9}
    w(G_{\ib,\jb}, e_\alpha) \le 2 L_{k,1} w(G_{\jb}, e_\alpha).
  \end{equation}
  Then for $k \in \lb 1, \ldots, d\rb$, let $\tilde f_\alpha \equiv
  \tilde{f}_{\jb, \alpha}$ be the unit vector with $\la \tilde f_\alpha ,
  v_{j_\alpha} \ra > 0$ lying in $ \SPAN \lb v_{j_1}, \ldots, v_{j_k} \rb
  \cap \lp \cap_{\gamma=1,\gamma\ne \alpha}^k H^0_{j_\gamma} \rp $,
  $\alpha=1,\ldots,k$, where $H_{j_\gamma}^0 := \lb y \in \RR^d : \la y,
  v_{j_\gamma} \ra = 0 \rb$.  Then for $k \in \lb 1, \ldots, d \rb$, we have
  \begin{equation}
    \label{eq:vol-Gij}
    \vol_{d} (G_{\ib,\jb})
    \le  (2 L_{k,1})^{d-k} \vol_{d-k}\lp G_{\jb}\rp
    \cdot
    \prod_{\alpha = 1}^k \frac{\delta_{i_\alpha+1}-\delta_{i_\alpha}}{\la
      \tilde f_\alpha, v_{j_\alpha} \ra }
  \end{equation}
  where $L_{k,1}$ is given by \eqref{eq:defn:Lk1} for $k \in \lb 1, \ldots, d
  -1 \rb$ (and we set  $L_{d,1}:= 1$ arbitrarily).
  % $H^+_\gamma := \lb y \in \RR^d : \la y,
  % v_{j_\gamma} \ra = \delta_{i_\gamma+1}-\delta_{i_\gamma} \rb$.
\end{proposition}
\begin{proof}
  Fix $k \in \lb 1, \ldots, d-1\rb$.  Let $x \equiv x_{\jb} \in G_{\jb}$
  (from \eqref{eq:FritzJohn}).
  % Let $x$ be an arbitrary fixed
  % point %%
  % in $G_{\jb}$  %%
  % which we take to be $x \equiv x_{\jb}$ (from
  % \eqref{eq:FritzJohn}) for definiteness.
  \begin{mynotes}
    \begin{notes}
      Choice of $x_{\jb}$ is entirely arbitrary.  Could fix any $x$.  Only
      point is to have less notation.  (Constants might modify; see defn of
      $u$ involving $x_{\jb}$).
    \end{notes}
  \end{mynotes}
  Let $f_{j_\gamma}$ be as given in \eqref{eq:defn:fjgamma}.
  % Let $\tilde{f}_{j_\gamma} $ be given by
  % Lemma~\ref{lem:parallelotope-vertex-representation} applied to the $k$
  % linearly independent unit normal vectors $v_{j_1}, \ldots, v_{j_k}$, and
  % (as in that lemma, with ``$d_\beta$'' given by $
  % (\delta_{i_\gamma+1}-\delta_{i_\gamma})$), let
  % \begin{equation*}
  %   %%   \label{eq:28}
  %   f_{\ib,\jb,j_\gamma}
  %   \equiv f_{j_\gamma} := (\delta_{i_\gamma+1}-\delta_{i_\gamma})
  %   \tilde{f}_{j_\gamma} / \la \tilde{f}_{j_\gamma}, v_{j_\gamma} \ra.
  % \end{equation*}
  Let  $P_{\ib,\jb} := \sum_{\gamma=1}^k [0,f_{j_\gamma}]$.
  % where for sets $A$
  % and $B$, we let $A+B = \lb a+b : a \in A, b \in B \rb$ and $[0,v]:= \lb
  % \lambda v : \lambda \in [0,1] \rb$.
  We will show that $G_{\ib,\jb}$ is
  contained in the set-sum of a hyperrectangle and $P_{\ib,\jb}.$ To begin
  with let $ G_{\ib,\jb} \ni z = x + \sum_{\gamma=1}^k f^*_{j_\gamma}$ where
  $f^*_{j_\gamma} = d_{j_\gamma} \tilde f_{j_\gamma} $ where
  \begin{equation}
    \label{eq:8}
    0 \le d_{j_\gamma} \le
    (\delta_{i_\gamma+1}-\delta_{i_\gamma}) / \la \tilde f_{j_\gamma},
    v_{j_\gamma} \ra
    \le u / \la \tilde{f}_{j_\gamma} , v_{j_\gamma} \ra.
  \end{equation}
  % and $\tilde f_{j_\gamma} $ is given by
  % Lemma~\ref{lem:parallelotope-vertex-representation} for the $k$ linearly
  % independent unit normal vectors $v_{j_1}, \ldots, v_{j_k}$.
  Take an
  arbitrary $e \in \SPAN \lb e_{k+1}, \ldots, e_d \rb$ with $\| e\| = 1$.
  Let $\lambda_{z,e} := d^+(z, \partial G_{\ib,\jb}, e)$ and let $j$ give the
  corresponding facet of $G_{\ib,\jb}$ that $x + \lambda_{z,e} e$ hits, so
  that $\la z + \lambda_{z,e} e, v_j \ra = p_j + u$ for some $j \notin \jb$
  (abusing notation to treat $\jb$ as if were a set rather than a vector).
  Note that this means
  \begin{equation}
    \label{eq:19}
    \la e, v_j \ra < 0 .
  \end{equation}
  If $\la \sum_{\gamma=1}^k f^*_{j_\gamma} , v_j \ra \le 0 $ then
  \begin{equation}
    \label{eq:18}
    \lambda_{z,e}
    %% \le \lambda_{x,e}
    \le     d^+(x, \partial G_{\ib,\jb}, e).
    %% \le     d^+(x, \relbd G_{\jb}, e).
  \end{equation}
  Thus if \eqref{eq:18} does not hold then
  $\la f^*_{j_\gamma}, v_j \ra > 0$ for some $\gamma \in \lb 1, \ldots, k\rb$,
  so
  $\la v_{j_\alpha}  , v_j \ra > 0$ for some $\alpha \in  \lb 1, \ldots, k \rb$.
  \begin{mynotes}
    \begin{notes}
      (This may need better clarification: $f^*_{j_\gamma} $ lies in the ``positive span'' )
    \end{notes}
  \end{mynotes}
  Now,
  since $\la z + \lambda_{z,e} e, v_j \ra = p_j + u$, we have
  \begin{equation*}
    \lambda_{z,e} = \frac{ p_j + u - \la z, v_j \ra}{\la e, v_j \ra}
    \le \frac{ \la x, v_j \ra - p_j +
      u \sum_{\gamma=1}^k \frac{\la\tilde{f}_\gamma, v_j\ra}{\la \tilde{f}_\gamma, v_{j_\gamma} \ra } }{ \la -e, v_j \ra}
  \end{equation*}
  \begin{mynotes}
    \begin{notes}
      (Random thought to triple check: do I need absolute value signs anywhere in definition of $L_{k,2}$ or is this correct as-is?)
    \end{notes}
  \end{mynotes}
  Now
  \begin{equation*}
    \la x, v_{j}
    \ra - p_{j} \le d(x, H_{j}) \le d^+(x, \relbd G_{\jb},
    e)
  \end{equation*}
  since $H_{j}$ is the closest hyperplane to $x$ in the
  direction $e$.
  Recall the definition of $L_{k,1}$ in \eqref{eq:defn:Lk1}.
  Now, by \eqref{eq:defn:u} and the definition of $L_{k,2}$ \eqref{eq:defn:Lk2},
  we have shown
  \begin{equation}
    \label{eq100}
    \lambda_{z,e} \le 2  L_{k,1}  d^+(x, \relbd G_{\jb}, e) ,
  \end{equation}
  by \eqref{eq:19} and \eqref{eq:18}.
  This means that
  \begin{equation*}
    (G_{\ib,\jb} - z) \cap \SPAN \lb e_{k+1}, \ldots, e_d \rb \subset
    2 L_{k,1} \lp G_{\jb} - x \rp
  \end{equation*}
  so we can conclude that $w(G_{\ib,\jb} - z, e_\alpha) \le 2 L_{k,1}
  w(G_{\jb},e_\alpha)$ and $w(G_{\ib,\jb}, e_\alpha) \le 2 L_{k,1}
  w(G_{\jb},e_\alpha)$ since $\la z, e_\alpha \ra =0 $ for all
  $d_{j_\gamma}$ given by the range \eqref{eq:8}, $\alpha=k+1,\ldots, d$,
  for $k =1, \ldots, d-1$.

  %% Let $A + B := \lb a + b : a \in A, b \in B \rb$ for sets $A,B$ and
  Let $\rho_{\jb, \alpha } := w(G_{\jb}, e_\alpha)$. Then let
  \begin{equation}
    \label{eq:defn:R}
    R_{\ib,\jb} :=
    P_{\ib,\jb}
    + \sum_{\alpha = k+1}^d \ls - 2 L_{k,1} \rho_{\jb,\alpha} e_\alpha,
    2 L_{k,1} \rho_{\jb,\alpha} e_\alpha \rs.
  \end{equation}
  Then for any $x \in G_{\ib,\jb}$ such that $\la x, v_{j_\alpha } \ra =
  p_{j_\alpha} + \delta_{i_\alpha}$ for $\alpha \in \lb 1,\ldots, k \rb$, we
  have shown
  \begin{equation}
    \label{eq:Gij-subset-Rij}
    G_{\ib,\jb} \subset
    x + R_{\ib,\jb}.
  \end{equation}
  It then also follows that
  \begin{mynotes}
    \begin{notes}
      (Does this need a factor of $d$? volume of $G_{\jb}$ is smaller than
      product of widths.)
    \end{notes}
  \end{mynotes}
  \begin{equation}
    \label{eq:volGij-volGj-1}
    \vol_{d} (G_{\ib,\jb})
    \le  (2 L_{k,1})^{d-k} \vol_{d-k}\lp G_{\jb}\rp
    \cdot \vol_k\lp \sum_{\alpha=1}^k \ls 0 , f_{j_\alpha} \rs \rp.
  \end{equation}
  Since of parallelotopes
  with given axis lengths, the one with largest volume is the hyperrectangle,
  $ \vol_k \lp \sum_{\alpha=1}^k \ls 0 , f_{j_\alpha} \rs \rp \le \prod_{\alpha =
    1}^k \frac{\delta_{i_\alpha+1}-\delta_{i_\alpha}}{\la \tilde f_{j_\alpha},
    v_{j_\alpha} \ra }$, and so we have shown \eqref{eq:vol-Gij} (with this
  bound on $\vol_k \lp \sum_{\alpha=1}^k \ls 0 , f_{j_\alpha} \rs \rp $ being all
  that is needed in the $k=d$ case).
\end{proof}
\begin{myold}
  \begin{old}

    Statement of this lemma from a previous version:
    \begin{lemma}
      % \label{lem:Gij-cover-D}
      %% Let $D$ be a polytope and
      Let Assumption~\ref{assm:simple-polytope} hold and
      let $G_{\ib,\jb}$ be as in \eqref{eq:defn:Gi}.
      If $x
      \in G_{\ib,\jb}$ then $d(x, G_{\jb}) \le K_D u$ where
      $u$ is as in
      \eqref{eq:defn:u}.  In particular, if $G_{\jb} =
      \emptyset$, then $G_{\ib,\jb} = \emptyset$ (using the convention $d(x,
      \emptyset) = \infty$).
    \end{lemma}

    \begin{lemma}
      \label{lem:Gij-cover-D}
      %% Let $D$ be a polytope and
      Let Assumption~\ref{assm:simple-polytope} hold and let $G_{\ib,\jb}$ be
      as in \eqref{eq:defn:Gi}.  If $G_{\jb} = \emptyset$, then $G_{\ib,\jb} =
      \emptyset$.
    \end{lemma}
    \begin{proof}
      This follows from Proposition~\ref{prop:width-upperbound} and its proof.
    \end{proof}

    The above provides a hyperrectangle containing $G_{\ib,\jb}$. Let $A + B :=
    \lb a + b : a \in A, b \in B \rb$ for sets $A,B$.
    Let $\rho_{\jb, \alpha } := w(G_{\jb}, e_\alpha)$ and then
    let
    \begin{equation}
      \label{eq:defn:R}
      R_{\ib,\jb} :=
      \sum_{\alpha = 1}^k \ls \alpha! (\delta_{i_\alpha +1 }-\delta_{i_\alpha}) e_\alpha, \alpha!
      (\delta_{i_\alpha +1} - \delta_{i_\alpha}) e_\alpha \rs
      %% + \sum_{\alpha = k+1}^d \ls -w(D) e_\alpha, w(D) e_\alpha \rs
      + \sum_{\alpha = k+1}^d \ls - 2 L_{k,1} \rho_{\jb,\alpha} e_\alpha,
      2 L_{k,1} \rho_{\jb,\alpha} e_\alpha \rs.
    \end{equation}
    Then, for any $x \in
    G_{\ib,\jb}$ we have shown
    \begin{equation}
      \label{eq:Gij-subset-Rij}
      G_{\ib,\jb} \subset
      x + R_{\ib,\jb}.
    \end{equation}

  \end{old}
\end{myold}

% The above allows us to embed $G_{\ib,\jb}$ in a polytope which is a parallelotope times a hyperrectangle. Let $A + B := \lb a + b : a \in A, b \in B \rb$ for sets $A,B$.  Let $\rho_{\jb, \alpha } := w(G_{\jb}, e_\alpha)$ and then let
% \begin{equation}
%   \label{eq:defn:R}
%   R_{\ib,\jb} :=
%   P_{\ib,\jb}
%   + \sum_{\alpha = k+1}^d \ls - 2 L_{k,1} \rho_{\jb,\alpha} e_\alpha,
%   2 L_{k,1} \rho_{\jb,\alpha} e_\alpha \rs.
% \end{equation}
% Then, for any $x \in    G_{\jb}$
% we have shown
% \begin{equation}
%   \label{eq:Gij-subset-Rij}
%   G_{\ib,\jb} \subset
%   x + R_{\ib,\jb}.
% \end{equation}

% The following lemma shows that $R_{\ib,\jb}$, defined in the proof above, satisfies the conditions of Theorem~\ref{thm:GS-extension-v2_directional-version}.

The previous proposition controls the width and volume of $G_{\ib,\jb}$ in directions lying in $\SPAN \lb G_{\jb} \rb$.  Next we control width, volume, and also distance to $\bd D$ in directions perpendicular to $\SPAN \lb G_{\jb} \rb$.
\begin{lemma}
  \label{lem:embed-parallelotope-hyperrectangle}
  Let $P := \sum_{\alpha = 1}^k [0, f_\alpha]$ be a parallelotope in $\RR^k$
  where $f_1, \ldots, f_k$ are $k$ linearly independent vectors.  Then there
  exists an orthonormal basis of $\RR^k$, $e_1, \ldots, e_k \in \RR^k$ and
  $\gamma_1, \ldots, \gamma_k \in \RR$, such that
  \begin{equation}
    \label{eq:40}
    P \subset \sum_{\alpha=1}^ k [0, \gamma_\alpha e_\alpha]
    \qquad    \text{ where } \qquad
    |\gamma_\alpha| \le k \diam( P, e_\alpha ).
  \end{equation}
  %% where
  % \begin{equation}
  %   \label{eq:46}
  %   |\gamma_\alpha| \le k \diam( P, e_\alpha ).
  % \end{equation}
\end{lemma}
\begin{proof}
  We will construct a permutation $\pi$ of $\lb 1, \ldots, k\rb$ and
  inductively define $e_1, \ldots, e_k$ based on the sequence
  $f_{\pi(1)}, \ldots, f_{\pi(k)}$.  Let
  $e_1 := f_{\pi(1)}/ \| f_{\pi(1)} \|$ where $\| f_{\pi(1)} \|$ is maximal over $\lb \| f_\alpha \| \rb_{\alpha=1}^k$.
  Now let $Q_{j-1}$ be the projection of $\RR^k$ onto
  $\SPAN \lb e_1, \ldots, e_{j-1} \rb$ and let $Q_{j-1}^\perp $ be the projection
  onto $\SPAN \lb e_1, \ldots, e_{j-1} \rb^\perp$.  Then let
  $e_j := Q_{j-1}^\perp f_{\pi(j)} / \| Q_{j-1}^\perp f_{\pi(j)} \|$ where
  $\pi(j) \in \lb 1, \ldots, k \rb \setminus \lb \pi(1), \ldots, \pi(j-1)
  \rb$ is defined so that $\| Q^\perp_{j-1} f_{\pi(j)} \|$ is maximal.
  \begin{mylongform}
    \begin{longform}
      (Is maximal over
      $\pi(j) \in \lb 1, \ldots, k \rb \setminus \lb \pi(1), \ldots, \pi(j-1)
      \rb$.)
    \end{longform}
  \end{mylongform}

  Let $P_j := \sum_{\alpha=1}^j [0, f_{\pi(j)}]$.
  Now, $\diam(P_j, e_\alpha)$ is given by the value of $\la x-y, e_\alpha \ra$ such that $x, y \in P_j$ and $\la x-y, e_\alpha \ra$ is maximal.  Since $f_{\pi(j)} \notin \SPAN \lb f_{\pi(i)} \rb_{i \ne j} $, we also have that $e_j \notin \SPAN \lb f_{\pi(i)} \rb_{ i \ne j}$.  Thus for
  $i \ge j$, $\diam(P_i, e_j) \le \la f_{\pi(j)} , e_j \ra$ and so in fact
  $\diam(P_i, e_j) = \diam(P, e_j) = \la f_{\pi(j)}, e_j \ra$.
  \begin{mynotes}
    \begin{notes}
      For proof of equality I think one has to show that the projection $Q_j
      f_{\pi(j)}$ lies in $P_{j-1}$.  This is based on the definition of $e_j$
      and may require assuming that each $\la f_i, f_j \ra$ is nonnegative, by
      reorienting the parallelotope.  (A parallelotope has some corner where
      angles are all acute?)
    \end{notes}
  \end{mynotes}

  Now we prove by induction that
  \begin{equation}
    \label{eq:47}
    P_j \subset \sum_{\alpha=1}^j [ 0, \gamma_{j,\alpha} e_\alpha ]
  \end{equation}
  where $0 \le \gamma_{j,\alpha} \le j \diam( P_j, e_\alpha) = j \diam(P,e_\alpha)$.
  The statement is immediate for $j=1$.  Thus let $1 < j \le k$ and assume the induction hypothesis holds for $j-1$.
  Then for $1 < i \le j$
  \begin{equation}
    \label{eq:49}
    \lv    \la e_i, f_{\pi(j)} \ra \rv
    \le \| Q_i^{\perp} f_{\pi(j)} \|
    \le \| Q_i^\perp f_{\pi(i)} \|
    = \lv \la e_i, f_{\pi(i)} \ra \rv
    = \diam(P, e_i)
  \end{equation}
  where the first inequality is because $e_i \in \SPAN \lb e_1 ,\ldots, e_{i-1} \rb^\perp$,
  and the  next inequality and equality are by the definition of $e_i$.
  \begin{mylongform}
    \begin{longform}
      And the last equality is stated in the previous paragraph.  One way to
      see the first equality is because
      $$ \| Q_i^\perp f_{\pi(j)} \| \ge
      \| Q_i^\perp f_{\pi(j)} \| \| e_i \| \ge
      \la Q_i^\perp f_{\pi(j)} , e_i \ra = \la f_{\pi(j)}, Q_i^\perp e_i \ra
      = \la f_{\pi(j)}, e_i \ra$$ where the second to last equality is a
      property of orthogonal projections.
    \end{longform}
  \end{mylongform}
  Also, \eqref{eq:49} is immediately verifiable for $i=1$.

  Now, we can write
  \begin{mylongform}
    \begin{longform}
      (since  $e_j$ is defined by an orthogonal projection of $f_{\pi(j)}$),
    \end{longform}
  \end{mylongform}
  \begin{equation}
    \label{eq:48}
    f_{\pi(j)} = \lambda_{j,1} e_{1} + \cdots + \lambda_{j,j} e_j
  \end{equation}
  where $| \lambda_{j,i} | \le \diam(P,e_i) $
  by \eqref{eq:49}.
  For any $ x \in P_j = P_{j-1} + [0, f_{\pi(j)}]$, we can write
  \begin{equation}
    \label{eq:50}
    x = \sum_{\alpha=1}^{j-1} \eta_{j-1, \alpha} e_\alpha
    + \eta f_{\pi(j)}
  \end{equation}
  where $0 \le \eta \le 1$ and
  $| \eta_{j-1, \alpha} | \le (j-1) \diam(P, e_\alpha)$ by the induction hypothesis.
  Thus
  \eqref{eq:50} equals
  \begin{align*}
    \sum_{\alpha=1}^{j-1} \lp \eta_{j-1, \alpha} + \eta \lambda_{j,\alpha} \rp e_\alpha
    + \eta \lambda_{j,j} e_j,
  \end{align*}
  and $| \eta_{j-1, \alpha} + \eta \lambda_{j,\alpha} |\le (j-1) \diam(P,
  e_\alpha) + \diam(P,e_\alpha)$ for $\alpha \le j-1$ and $| \lambda_{j,j}|
  = \diam(P, e_j)$, so the induction hypothesis is shown.
\end{proof}
\begin{myold}
  \begin{old}
    \begin{mynotes}
      \begin{notes}
        The following is a previous version of the parallelotope embedding.  This lemma is almost true except  for the statement that $\gamma_\alpha \le k \diam(P, e_\alpha )$.
      \end{notes}
    \end{mynotes}
    \begin{lemma}
      \label{lem:embed-parallelotope-hyperrectangle-old}
      Let $P := \sum_{\alpha = 1}^k [0, f_\alpha]$ be a parallelotope in $\RR^k$
      where $f_1, \ldots, f_k$ are $k$ linearly independent vectors.  Then there
      exists an orthonormal basis of $\RR^k$, $e_1, \ldots, e_k \in \RR^k$, such
      that
      \begin{equation}
        \label{eq:38}
        \diam(P, e_\alpha) \le \| f_{\pi(\alpha)} \|
      \end{equation}
      and
      \begin{equation}
        \label{eq:29}
        \sum_{\alpha=1}^k [0, f_{\alpha}] \subseteq \sum_{\alpha=1}^k [0, \gamma_\alpha e_\alpha]
      \end{equation}
      where
      \begin{equation}
        \label{eq:41}
        0 \le \gamma_\alpha \le k \diam(P, e_\alpha) \le k \| f_{\pi(\alpha)} \|
      \end{equation}
      for some permutation $\pi$ of $\lb 1, \ldots, k \rb$.
    \end{lemma}
    \begin{proof}
      Let $\tilde{f}_\alpha := f_\alpha / \| f_\alpha \|$, $\alpha \in \lb 1, \ldots ,k \rb$, and
      let $\pi(1), \ldots, \pi(k)$ be a permutation of $(1, \ldots, k )$ such that
      \begin{equation}
        \label{eq:30}
        \| f_{\pi(1)} \|
        \ge \cdots \ge \| f_{\pi(k)} \|
      \end{equation}
      (if there are ties then any order is acceptable).  Define $e_1 :=
      \tilde{f}_{\pi(1)}$ and for $1 < j \le k$, let $\lambda_j w_j$ be the
      orthogonal projection of $\tilde{f}_{\pi(j)}$ onto $\SPAN \lb e_1, \ldots
      e_{j-1} \rb$, where $\| w_j \| = 1$ and $1 \ge \lambda_j \ge 0$.  Then let $e_j$ be the unit vector
      defined as being in the span of the orthogonal projection of $\tilde{f}_{\pi(j)}$ onto the
      complementary vector space, meaning that
      \begin{equation}
        \label{eq:31}
        \tilde{f}_{\pi(j)} =: \lambda_j w_j +
        %% \sqrt{\| \tilde{f}_{\pi(j)}\|^2 -      \lambda_j^2 } e_j.
        \sqrt{ 1 -      \lambda_j^2 } e_j.
      \end{equation}
      %% and $\| e_j \| = 1$.
      Note we can write
      $\sum_{\alpha = 1}^{j-1} \lambda_{j,\alpha} e_\alpha = \lambda_j w_j$ for
      real numbers $\lambda_{j,1},\ldots, \lambda_{j,j-1}$ satifying
      $\sum_{\alpha=1}^{j-1} \lambda_{j,\alpha}^2 = \lambda_j^2$.

      We will now
      prove by induction that
      \begin{equation}
        \label{eq:32}
        \sum_{\alpha=1}^j [0, f_{\pi(\alpha)}] \subseteq
        \sum_{\alpha=1}^j [0, \gamma_\alpha e_\alpha ]
      \end{equation}
      where $0 \le \gamma_\alpha \le (j - \alpha +1) \| f_{\pi(\alpha)} \| \le j
      \|f_{\pi(\alpha)}\|$, for $\alpha \in \lb 1, \ldots, j \rb$, $j \in \lb 1,
      \ldots, k\rb$; this will prove the lemma.
      The statement is immediate when $j=1$.  Thus let $1 < j \le k$ and assume
      the induction hypothesis holds for $j-1$.  Now for any
      $x \in \sum_{\alpha=1}^j [0, f_{\pi(\alpha)}] = \sum_{\alpha=1}^{j-1} [0,
      f_{\pi(\alpha)}] + [0, f_{\pi(j)}]$ we can write $x$ as
      \begin{equation}
        \label{eq:33}
        x = \eta_{x,1} e_x + \eta_{x,2} \tilde{f}_{\pi(j)}
      \end{equation}
      where $0 \le \eta_{x,2} \le
      \| f_{\pi(j)} \|$ and $\| e_x \| = 1$.
      %% and $e_x \in \SPAN \lb  \sum_{\alpha=1}^{j-1} [0 ,e_\alpha ] \rb$.
      Here we take $\eta_{x,1} e_x \in \sum_{\alpha=1}^{j-1} [0,
      f_{\pi(\alpha)}]$ and so by the induction hypothesis $\eta_{x,1} e_x =
      \sum_{\alpha=1}^{j-1} \gamma^*_{\alpha} e_\alpha$ where
      %% $0 \le  \gamma^*_\alpha \le \gamma_\alpha \le (j - \pi(\alpha) + 1) \|  f_{\pi(\alpha)} \|$.
      $0 \le  \gamma^*_\alpha \le \gamma_\alpha \le (j - \alpha ) \|  f_{\pi(\alpha)} \|$.
      \begin{mylongform}
        \begin{longform}
          (since $j-1 - \alpha + 1 = j-\alpha$.)
        \end{longform}
      \end{mylongform}
      By \eqref{eq:31} and \eqref{eq:33},
      we have
      \begin{equation}
        \label{eq:34}
        x = \sum_{\alpha=1}^{j-1} \lp \gamma_\alpha^* + \eta_{x,2} \lambda_{j,\alpha} \rp e_\alpha
        + \eta_{x,2} \sqrt{ 1- \lambda_j^2} e_{j}.
      \end{equation}
      Then, for $\alpha \in \lb 1, \ldots, j-1 \rb$, since  $0 \le \lambda_{j,\alpha} \le 1$,
      \begin{equation}
        \label{eq:35}
        0 \le \gamma_\alpha^* + \eta_{x,2} \lambda_{j,\alpha}
        \le  (j- \alpha) \| f_{\pi(\alpha)} \| + \| f_{\pi(j)} \|
        \le  (j- \alpha + 1) \| f_{\pi(\alpha)} \| ,
      \end{equation}
      by \eqref{eq:30}.  Since $\eta_{x,2} \sqrt{ 1- \lambda_j^2} \le \|
      f_{\pi(j)} \|$, we have shown \eqref{eq:29}.

      That \eqref{eq:38} holds follows immediately.  This is because
      $\diam(P, e_j)$ is given by the value of $\la x-y, e_j \ra$ such that $y$
      and $x $ are both in $\sum_{\alpha=1}^k [0, f_\alpha]$ and
      $\la x-y, e_j \ra$ is maximal.  Since
      $f_{\pi(j)} \notin \SPAN \lb f_{\pi(i)} \rb_{i \ne j }$, also
      $e_j\notin \SPAN \lb f_{\pi(i)} \rb_{i \ne j }$; thus it is immediate that
      $\la x-y, e_j \ra \le \la f_{\pi(j)}, e_j \ra \le \| f_{\pi(j)} \|$.
    \end{proof}
  \end{old}
\end{myold}

\begin{mynotes}
  \begin{notes}
    Following Lemma hasn't been double/triple checked a final time for
    correctness in statement and proof.
  \end{notes}
\end{mynotes}
To state the next lemma we make the following definitions.
For a set $D \subset \RR^d$ and a unit vector $v$, let
\begin{equation}
  \label{eq:24}
  \diam(D,v) := \sup_{\substack{x,y \in D \\  x-y \in \SPAN\lb v \rb}}
  \| x - y \|.
\end{equation}
% Let
% \begin{equation}
%   %%   \label{eq:defn:Lj4}
%   L_{\jb,4} :=
%   1 \wedge \inf_{\substack{\| v\|=1,\\v \in \SPAN \{ v_{j_1}, \ldots, v_{j_k} \} }}
%   \max_{\substack{i \in \lb 1, \ldots, k \rb }} \lv\la  v, v_{j_i} \ra \rv.
% \end{equation}
\begin{lemma}
  \label{lem:Rij-Lipschitz-diameter-bound}
  Let Assumption~\ref{assm:simple-polytope} hold. Let
  $k \in \lb 1, \ldots, d \rb$, $\ib \in I_k$, and $\jb \in J_k$.
  Then for any unit length $v \in \SPAN \lb v_{j_1}, \ldots, v_{j_k} \rb$,
  \begin{equation}
    \label{eq:36}
    \diam(G_{\ib,\jb}, v) \le
    \min_{\substack{\alpha \in \lb 1, \ldots, k \rb}}
    \frac{\delta_{i_\alpha+1}}{  \lv \la  v, v_{j_\alpha } \ra  \rv },
    \quad \text{ and }
  \end{equation}
  %% and
  \begin{equation}
    \label{eq:37}
    d(G_{\ib,\jb}, \bd D , v)
    \ge %% L_{\jb, 4}
    \max_{\substack{\alpha \in \lb 1, \ldots , k \rb }}
    \frac{ \delta_{i_\alpha}}{\lv \la - v, v_{j_\alpha} \ra \rv}.
  \end{equation}
  %% where $L_{\jb,4}$, defined in \eqref{eq:defn:Lj4}, is strictly positive.
\end{lemma}
\begin{proof}
  Fix $k \in \lb 1, \ldots, d \rb$, $\ib \in I_k$, $\jb \in J_k$. Fix $v \in
  \SPAN \lb v_{j_1}, \ldots, v_{j_k} \rb$ with $\| v \| = 1$, fix $\alpha \in
  \lb 1, \ldots, k \rb$.
  Since   $\diam(G_{\ib,\jb}, v) = \diam(G_{\ib,\jb}, -v)$, we  restrict
  attention to $v$ such that
  \begin{equation}
    \label{eq:v-direction}
    \la -v, v_{j_\alpha} \ra \ge 0.
  \end{equation}
  We will upper bound $\diam(G_{\ib,\jb},v)$.  Consider $x,y \in G_{\ib,\jb}$
  such that $x-y \in \SPAN \lb v \rb$.  In particular, assume without loss of
  generality that $ x- y = \lambda v $ for $\lambda \ge 0$.
  \begin{mynotes}
    \begin{notes}
      (Do we need to use that  $\diam(G_{\ib,\jb}, v) = \diam(G_{\ib,\jb}, -v)$ ?)
    \end{notes}
  \end{mynotes}
  \begin{mylongform}
    \begin{longform}
      By the definition of $G_{\ib, \jb}$, for $x \in G_{\ib,\jb}$, $p_{j_\alpha} + \delta_{i_\alpha} \le \la x, v_{j_\alpha} \ra
      \le p_{j_\alpha} + \delta_{i_\alpha + 1}$ for any $\alpha \in \lb 1, \ldots, k \rb$.
    \end{longform}
  \end{mylongform}
  Since $x,y \in G_{\ib,\jb}$, $\la y, v_{j_\alpha} \ra \le p_{j_\alpha} + \delta_{i_\alpha+1}$ and $ p_{j_\alpha}+ \delta_{i_\alpha} \le \la x , v_{j_\alpha} \ra$;
  thus
  $\delta_{i_\alpha  } - \delta_{i_\alpha+1} \le  p_{j_\alpha} + \delta_{i_\alpha} - \la y , v_{j_\alpha} \ra \le  \lambda \la v, v_{j_\alpha} \ra $.
  Since
  $\la -v, v_{j_\alpha} \ra \ge 0$,
  we have $ \lambda \le (\delta_{i_\alpha+1}-\delta_{i_\alpha}) / \la -v, v_{j_\alpha} \ra$.
  % This shows that for any $\alpha \in \lb 1,\ldots, k \rb$ such that \eqref{eq:v-direction} holds,
  % $\diam(G_{\ib,\jb}, v) \le
  % (\delta_{i_\alpha+1}-\delta_{i_\alpha}) / \la -v, v_{j_\alpha} \ra.$
  % Since $\diam(G_{\ib,\jb}, v) = \diam(G_{\ib,\jb}, -v)$, we have that
  % $\diam(G_{\ib,\jb}, v) \le
  % (\delta_{i_\alpha+1}-\delta_{i_\alpha}) / \la -v, v_{j_\alpha} \ra$
  % for at least one $\alpha \in \lb 1, \ldots, k \rb$.
  Thus we see $\diam(G_{\ib,\jb} , -v) = \diam(G_{\ib,\jb} , v) \le (\delta_{i_\alpha+1}-\delta_{i_\alpha}) / | \la v, v_{j_\alpha} \ra |$.
  This holds for all $\alpha \in \lb 1, \ldots, k \rb$, % such that
  % \eqref{eq:v-direction} holds, so
  so for any $\tilde{v} \in \SPAN \lb v_{j_1}, \ldots, v_{j_k} \rb$ (where we do not assume
  $\la \tilde{v}, v_{j_\alpha} \ra \ge 0$)
  \begin{mylongform}
    \begin{longform}
      we have shown that $\diam(G_{\ib,\jb}, -v) = \diam(G_{\ib,\jb}, v) \le (\delta_{i_\alpha+1} - \delta_{i_\alpha}) / \la -v , v_{j_\alpha} \ra$.  Thus, dropping the assumption \eqref{eq:v-direction}, we see that $\diam(G_{\ib,\jb},v) \le (\delta_{i_\alpha+1}-\delta_{i_\alpha}) / | \la v, v_{j_\alpha} \ra |$.
    \end{longform}
  \end{mylongform}
  \begin{equation}
    %% \label{eq:36}
    \diam(G_{\ib,\jb}, \tilde{v}) \le
    %% \max_{\alpha \in \lb 1, \ldots, k \rb}
    %% \frac{\delta_{i_\alpha+1} - \delta_{i_\alpha}}{ | \la v, v_{j_\alpha } \ra | }
    \min_{\substack{\alpha \in \lb 1, \ldots, k \rb }}
    \frac{\delta_{i_\alpha+1} - \delta_{i_\alpha}}{ \lv \la  \tilde v, v_{j_\alpha } \ra  \rv }.
  \end{equation}
  % (\delta_{i_\alpha+1}-\delta_{i_\alpha}) / \la -v, v_{j_\alpha} \ra
  % \quad   \text{ for at least one } \quad
  % \alpha \in \lb 1, \ldots, k \rb.

  Next we take $v$ as above and now lower bound $d(G_{\ib,\jb}, \bd D, v)$. %
  Fix $\alpha \in \lb 1, \ldots , k \rb$.
  We begin by considering $d(G_{\ib,\jb}, H_{j_{\alpha}}, v)$. %
  \begin{mynotes}
    \begin{notes}
      Do we need: ``(Note that if \eqref{eq:v-direction} does not hold then $x + \lambda v \notin H_{j_\alpha}$ for any $\lambda \ge 0$, so we may indeed assume that \eqref{eq:v-direction} holds.)'' ?
    \end{notes}
  \end{mynotes}
  Again, since $d(G_{\ib,\jb}, \bd D, v) = d(G_{\ib,\jb}, \bd D, -v)$, we can and do assume \eqref{eq:v-direction} holds.  Fix $x \in G_{\ib, \jb}$.  Consider $\lambda > 0$ such that $x + \lambda v \in H_{j_\alpha}$.  Then
  %% $\lambda \la v, v_{j_\alpha} \ra = p_{j_\alpha} - \la x, v_{j_\alpha} \ra$,
  $\lambda \la v, v_{j_\alpha} \ra = p_{j_\alpha} - \la x, v_{j_\alpha} \ra \le
  - \delta_{i_\alpha}$ since $\la x, v_{j_\alpha} \ra \ge \delta_{i_\alpha} +
  p_{j_\alpha}$, and so $\lambda \ge \delta_{i_\alpha} / \la - v, v_{j_\alpha}
  \ra$.  This shows for any $\beta \in \lb 1, \ldots , k \rb$ that
  \begin{equation}
    \label{eq:38}
    d(G_{\ib,\jb}, \cup_{\alpha=1}^k F_{j_\alpha}, v) \ge
    \min_{\substack{\alpha \in \lb 1, \ldots , k \rb }} \frac{ \delta_{i_\alpha}}{\lv \la  v, v_{j_\alpha} \ra \rv}.
  \end{equation}
  To complete the proof, note  for $j \in \lb 1, \ldots, N \rb \setminus \jb$, that
  $$d(x, F_j , v) \ge u
  = u \min_{\alpha=1,\ldots,k}  \lv \la v, v_{j_\alpha} \ra \rv^{-1} $$
  which is larger than the right hand side of \eqref{eq:38}.
\end{proof}

\begin{myold}
  \begin{old}

    \subsubsection{Incomplete: Argument showing $G_{\ib,\jb}$ is far from facets except $F_{j_1}, \ldots, F_{j_k}$}

    Current idea is to consider $F_j$ that are adjacent to $G_{\jb}$ and $F_j$ that is not adjacent separately.

    {\bf Step 1} Assume $F_j$ is not adjacent to $G_{\jb}$.  Want to show
    \begin{equation*}
      d(G_{\jb}, F_j) - u \ge \min_{\alpha \in \lb 1, \ldots, k \rb} \frac{u}{\la v_{\jb}^*, v_{j_\alpha} \ra }
    \end{equation*}
    where $v_{\jb}^*$ is the ``slowest direction of approach to the boundary $\bd D$ near $G_{\jb}$ from $G_{\ib,\jb}$''.

    {\bf Step 2:} Let $\jb \in J_k$ and $j \in \lb 1,\ldots, N \rb \setminus
    \jb$.  Assume $F_j$ is adjacent to $G_{\jb}$.  Somehow need to use that $v_j$
    does not lie in $\SPAN \lb v_{j_1}, \ldots, v_{j_k} \rb$, because $H_j$
    intersects $G_{\jb}$.  If $v \in \RR^d$, $\|v\|=1$, is such taht
    $d(G_{\ib,\jb}, F_j, v)$ is small then for facets of $G_{\jb}$ near $F_j$,
    $\la v, v_j \ra$ must be almost $0$ so $\la v, v_i \ra$ must be large for
    other facets $F_i$ of $G_{\jb}$, some $i \ne j$.

  \end{old}
\end{myold}

\begin{myold}
  \begin{old}

    It remains to show that for all $\ib \in I_k$, $ d(G_{\ib,\jb}, (\bd
    D)\setminus \cup_{\alpha=1}^k F_{j_\alpha}, v)$ is larger than the right
    hand side of \eqref{eq:38}, up to the constant $L_{\jb,4}$, and that
    $L_{\jb,4} > 0$.
    % since $\min_{\alpha \in \lb 1, \ldots k \rb} \frac{u}{d(G_{\ib,\jb} , F_{j_\alpha}, v)}$ governs the Lipschitz bound for functions on $G_{\ib,\jb}$.
    %% Fix $k \in \lb 1, \ldots, d \rb$.
    Note that $v \mapsto \la v, v_{j_i} \ra$ is continuous for all $i \in \lb 1, \ldots, k \rb$.
    Thus $v \mapsto \max_{i \in \lb 1, \ldots, k \rb } \la v, v_{j_i} \ra$ is continuous, and it is bounded away from $0$ for all unit vectors $v \in \SPAN \lb v_{j_1}, \ldots, v_{j_k} \rb$.
    \begin{mylongform}
      \begin{longform}
        (almost by definition of ``span'', although $v_{j_i}$ are not orthonormal.)
      \end{longform}
    \end{mylongform}
    Thus $L_{\jb,4} > 0$.  Then for $j \in \lb 1, \ldots, N \rb \setminus \jb$,
    $d(x, F_j , v) \ge u
    \ge L_{\jb, 4} \frac{u}{L_{\jb,4}},$
    % \begin{equation}
    %   %%   \label{eq:37}
    %   d(x, F_j , v) \ge u
    %   \ge L_{\jb, 4} \frac{u}{L_{\jb,4}},
    % \end{equation}
    \begin{mynotes}
      \begin{notes}
        Want to say that the above is
        \begin{equation*}
          \ge L_{\jb, 4} \min_{i \in \lb 1, \ldots, k \rb}d(x, F_{j_i}, v),
        \end{equation*}
        but this is not quite true (since the right hand side of the inequality
        \eqref{eq:26} could be strictly larger than the left hand side).

        Need clarify the last inequality.

        Note this approach does all facets at once; doesn't break into incident
        ones and other ones.
      \end{notes}
    \end{mynotes}
    which shows that $d(x, F_j,v ) $ is no smaller than the right hand side of \eqref{eq:38} for any unit vector $v \in \SPAN \lb v_{j_1}, \ldots, v_{j_k} \rb$, up to a factor of $L_{\jb,4}$.
  \end{old}
\end{myold}

% \begin{lemma}
%   \label{lem:Rij-Lipschitz-diameter-bound}
%   Let Assumption~\ref{assm:simple-polytope} hold, fix
%   $k \in \lb 1, \ldots, d \rb$, $\ib \in I_k$, $\jb \in J_k$, and let
%   $e_{\ib,\jb, \alpha} \equiv e_{\alpha}$, $\alpha \in \lb 1, \ldots, d \rb$,
%   be given by Proposition~\ref{prop:basis-lipschitz}.  Then let $V_1$ be
%   $\SPAN \lb e_1, \ldots,e_k \rb$ and $V_2, \ldots, V_{d-k+1}$ be
%   $\SPAN \lb e_{k+1} \rb, \ldots, \SPAN \lb e_{d} \rb$, respectively.
%   Let $R_1 := \max_{\alpha \in \lb 1, \ldots, k \rb}
%   2(\delta_{i_\alpha+1}-\delta_{i_\alpha}) / \delta_{i_\alpha}$, %
%   and %
%   let
%   $    \ve{R}:= \lp R_1,
%   \frac{8 L_{k,1} \rho_{\jb, k+1}}{u},
%   \ldots,
%   \frac{8 L_{k,1} \rho_{\jb, d}}{u} \rp$
%   %   \begin{equation*}
%   %     \ve{R}:= \lp R_1,
%   %     \frac{8 L_{k,1} \rho_{\jb, k+1}}{u},
%   %     \ldots,
%   %     \frac{8 L_{k,1} \rho_{\jb, d}}{u} \rp .
%   %   \end{equation*}
%   where $\rho_{\jb,\alpha} := w(G_{\jb},e_\alpha)$.
%   Then
%   \begin{equation}
%     \label{eq:C-inclusion-lemma}
%     \cC[D,1]|_{G_{\ib,\jb}}
%     \subset \cC[G_{\ib,\jb}, 1, \ve{R}, V_1, \ldots, V_{d-k+1} ].
%   \end{equation}
% \end{lemma}

\section{Further applications}
\label{sec:further-applications}

We now consider further entropy bounds that rely on the above ideas, results, or their proofs.
In Subsection~\ref{sec:m-monotone} we consider so-called univariate and multivariate {\it $m$-monotone functions}.
In Subsection~\ref{sec:entr-level-set} we briefly consider estimation of level sets of convex functions and the question of adaptation to polytopal level sets.  Further discussion is given at the beginning of the two subsections.

\subsection{Bracketing entropy of $\mm$-monotone function classes}
\label{sec:m-monotone}

The shape constraint of {\it $\mm$-monotonicity}, for $\mm \in \lb 0,1,2, \ldots \rb$, is useful because it serves, roughly, as a higher order convexity restriction (when $\mm > 2$).  An $\mm$-monotone function $f$ satisfies further convexity restrictions % (see Definition~\ref{defn:k-monotone-d=1})
besides simply convexity of $f$, and so in many settings is even nicer to work with than convex functions are.
When $d=1$, $\mm$-monotonicity is defined as follows (by, e.g., \cite{Williamson:1955uy,Williamson:2010wy}). %% and the references therein for some of the important properties of (univariate) $\mm$-monotone  functions.
\begin{definition}
  \label{defn:k-monotone-d=1}
  A function $f \colon [0,\infty) \to \RR$ is $0$-monotone if it is
  nonnegative, $1$-monotone if it is nonnegative and nonincreasing, and
  $2$-monotone if it is nonnegative, nonincreasing, and convex; $f$ is
  $\mm$-monotone for $\mm \ge 2$ if $(-1)^l f^{(l)}$ exists and is nonnegative,
  nonincreasing, and convex for $l= 0, 1, \ldots, \mm-2$.
\end{definition}
(Here $f^{(l)}$ is the $l$th derivative,  with $f^{(0)} \equiv f$.)
When $\mm=1$ or $2$, a large body of statistical work and results exists (some of which also allows the case where $d > 1$), some of which is referenced in the introduction of this paper.  
Statistical properties of two (nonparametric) estimators of a (univariate) $\mm$-monotone density, for general $\mm$ (and $d = 1$), were introduced and studied in
\cite{Balabdaoui:2007jj,Balabdaoui:2010do}; see also \cite{Gao:2009hf}. For instance, in statistical settings, if a function being nonparametrically estimated is known to be $\mm$-monotone, then it can be estimated at a faster rate of convergence than if it were just convex \citep{Balabdaoui:2007jj,Balabdaoui:2010do,Gao:2009hf}.  As discussed in Section~\ref{sec:intr}, % (with references also given there)
in the univariate setting $\mm$-monotone functions have been studied, but we do not even know of a formal definition of $\mm$-monotonicity in the multivariate case.
In fact, as discussed below, there are several possible definitions one could use for $\mm$-monotonicity that generalize the univariate definition.  We present one definition which has the benefit of being amenable to finding bracketing upper bounds.  We then show that our proof for the bracketing upper bound for convex functions, Theorem~\ref{thm:main-thm} (and Corollary~\ref{cor:general-convex-polytope}), applies to yield a bracketing bound for classes of $\mm$-monotone functions.  This is the main result of Subsection~\ref{sec:m-monotone} and 
is given in Theorem~\ref{thm:main-extension-m-mono}.  Recall that the proof of Theorem~\ref{thm:main-thm} relies on Theorem~\ref{thm:GS-extension}.  There is no known or immediate analog for Theorem~\ref{thm:GS-extension} in the general $\mm$-monotone case.  Thus we prove an analog, Theorem~\ref{conj:m-mono-sup-bracket} (under a certain technical restriction, given below in Definition~\ref{defn:m-mono} Part~\ref{item:3}), which we use to prove
Theorem~\ref{thm:main-extension-m-mono}.
% Unfortunately, it is unknown as of yet whether Property~\ref{property:bracketing-FTOC} holds; thus, there is a gap in our proof.  However, we do nonetheless show in our proof of Theorem~\ref{thm:main-extension-m-mono} that the techniques used to prove Theorem~\ref{thm:main-thm} can be applied to any class for which the conclusion (i.e., \eqref{eq:conj-k-mono-multivariate}) of Theorem~\ref{conj:m-mono-sup-bracket} holds.  In Subsubsection~\ref{sec:mm-mono-d=1} we show, in Theorem~\ref{thm:k-mono-entropy}, that in the $d=1$ case the conclusion of Theorem~\ref{conj:m-mono-sup-bracket} indeed holds, potentially providing some evidence for this result in the $d > 1$ case.  Theorem~\ref{thm:k-mono-entropy} is also of interest in its own right.

As mentioned above, there are many possible methods
for defining a class of $\mm$-monotone functions in the multivariate setting.
This is perhaps illustrated by the fact that there are many competing
definitions of monotonicity (i.e., $1$-monotonicity) in dimension $d \ge 2$.
One can define a function $f$ to be multivariate monotone (or unimodal) via {\it star monotonicity}, meaning that along all rays emanating from a special fixed point, $f$ is monotone.  This, for instance, is a suggested definition used in the related context of {hyperbolic monotonicity}  by \cite[p. 600]{Cule:2010dc}.
\cite{Gao:2007kw} consider entropy bounds for {\it block-decreasing}
densities.  Even ``block decreasing'' can be defined in multiple ways:
\cite{Gao:2007kw} and
\cite{Dharmadhikari:1988um} differ in their definitions of this term.
\begin{mylongform}
  \begin{longform}
    (using a slightly different definition than that on p.\ 42 of
    \cite{Dharmadhikari:1988um}): \cite{Gao:2007kw} define a function $f$ to
    be block-decreasing on $[0,\infty)^d$ if it is nonincreasing along any
    line parallel to an axis.
  \end{longform}
\end{mylongform}
Very recent statistical work has considered {\it entire monotonicity} in the regression setting \citep{Fang:2019vb}.
See \cite[chapter
2]{Dharmadhikari:1988um} for several other possible definitions of unimodality (they focus on unimodality rather than monotonicity, but the two settings are very similar).
Many of the above definitions are not amenable to accurate entropy computations, at least with the tools we are aware of  at present.  
In Subsubsection~\ref{sec:mult-mm-monot} we present a definition of multivariate $\mm$-monotonicity that is amenable to entropy calculations; the results we get suggest that the entropies are of the ``right'' order of magnitude ($\epsilon^{-d/m}$ as $\epsilon \searrow 0$) that we might expect a priori.  This suggests that our definition is indeed a reasonable one.
In Subsubsection~\ref{sec:mm-mono-d=1} we return briefly to the particular $d=1$ case.

\begin{mylongform}
  \begin{longform}

    star unimodal

    block unimodal

    linear unimodal

    convex unimodal

    central convex unimodal

    monotone unimodal

  \end{longform}
\end{mylongform}

% This definition motivates our definition of
% $\mm$-monotonicity for $\mm \ge 3$ when $d \ge 2$, which is as follows.

% \begin{definition}
%   \label{defn:k-monotone-d>=1}
%   For a convex set $D_0 \subset [0,\infty)^d$, $0 < B<\infty$, and
%   $\mm \in \lb 2, 3, \ldots \rb$, let $\cC^{\mm}(D_0, B)$ be the class of
%   $\mm$-monotone functions $f \colon D_0 \to [0,B]$, satisfying the following.  For
%   $f \in \cC^{\mm}(D_0, B)$, assume $0 \le f \le B$,\footnote{Note that we
%   restrict $f \in \cC^k(D_0,B)$ to be positive, so $\cC^2(D_0, B)$ is not
%   quite identical to $\cC(D_0,B)$, where $-B \le f \le B$ is allowed.}  and
%   assume $f$ is $(\mm-2)$-times continuously differentiable on $D_0$ (if $\mm=2$
%   then assume $f$ is continuous on $D_0$).  For any
%   $v = (0,\ldots, 0,1,0,\ldots, 0) \in \RR^d$ and $t \ge 0$, assume that
%   $g(t):= f(x + t v)$ is $\mm$-monotone (according to
%   Definition~\ref{defn:k-monotone-d=1}) for any $x \in \bd([0,\infty)^d)$.
%   Further, for $i \in \lb 0, 2, 4, \ldots \mm-2 \rb $, for any $x \in D_0$,
%   $v \in \RR^d$, and for $t \in \RR$ such that $x + tv \in D_0$, assume that
%   $g(t) := f^{(i)}(x + t v )$ is convex.
% \end{definition}

\subsubsection{Multivariate $\mm$-monotone functions}
\label{sec:mult-mm-monot}

Fix the dimension $d \ge 1$.  We will use so-called $d$-dimensional multi-index notation: a vector of nonnegative integers $\ib = (i_1, \ldots, i_d)$ is a multi-index.  We let $|\ib| := i_1 + \cdots + i_d$.  Let $\II_\mm$ be the set of multi-indices $\ib$ with $|\ib| = \mm$.\newnot{symbol:IIk} For two vectors $K=(K_1, \ldots, K_j), $ $L = (L_1, \ldots, L_j) \in \RR^j$ with $L_i > 0$, we let $L^{K} := L_1^{K_1}\cdots L_j^{K_j}$.  For any function $f$, we let $f^{(\ib)}$ be $ \frac{\partial^{|\ib|}}{\partial x_{i_1} \cdots \partial x_{i_d} } f$, whenever this is well-defined.  We let $\pderiv{e_i}f(x)$ denote $\deriv{t}\vert_{t=0}f(x + t e_i)$, and for an orthonormal basis ${\bs e} := \lb e_1, \ldots, e_d \rb$ and $\jb \in \II_i$, we let $f^{(\jb)}_{\bs e} := \frac{\partial^{|\jb|}}{\partial e_1^{j_1} \cdots \partial e_d^{j_d}}f$.

Our $\mm$-monotone classes are based on any subclass $\cC^*$ of convex functions having a certain needed property.
%% We defien $\mm$-monotonicity has a ``directionality'' to 
The idea of multivariate $\mm$-monotonicity involves convexity of partial derivatives in certain directions; since such convexity is not preserved by rotation (see Remark~\ref{rem:2}),
%% in our Definition~\ref{defn:m-mono}
we will define  $\mm$-monotonicity to be {\it relative} to a domain $D_0$.  In the case where $D_0$ is a hyperrectangle, the definition simplifies (see Remark~\ref{rem:3}).  

\begin{definition}
  \label{defn:m-mono}

  \

  \begin{enumerate}[label=\Alph*.,ref=\Alph*,leftmargin=*] %%\label=\Alph*.,
  \item
    \label{item:3}
    For a convex set $G \subset \prod_{i=1}^d [a_i,b_i]$, let $\cC^*(G)$ be any subclass of $\cC(G)$ such that for all $B, \Gb$,  $\cC^*(G, B, \Gb) := \cC^*(G) \cap  \cC(G,B,\Gb)$ satisfies the following $L_\infty$ cover property.  For all $\epsilon > 0$, there exists a
    $L_\infty$-$\epsilon$-cover of log cardinality no larger than $c \epsilon^{-d/2} ( B + \sum_{i=1}^d \Gamma_i(b_i-a_i))^{d/2}$.  The cover must satisfy the following.  For any $f \in \cC^*(G,B,\Gb)$, any $x, y \in G$, and any rotation matrix $A \in \RR^{d \times d}$ ($\det A = 1$ and $A' = A^{-1}$), %%  $g(y) := f(Ay)$ satisfies
    we have
    \begin{equation}
      \label{eq:82}
      \pderiv{x_i} \int_0^1 |  h(A(x + t(y-x))) -  f(A(x+t(y-x)))| dt \le \epsilon 
    \end{equation}
    for all $i \in \lb 1, \ldots, d \rb$, where $h$ is the $L_\infty$-closest element of the cover to $f$.

  \item \label{item:4}
    Let $D_0 \subset [0,\infty)^d$ be a convex polytope, let $0 < B < \infty$, and let $\mm \ge 2$ be an integer. We define the class of {\it $m$-monotone functions relative to $D_0$}, denoted $\cC^{\mm}(D_0, B)$, to be the set of all $f \in \cC(D_0,B)$ satisfying the following.
    %% $f \colon D_0 \to [0,B]$ such that 
    For each vertex of $D_0$, 
    for all ($d!$ possible) orthornomal bases $\bs{e}$ given by % 
    Proposition~\ref{prop:basis-lipschitz}, %
    %% either   $f^{(\ib)}_{\bs e}$  or $(-1)^{|\ib|} f^{(\ib)}_{\bs e}$
    either $f^{(\ib)}_{\bs e}$ or $- f^{(\ib)}_{\bs e}$
    lies in $\cC^*(D_0)$
    for all $\ib \in \II_{j}$, $0 \le j \le m-2$.
  \end{enumerate}
\end{definition}

% \begin{definition}
%   \label{defn:m-mono}
%   Let $D_0 \subset \RR^d$ be a convex %% not nec compact
%   set, let $0 < B < \infty$, and let $\mm \ge 2$ be an integer.
%   We define  $\cC^\mm_d(D_0, B)$ to be the class of functions
%   $f \colon D_0 \to [0,B]$
%   such that $(-1)^{|\ib|}f^{(\ib)}$ is convex for all
%   $\ib \in \II_{j}$, $0 \le j \le m-2$.
% \end{definition}

% \begin{remark}
%   \label{rem:1}
%   Note that because sums of convex functions are convex and convex functions of linear functions are convex, if $f \in \cC^{\mm}(D_0,B)$, then we also have $g(x) := f(A x) \in \cC^{\mm}(A^{-1} D_0,B)$ for any invertible linear function $A$.
% \end{remark}

\begin{remark}
  The fundamental idea of $\mm$-monotonicity is given by Part~\ref{item:4} of
  Definition~\ref{defn:m-mono}.  The technical requirement \eqref{eq:82} of
  Part~\ref{item:3} is needed for the proof of our bracketing bound.  It is
  not clear at this point if it can be removed or not.  When $y = x+ x_j$ and
  $A$ is the identity, for continuously differentiable $h$ and $f$, the
  property holds automatically by the Fundamental Theorem of Calculus.
  Ideally we would like to replace $\cC^*$ by the full class $\cC$ (which is
  possible when $d=1$, see Remark~\ref{rem:4} and  Subsubsection~\ref{sec:mm-mono-d=1}).  We leave an investigation of whether this is possible for future work.
\end{remark}

\begin{remark}
  \label{rem:2}   \label{rem:3}
  The property of $\mm$-monotonicity is preserved by translation and
  rescaling.  However, while rotations of convex functions are still convex,
  if $f^{(\jb)}$ is convex, $|\jb| > 2$, then after rotation (by, say, a matrix $A$), $g^{(\jb)} := (f(A \cdot))^{(\jb)}$ is not necessarily convex.
  This is because a mixed partial derivative of a rotation is a linear
  combination of mixed partial deriatives; if some of the linear coefficients
  are negative, then the resulting function may no longer be convex.  This is
  why our definition is relative to the domain, $D_0$ (so an $\mm$-monotone
  function after a rotation will be $\mm$-monotone relative to the rotated domain).
  Note that when $D_0$ is a hyperrectangle, for $f$ to be $\mm$-monotone relative to
  $D_0$, it is sufficient that $f^{(\ib)}$ be convex for all $\ib \in \II_i$,
  $0 \le i \le \mm-2$.
\end{remark}

\begin{remark}
  Our definition of multivariate $\mm$-monotonicity captures a higher order
  type of convexity.  It does not enforce the alternating sign condition
  ``$(-1)^jf^j \ge 0$'' that is generally required in the univariate case.
  That is,
  %% when $|\ib|$ is odd,
  we allow $f_{\ve e}^{(\ib)}$ to be either
  convex or to be concave.  In the univariate case, there is only one
  direction in which one is computing a derivative.  In the multivariate
  case, since we consider many different bases $\ve e$, and may have
  instances where a vector $e_i$ and its opposite $-e_i$ are contained in two
  different bases, thus potentially switching the sign of
  $f_{\ve e}^{(\ib)}$, we must allow $f_{\ve e}^{(\ib)}$ to be either convex
  or concave.  Further restrictions to our definition could be enforced if
  needed in a specific application; our entropy bounds would of course still
  apply.
\end{remark}

\begin{example}
  \label{ex:1}
  Let $a \in \RR^d$ have nonnegative components and let $b > 0$; let
  $( \cdot )_+ := \max(\cdot ,0)$.  The function
%% $f(x) :=  \lp 1 -  (x_1 + \cdots + x_d) \rp_+^{\mm-1} \one_{\{ 0 \le x_i < \infty, \text{ for all } i \}}$
  $f(x) := b \lp 1 - \la a, x \ra  \rp_+^{\mm-1}
  \one_{[0,\infty)^d}(x)$ is a primary example
  of an $\mm$-monotone function (i.e., satisfies Part~\ref{item:4} of Definition~\ref{defn:m-mono}).  For any $\mm \ge 1$, the function
  $f(x) := e^{- b  \la a , x \ra } \one_{[0,\infty)^d}(x)$ is
  $\mm$-monotone.  Both are $\mm$-monotone relative to any
  hyperrectangle.
  Further examples of  $\mm$-monotone functions can be generated by taking linear combinations.
\end{example}

\begin{example}
  The functions in Example~\ref{ex:1} are also $\mm$-monotone relative to polytopes beyond hyperrectangles.  For simplicity, let $f(x) := (1 - (x_1 + x_2))_+^{m-1} \one_{[0,\infty)^2}(x)$.
  Note that if we let $g(x) = (1 - (a_1 x_1 + a_2 x_2))_+^{m-1} \one_{[0,\infty)^2}(x)$,
  then for $k \le m-1,$
  \begin{equation}
    \label{eq:83}
    \frac{ \partial^k}{ \partial x_1^j \partial x_2^{k-j}} g(x)
    = \frac{ (m-1)! }{(m-1-k)!} (-1)^k a_1^j a_2^{k-j}
    ( 1 - a_1 x_1 - a_2 x_2)_+^{(m-1-k)}    
  \end{equation}
  Thus $(-1)^k     \frac{ \partial^k}{ \partial x_1^j \partial x_2^{k-j}} g(x)$ is convex if $a_i >  0$, $i=1,2$, and $k \le m-2$.

  Let $e_1= (1,0)'$, $e_2 = (0,1)'$ (where $'$ denotes transpose), and let the basis $\ve d : = \lb d_1, d_2 \rb$ be defined by $d_i := A e_i $ where 
  \begin{equation*}
    A =
    \begin{pmatrix}
      \cos \theta & - \sin \theta \\
      \sin \theta & \cos \theta 
    \end{pmatrix}
  \end{equation*}
  is the matrix giving rotation by angle $\theta$.  If $g(y) := f(Ay')$ then $g^{(\jb)} = f_{\ve d}^{(\jb)}$.  Thus, by \eqref{eq:83}, $f$ is $\mm$-monotone relative to any polytope since the partial derivatives are always either convex or concave.

  Furthermore, as long as $\cos \theta + \sin \theta \ge 0$ and
  $\cos \theta - \sin \theta \ge 0$, i.e., as long as
  $ - \pi / 4 \le \theta \le \pi / 4$, $f^{(\jb)}_{\ve d}$ is convex (for
  $\jb \in \II_{j}$, $0 \le j \le m-2$). Thus, for instance, if we take
  $ - \pi /4 \le \theta \le 0$, and let
  $H_{\theta} := \lb x : \la d_2 , x \ra \le \cos \theta \rb$ (the rotation
  of the line $\lb x_2 = 1 \rb$ by angle $\theta$ about the point $(0,1)$),
  %% then $f$   is $\mm$-monotone relative to $[0,1]^2 \cap H_{\theta}$
  then if we let $D_0 := [0,1]^2 \cap H_{\theta}$, then  $f^{(\jb)}_{\ve d}$ is convex where
  $\ve d$ is one basis given by Proposition 3.2 at the upper right vertex of $D_0$.  (The other basis given by Proposition 3.2 is the standard basis $\ve e = \lb e_1, e_2 \rb$.)
\end{example}

\noindent We now define classes of Lipschitz bounded $\mm$-monotone functions, which are needed for us to generalize Theorem~\ref{thm:GS-extension}.

\begin{definition}
  \label{defn:m-mono-Lipschitz}
  Let $D_0 \subset [0,\infty)^d$ be a convex polytope, let $0 < B < \infty$, and let $\mm \ge 2$ be an integer.  For all
  vertices $v$ of $D_0$, for all orthonormal bases $\ve e$ given by
  Proposition~\ref{prop:basis-lipschitz}, and all $\ib \in \II_{m-1}$, let
  $0 < \Gamma_{\ve{e},\ib} < \infty$, and let $\Gb$ be the set of all such
  $ \Gamma_{\ve{e}, \ib}$.  Let $\cC^{\mm}_d(D_0, B, \Gb)$ be the class of
  functions $f \in \cC^{\mm}_d(D_0,B)$ such that for all
  $\ib \in \II_{\mm-2}$ and orthonormal bases $\ve e$ (given by
  Proposition~\ref{prop:basis-lipschitz} for any vertex of $D_0$) the
  function $f_{\ve e}^{(\ib)}$ is Lipschitz in the following sense.  For each
  $\ib \in \II_{\mm-1}$ and $\jb \in \II_{\mm-2}$, $\ib - \jb$ is $1$ in a
  single coordinate, which we denote $\alpha_{\ib,\jb}$.  Then, for {\it any}
  $\jb \in \II_{\mm-2}$, assume
  for all $x, x+\lambda e_{\alpha_{\ib,\jb}} \in D_0$ that
  $| f^{(\jb)}(x+\lambda e_{\alpha_{\ib,\jb}}) - f^{(\jb)}(x)
  | \le \Gamma_{{\ve e}, \ib} | \lambda |.$
\end{definition}

  % Now for all
  % $\ib \in \II_{\mm-1}$, let $0 < \Gamma_{\ib} < \infty$ and let
  % $\Gb := \lb \Gamma_{\ib} : \ib \in \II_\mm \rb$.  Let
  % $\cC^\mm_d(D_0, B, \Gb)$ be the class of functions
  % $f \in \cC^{\mm}_d(D_0,B)$ such that for all $\ib \in \II_{\mm-2}$, the function
  % $f^{(\ib)}$ is Lipschitz in the following sense.  For each
  % $\ib \in \II_{\mm-1}$ and $\jb \in \II_{\mm-2}$, $\ib - \jb$ is $1$ in a
  % single coordinate, which we denote $\alpha_{\ib,\jb}$.  Then, for {\it any}
  % $\jb \in \II_{\mm-2}$, assume for all $x,y \in D_0$ that
  % $| f^{(\jb)}(x) - f^{(\jb)}(y)| \le \Gamma_{\ib} | x_{\alpha_{\ib,\jb}} -
  % y_{\alpha_{\ib,\jb}} |.$

\begin{remark}
  \label{rem:m-mono-deriv}
  Let $\cC^{\mm, \circ}_d(D_0, B, \Gb)$ denote the subset of 
  $f \in \cC^{\mm}_d(D_0,B, \Gb)$ that are
  $(m-1)$-times continuously differentiable.
  %% infinitely differentiable.
  %%$m$-times continuously differentiable.
  Note that for such $f$, we have
  $| f^{(\ib)} | \le \Gamma_{\ib}$.
  Let $C^{m-1}$ denote the class of $(m-1)$-times continuously differentiable functions.
  Then $C^{1} \cap  \cC_d(D_0, B, \Gb)$ is $L_\infty$-dense in $ \cC_d(D_0, B, \Gb)$,
  so $\cC^{\mm, \circ}_d(D_0,B,\Gb) := C^{m-1} \cap \cC_d^{\mm}(D_0, B, \Gb)$
  is dense in $\cC_d^{\mm}(D_0, B, \Gb)$  (\cite{Czarnecki:2006dv}; see also Lemma~1.1 of \cite{Gao:2008wf}).
  This means that any $L_\infty$ (bracketing or metric) entropy bound for $\cC^{\mm, \circ}_d(D_0, B, \Gb)$ implies the same bound on $\cC^{\mm}_d(D_0, B, \Gb)$.  We will use this in our proofs. 
\end{remark}

% \begin{definition}
%   \label{rem:m-mono-deriv}
%   For $\mm \ge 2$, let $\cC^\mm_d(D_0, B, \Gb)$ be the class of functions $g \colon D_0 \to [0,B]$ such that $f(x) := g(Ax)$, for any rotation matrix $A$ (det $A =1$) is such that $(-1)^{m-2}f^{(\ib)}$ is convex for all $\ib \in \II_{\mm-2}$, and further $f^{(\ib)}$ is well defined and continuous for all $\ib \in \II_{\mm-1}$ and is bounded by $\Gamma_{\ib}$: 
% \end{definition}

The following lemma provides uniform bounds on  the smoothness of  %
the % spacing ? 
functions in $\cC^m_d(D,1)|_{G_{\ib,\jb}}$, which, together with Theorem~\ref{conj:m-mono-sup-bracket} below,
allows us to later prove Lemma~\ref{lem:1} and thus prove the main Theorem~\ref{thm:main-extension-m-mono}.
% To state the lemma, we let $\pderiv{e_i}f(x)$ denote $\deriv{t}\vert_{t=0}f(x + t e_i)$,
% and for an orthonormal basis ${\bs e} := \lb e_1, \ldots, e_d \rb$ and $\jb \in \II_i$,  we let $f^{(\jb)}_{\bs e} := \frac{\partial^{|\jb|}}{\partial e_1^{j_1} \cdots \partial e_d^{j_d}}f$.

\begin{lemma}
  \label{lem:m-mono-deriv-bound}
  Let $D_0$ be a convex polytope, let $\cC^{\mm,\circ}_d(D_0, 1)$ be as defined in Remark~\ref{rem:m-mono-deriv} for $m \ge 2$, and let $f \in \cC^{\mm, \circ}_d(D_0,1)$.

  \begin{enumerate}[label=\Alph*.,ref=\Alph*,leftmargin=*] %%\label=\Alph*., ref=\Alph*

  \item  \label{item:1}
    If  $x$ is interior to $D_0$ such that $B(x, r_0) \subset D_0$, $r_0 > 0$,
    then for any $\jb \in \II_j$, $0 \le j \le m-1$, we have
    $| f^{(\jb)}(x)|  \le K_j / r_0^j$ for a constant $0 < K_j$.
    %% $K$ a constant in what sense 

  \item
    \label{item:2}
    %% Let     ${\ve e} := (e_1, \ldots, e_d)$ be the standard orthonormal basis of $\RR^d$,
    Let $\ve e$ be any orthonormal basis  of $\RR^d$
    such that $f_{\ve e}^{(\ib)}$ is convex for all $\ib \in \II_{j}$, $0 \le j \le m-2.$
    Assume  we have $d(x, \bd D_0, e_i) \ge \delta_i$, where $\delta_i > 0$ for $i=1, \ldots, d$.
    Let $\delta = (\delta_1, \ldots, \delta_d)$.  Then
    for any $\jb \in \II_j,$ $0 \le j \le m-1$, we have $| f_{\ve e}^{(\jb)}(x)| \le K_j / \delta^{\jb}$.
    % \item
    %   \label{item:2}
    %   Let 
    %   ${\ve e} := (e_1, \ldots, e_d)$ be the standard orthonormal basis of $\RR^d$,
    %   $e_1 = (1, 0, \ldots, 0),$ ..., $e_d = (0, \ldots, 0, 1)$.  Assume   we have $d^+(x, \bd D_0, -e_i) \ge \eta > 0$ and $\eta \ge d^+(x, \bd D_0, e_i) \ge \delta_i$, where $\delta_i > 0$ for $i=1, \ldots, d$.
    % %   for $l \in \lb 2, \ldots, m-1 \rb$, we have $0 \le \pderiv{e_i}[][l]f(x) \le c_l / \delta_i^l $ for a constant $0<c_l<\infty$, and $| \pderiv{e_i}[][]f(x) | \le c_1 / \delta_i$ for a constant $0<c_1<\infty$.
    %   Let $\delta = (\delta_1, \ldots, \delta_d)$.  Then $0 \le \sum_{ \ve l \in \JJ_{j}} f^{(\ve l)}_{\ve e}(x) \delta^{\ve l} \le K < \infty$ for a constant $K \equiv K_{d,m}$ and any $j \in \lb 0, \ldots, m-1 \rb.$
  \end{enumerate}
\end{lemma}

\begin{proof}

  We first show part~\ref{item:1}.
  % We will use the following basic fact:
  % \begin{equation}
  %   \label{eq:79}
  %   \text{If } f \in \cC(D, B),
  %   %%   \text{ for }
  %   B > 0, \text{ and } x \in D \text{ is such that } d(x, D, e_i) \ge \delta > 0
  %   \text{ then }  \lv \pderiv{x_i}f(x)\rv \le \frac{2B}{\delta}.
  % \end{equation}
  We will show, by induction on $l < m-1$, that for $\ib \in \II_l,$ we have
  that $(f^{(\ib)})|_{B(x, r_0/2^l)} \in \cC(B(x, \frac{r_0}{2^l}), \frac{2^{l(l+3)/2}}{r_0^l})$.  When $l = m-1$, the statement
  $(|f^{(\ib)}|)|_{B(x, r_0/2^l)} \le \frac{2^{l(l+3)/2}}{r_0^l}$ holds.
  The base case of $l=0$ is satisfied trivially by assumption, since  $f \in \cC^{\mm, \circ}_d(D_0,1)$.  Now we show the induction hypothesis holds for a general $l \le m-1$  by assuming it holds for the $l-1$ case.   
  Take $\ib \in \II_{l}$.  Write (non-uniquely) $\ib = \ib_1 + \ib_2$ for $\ib_1 \in \II_{l-1}$ and $\ib_2 \in \II_1$.  Since $l-1 \le m-2$, $f^{(\ib_1)}$ is convex by assumption, so by the induction hypothesis $(f^{(\ib_1)})|_{B(x, r_0/2^{l-1})} \in \cC(B(x, \frac{r_0}{2^{l-1}}), \frac{2^{(l-1)(l+2)/2}}{r_0^{l-1}})$.  Note for $z \in B \lp x, \frac{r_0}{2^{l}} \rp $, that $d(x, \bd B \lp x, \frac{r_0}{2^{l-1}} \rp, e_i) \ge r_0 / 2^{l}$ for any $i.$ Since $f^{(\ib)} = (f^{(\ib_1)})^{(\ib_2)}$,
  Lemma~\ref{lem:basic-lemma_conv-lipschitz} implies for any $z \in B \lp x, \frac{r_0}{2^{l}} \rp$ that $|f^{(\ib)}(z)| \le \frac{2}{ r_0/2^{l}} \frac{2^{(l-1)(l+2)/2}}{r_0^{l-1}} = \frac{2^{l(l+3)}}{r_0^l} $.  Thus part~\ref{item:1} has been shown.
  Part~\ref{item:2} follows from part~\ref{item:1} by a simple scaling argument: let $A$ be the diagonal matrix of $\delta$, so that $g(y) := f(A y)$ is defined on a hyperrectangle $E$ where $d(Ax, \bd E, e_i) \ge 1$.  Note $\delta^{\ib} f^{(\ib)} = g^{(\ib)}$, and then apply
  part~\ref{item:1}.
\end{proof}

\begin{theorem}
  \label{conj:m-mono-sup-bracket}
  Let $m \ge 2$. %% be even.
  Let $D = \prod_{i=1}^d [a_i,b_i]$ be a hyperrectangle, with
  $-\infty < a_i < b_i < \infty$.  For all $\ib \in \II_{m-1}$, let
  $0 < \Gamma_{\ib} < \infty$ and let
  $\Gb := \lb \Gamma_{\ib} : \ib \in \II_{m-1} \rb$.  Let
  %% $0 < B \le \sum_{\ib \in \II_{m-1}} \Gamma_{\ib} (b-a)^{\ib}$.
  $0 < B \le \max_{\ib \in \II_{m-1}} \Gamma_{\ib} (b-a)^{\ib}$.
  % If the $\epsilon$-brackets for
  % $\cC[D, B, \Gb]$ (given by
  % Theorem~\ref{thm:GS-extension}) can be taken
  % to satisfy Property~\ref{property:bracketing-FTOC},
  Then there exists  $c \equiv c_{m,d}$ such that for all $\epsilon > 0$,
  \begin{equation}
    \label{eq:conj-k-mono-multivariate}
    %% \Nb[ \epsilon \vol_d(D)^{1/p}, \cC^m_d(D, B, \Gb),  L_p]     \le 
    \Nb[ \epsilon,
    \cC^m_d(D, B, \Gb),
    L_\infty]
    %% \le c \lp \frac{\sum_{\ib \in \II_{m-1}} \Gamma_{\ib} (b-a)^{\ib}} {\epsilon} \rp^{d/m},
    \le \exp \lb c \lp \frac{\max_{\ib \in \II_{m-1}} \Gamma_{\ib} (b-a)^{\ib}} {\epsilon} \rp^{d/m}
    \rb,
  \end{equation}
  where $a := (a_1, \ldots, a_d)$ and $b:= (b_1, \ldots,b_d)$.
  Note in \eqref{eq:conj-k-mono-multivariate},
  $\cC^m_d(D, B, \Gb)$ may trivially be replaced by
  $\cC^m_d(D, B, \Gb)|_{G}$ for any $G \subset D.$
\end{theorem}

%% The proof proceeds in fashion similar to the proof of Theorem~\ref{thm:k-mono-entropy}, via several lemmas.

The proof proceeds via several lemmas.
The following lemma was inspired in part by
Lemma~1 in \cite{Gao:2009hf}.
%%although here we consider the multivariate setting and we also conclude with $L_\infty$ entropy bounds rather than $L_p$ ones.

\begin{lemma}
  \label{lem:3}
  Let $\mc F$ be a class of functions on $\prod^d_{i=1} [0,L_i]$, $0 < L_i<\infty$, let $x \in [0,1]^d$,
  and let
  $$\mc G:=
  \lb y \mapsto  \int_0^1 f(x + t(y-x)) dt : f \in \mc F \rb.$$
  Assume $\log\Nb( \epsilon, \mc F, L_\infty) \le \phi(\epsilon) < \infty$ for some function $\phi$ and all $\epsilon > 0$, and 
  assume further that the $\epsilon$-bracketing cover of $\mc F$ can be taken to satisfy
 \eqref{eq:82} with $A$ the identity (and where $h$ is replaced by the lower and upper bracket of $f$).
  Then there exists $0 < C < \infty$ such that
  \begin{equation*}
    \log \Nb[ \epsilon / \phi(\epsilon)^{1/d}, \mc G, L_\infty] 
    \le C \phi(\epsilon).
  \end{equation*}
\end{lemma}
\begin{proof}
  By \eqref{eq:81}, we will bound the metric covering number rather than the
  bracketing number, just for ease of notation.  Without loss of generality,
  assume $\phi(\epsilon)$ takes on integer values and take $x = 0$.  Let
  $\lb f_i \rb_{i=1}^{e^{\phi(\epsilon)}}$ be an $\epsilon$-$L_\infty$-net
  for $\mc F$.  For $f \in \mc F$ write
  $g(y) := \int_0^1 f(t y) dt = \int_0^1 (f(ty) - f_i(ty)) dt + g_i(y)$ where
  $g_i(y) := \int_0^1 f_i(ty)dt$ and $L_\infty(f -f_i) \le \epsilon$.  Define
  \begin{equation*}
    \mc G_i := \lb g(y) =  \int_0^1 (f(ty)-f_i(ty)) dt : y \in [0,1]^d, f \in \mc F, L_\infty(f - f_i) \le \epsilon \rb.
  \end{equation*}
  Thus $\mc G \subseteq \cup_i ( \mc G_i + g_i)$.

  Now, for each $i$, $\mc G_i$ consists of functions $g$ satisfying $L_\infty(g) 
  \le \epsilon,$ and also (by  \eqref{eq:82}) satisfying
  \begin{equation*}
    L_\infty \lp \pderiv{x_j} g \rp
    \le \epsilon
    \text{ for } j \in \lb 1, \ldots, d \rb.
  \end{equation*}
  Thus, by Theorem~\ref{thm:1} (in the appendix), we see that
  $\log N(\delta , \mc G_i, L_\infty) \le C (2 \epsilon / \delta)^d$ for a
  constant $C$ and any $\delta > 0$.  Take
  $\delta = \epsilon / \phi(\epsilon)^{1/d}$ and see
  $$\log N( \epsilon / \phi(\epsilon)^{1/d} , \mc G_i, L_\infty) \le 2 C  \phi(\epsilon),$$
  and let $g_{ij}$, $1 \le j \le e^{C \phi(\epsilon)}$, denote a corresponding cover.  Then $\lb g_i + g_{ij} \rb_{i,j}$ is an $L_\infty$-cover of $\mc G$ with size $(\epsilon/ \phi^{1/d}(\epsilon))$ and with cardinality no larger than $e^{(2C+1)\phi(\epsilon)}$, so we are done.
\end{proof}

\begin{remark}
  \label{rem:4}
  The above lemma depends on \eqref{eq:82}.  Note that in the $d=1$ case,
  this property is satisfied for the entire class $\cC(L, B, \Gb)$: see Lemma~\ref{lem:Gao-Wellner-extension}
  below.
\end{remark}

Let $D \subset \prod_{i=1}^d [0,L_i]$ be convex, and for simplicity assume $ 0 \in D$.
%% \commentC{(I guess this assumption of $0 \in D$ not necessary?)}
Let $0 < B$ and $\Gb := (\Gamma_1, \ldots, \Gamma_d)$.
For $m \ge 3$,  let
\begin{equation*}
  %% \mc G^m_d \equiv
  \mc G^{\mm}_d(D, B, \Gb)
  := \lb x \mapsto  \int_0^{1} \int_0^{z_1} \cdots \int_0^{z_{m-2}} f \lp s x \rp ds
  dz_{\mm-2} \cdots dz_1   : f \in \cC_d(D, B, \Gb)
  \rb,
\end{equation*}
be a class of functions, where $x \in D$.
Note that the functions in $\mc G^m_d(D,B, \Gb)$ are normalized so their size does not increase with the size of $D$.

\begin{lemma}
  \label{lem:2}
  Fix $D \subset \prod_{i=1}^d [0,L_i]$ be convex with $0 \in D$. Let $\Gb := (\Gamma_1, \ldots, \Gamma_d) \in (0, \infty)^d$ and $L := (L_1, \ldots, L_d) \in (0,\infty)^d$.
  %% Let $0 < B < \sum_{i=1}^d \Gamma_i L_i^{m-1}.$
  Let $0 < B \le \sum_{i=1}^d \Gamma_i L_i$.  Let $m \ge 2$ be an integer and $p \ge 1$.
  % If the bracketing cover of $\cC_d(D, B, \Gb)$ from Theorem~\ref{thm:GS-extension} can be taken to satisfy Property~\ref{property:bracketing-FTOC}
  Then,
  abbreviating $\mc G^m_d \equiv  \mc G^m_d(D, B, \Gb)$, 
  we have
  \begin{equation}
    \label{eq:2}
    \log \Nb( \epsilon \vol_d(D)^{1/p}, \mc G^m_d, L_p)     \le
    \log \Nb[\epsilon, \mc G^m_d% (D, B, \Gb)
    , L_\infty]
    \le c_m \lp \frac{ \sum_{i=1}^d \Gamma_i L_i }{\epsilon} \rp^{d/m}.
  \end{equation}
\end{lemma}

\begin{proof}
  The first inequality of \eqref{eq:2} is immediate.
  The proof of the second inequality is by induction.  We can start with the base case of $m = 2$ by identifying $ \mc G^2_d(D,B, \Gb)$ with $\cC_d(D, B, \Gb)$ and then the result is by Theorem~\ref{thm:GS-extension}.
  Now we assume the $m-1$ case holds, i.e., 
  \begin{equation}
    \label{eq:42}
    \log \Nb[\epsilon, \mc G^{m-1}_d(D, B, \Gb), L_\infty]
    \le c_{m-1} \lp \frac{ \sum_{i=1}^d \Gamma_i L_i }{\epsilon} \rp^{d/(m-1)},
  \end{equation}
  and show \eqref{eq:2} holds.
  By \eqref{eq:42} and Lemma~\ref{lem:3}, we have
  \begin{equation}
    \label{eq:52}
    \log \Nb[ \frac{\epsilon }{
      \lp \frac{\sum_{i=1}^d \Gamma_i L_i }{ \epsilon} \rp^{\inv{m-1}} },
    \mc G_d^m(D, B, \Gb), L_\infty ]
    \le C_m \lp \frac{ \sum_{i=1}^d \Gamma_i L_i}{\epsilon} \rp^{\frac{d}{ m-1}}
  \end{equation}
  which is equivalent to  \eqref{eq:2}.
          %           , because  $(\sum_{i=1}^d  L_i)(\sum_{i=1}^d \Gamma_i L_i) (\sum_{i=1}^d L_i)^{m-3}$ is equal to  \\
          %           $(\sum_{i=1}^d \Gamma_i L_i) (\sum_{i=1}^d L_i)^{m-2}.$
\end{proof}

\begin{proof}[Proof of Theorem~\ref{conj:m-mono-sup-bracket}]
  We consider (only) $f \in \cC^{\mm, \circ}_d(D, B, \Gb)$, by
  Remark~\ref{rem:m-mono-deriv}.  Since $D$ has nonempty interior there exists an open ball contained in $D$, which, by translation, we take to be $B(0, r_0)$ without loss of generality, for $r_0 > 0$.
  Now, by iterated application of the Fundamental Theorem of Calculus,
  for any $(m-1)$-times continuously differentiable $h \colon \RR \to \RR$,   
  we can write
      \begin{align}
        \label{eq:FOTC-expansion}
        h(x) =
        h(0)  
        %% + h'(0) x
        + \cdots
        + \frac{h^{(\mm-2)}(0)}{(\mm-2)!} x^{\mm-2}
        + \int_0^x \int_0^{z_1} \cdots \int_0^{z_{\mm-2}} h^{(\mm-1)}(s) ds dz_{\mm-2} \cdots dz_1.
      \end{align}
  By applying \eqref{eq:FOTC-expansion} to $t \mapsto f(ty)$, for any $y \in D$ we can write
  \begin{equation}
    \label{eq:53}
    \begin{split}
      f(y) =
      %% f(0) +
      \sum_{i=0}^{m-3} \sum_{\jb \in \II_i} \inv{\jb!} f^{(\jb)}(0)
      y^{\jb} +
      \sum_{\jb \in \II_{m-2}} {m-2 \choose \jb}
      y^{\jb}
      I_{m-2}(f^{(\jb)},y)
          %           \int_0^{1} \int_0^{z_1} \cdots
          %           \int_0^{z_{m-3}} f^{(\jb)} \lp s y \rp ds
          %           dz_{m-3} \cdots dz_{1},
    \end{split}
  \end{equation}
  where
  \begin{equation*}
    I_{m-2}(f^{(\jb)}, y) := \int_0^{1} \int_0^{z_1} \cdots
    \int_0^{z_{m-3}} f^{(\jb)} \lp s y \rp ds
    dz_{m-3} \cdots dz_{1}.
  \end{equation*}
  Let
  \begin{equation*}
    \mc P^{\mm} :=
    \lb y \mapsto \sum_{i=0}^{m-3} \sum_{\jb \in \II_i} a_{\jb} y^{\jb} : 0 \le a_{\jb} \le c_{\jb} \rb
  \end{equation*}
  where $c_{\jb} := K_{m}/ r_0^{m} \jb! $ and where $K_m := \max_j K_j$ comes from Lemma~\ref{lem:m-mono-deriv-bound}.

  Now, for $\ib \in \II_{m-2}$, let
  $\jb_\alpha(\ib) := \ib + (0,\ldots, 0, 1,0,\ldots, 0)$, where the $1$ is
  in the $\alpha$ index.  %% Thus $\jb_1, \ldots, \jb_d$ are the
  Let $\Gb_{\ib} := (\Gb_{\jb_1(\ib)}, \ldots, \Gb_{\jb_d(\ib)})$.  This is
  the vector of Lipschitz constraints for $f^{(\ib)}$.
  Let
  \begin{equation}
    \label{eq:63}
    \mc F^{\mm} :=
    \lb
    y \mapsto
    %% m^m
    \sum_{\jb \in \II_{m-2}}
    {m-2 \choose \jb}
    y^{\jb}
    I_{m-2}(g_{\jb}, y)
    :  g_{\jb} \in \cC_d(D, B, \Gb_{\jb})
    \rb.
  \end{equation}
  Then $\cC^{\mm}_d(D, B, \Gb) \subset \mc P^{\mm} + \mc F^{\mm}$
  so
  \begin{equation}
    \label{eq:80}
    N(\epsilon, \cC^{\mm}_d(D,B,\Gb), L_\infty)
    \le N(\epsilon/2, \mc P^{\mm}, L_\infty)
    N( \epsilon / 2, \mc F^{\mm}, L_\infty).
  \end{equation}
  Recall by \eqref{eq:81}, $L_\infty$-$\epsilon$-bracketing numbers equal $L_\infty$-$(\epsilon/2)$-covering numbers, and so simply for ease of notation and without any loss of generality, we form a $L_\infty$ (metric) cover rather than $L_\infty$ bracketing cover.

  First we form a cover for $\mc P^{\mm}$.  For an integer $N \ge 1$, we can construct a grid %% (as in \eqref{eq:71})
  to cover $\mc P^{\mm}$ by taking $a_{\jb} \in \lb c_{\jb} / N, \ldots, N c_{\jb} / N \rb.$ Since $\jb \in \II_{i}$, $0 \le i \le m$ takes on no more than $m^d$ values, the cover has cardinality $N^{m^d}.$ The $L_\infty$ size is $m^d C L^{\mu} / N$ where $C := \max_{\jb} c_{\jb}$, $L := b-a$, and $\mu:= (m-3, \ldots, m-3)$.  Take $N$ to be $\lceil \epsilon^{-1} m^d C L ^{\mu} \rceil$.  Then we have formed a cover for $\mc P^{\mm}$ in the $L_\infty$ norm with distances no larger than $\epsilon$ and log cardinality bounded above by $m^d \log (1 + m^d C L^{\mu} \epsilon^{-1})$.

  Now consider forming an
  $L_\infty$-cover for
  $\mc F^{\mm}$.
          %           where
          %           $$I_{m-2}(g_{\jb}) := \int_0^{1} \int_0^{z_1} \cdots
          %           \int_0^{z_{m-3}} g_{\jb} \lp s y \rp ds
          %           dz_{m-3} \cdots dz_{1}.$$
  By Lemma~\ref{lem:2}, for a fixed $\jb \in \II_{m-2}$, we can form an $\epsilon$-$L_\infty$-cover, $h_{\jb,i}$, for $i=1,\ldots, N_{\jb}$,  for the functions $I_{m-2}(g_{\jb})$, where
  $\log N_{\jb} \le c_m \lp \sum_{i=1}^d \Gamma_{\jb,i} L_i / \epsilon \rp^{d/m}$ and $L_i := b_i-a_i$.
          %           $$\int_0^{1} \int_0^{z_1} \cdots
          %           \int_0^{z_{m-3}} g_{\jb} \lp s y \rp ds
          %           dz_{m-3} \cdots dz_{1}.$$
  Let $f_{\jb,i}(y) := {m-2 \choose \jb} y^{\jb} h_{\jb,i}(y)$,
  $i=1, \ldots, N_{\jb}$.  Let $L= (L_1, \ldots, L_d)$.  Then for a function
  $f(y) = {m-2 \choose \jb} y^{\jb} I_{m-2}(g_{\jb})$, with
  $L_\infty(f_{\jb,i} - f) \le \epsilon$, we have
  $L_\infty( f_{\jb,i} - f) \le m^m L^{\jb} \epsilon$.  Equivalently, we can
  cover the same function class with size $\epsilon$ and log cardinality
  $c_{m,2} \lp \sum_{i=1}^d \Gamma_{\jb,i} L_i L^{\jb} / \epsilon \rp^{d/m}$.
  Let $g_{\jb, i}$ denote such a cover.  Then the class ${\mc F}^{\mm}$ is
  covered by the set of functions $\sum_{\jb \in \II_{m-2}} g_{\jb, i}$ which
  have $L_\infty$ distance bounded above by
  $ \sum_{\jb \in \II_{m-2}} \epsilon $ and log cardinality bounded above by
  $ c_{m,2} \sum_{\jb \in \II_{m-2}} \lp \sum_{i=1}^d \Gamma_{\jb,i} L_i
  L^{\jb} / \epsilon \rp^{d/m}.$
  Equivalently, $\mc F^{\mm}$ can be covered by a class with $\epsilon/2$ $L_\infty$-distance and log cardinality bounded above by $c_{m,3} \epsilon^{-d/m} (\max_{\jb \in \II_{m-1}} \Gamma_{\jb} L^{\jb})^{d/m}$.
  Then, by \eqref{eq:80}, we have
  \begin{equation*}
    \log     N(\epsilon, \cC^{\mm}_d(D,B,\Gb), L_\infty)
    \le
    m^d \log (1 + m^d 2 C L^{\mu} \epsilon^{-1})
    + c_{m,3} \epsilon^{-d/m} \lp \max_{\jb \in \II_{m-1}} \Gamma_{\jb} L^{\jb} \rp^{d/m}.
  \end{equation*}
  This completes the proof.
\end{proof}

\medskip
\noindent We can now prove the following bound on bracketing entropy of $\mm$-monotone function classes, using the same approach used to prove Theorem~\ref{thm:main-thm}.  We use the same $G_{\ib, \jb}$ partition construction as in the proof of  Theorem~\ref{thm:main-thm}, except that we modify the $\delta_i$'s.
First we have the following $m$-monotone version of Lemma~\ref{lem:Pij-bracketing-bound}.

\medskip

\begin{lemma}
  \label{lem:1}
  Let Assumption~\ref{assm:simple-polytope} hold.
  %% and assume   \eqref{eq:conj-k-mono-multivariate} holds.
  Fix $k \in \lb 1, \ldots, d \rb$, $\ib \in I_k$, $\jb \in J_k$.  Then for any $p \ge 1$ and $\epsilon > 0$,
  \begin{align*}
    % \MoveEqLeft
    \log \Nb[\epsilon \vol_d( G_{\ib,\jb})^{1/p}, \cC^{\mm}_d(D_0, 1)|_{G_{\ib,\jb}}, L_p ]
    %% \\
    %% &     
         \le c_{D,d}
         \lp
         \inv{\epsilon}
         \max_{\alpha=1,\ldots, k}
         \frac{\delta_{i_\alpha+1}^{m-1} }{\delta_{i_\alpha}^{m-1}}
         \rp^{d / \mm}.
  \end{align*}
\end{lemma}
\begin{proof}
  Let $e_{\ib,\jb, \alpha} \equiv e_\alpha$, $\alpha=1,\ldots,d$ be given
  by Proposition~\ref{prop:basis-lipschitz}, so that
  $d(G_{\ib,\jb}, \bd D_0, e_\alpha) \ge \max_{\beta \in 1,\ldots, k} \delta_{i_\beta} / |\la e_\beta, v_{j_\beta} \ra |$ for $\alpha = 1, \ldots, k$, and
  $d(G_{\ib,\jb}, \bd D_0, e_\alpha) \ge 2/u$ for $\alpha = k+1, \ldots, d$.
  Recall \eqref{eq:39} from the proof of
  Lemma~\ref{lem:Pij-bracketing-bound}, and recall that
  $\gamma_{\alpha} \le d L_{\jb,4}^{-1} \max_{\beta \in \lb 1, \ldots , k \rb } \delta_{i_\beta + 1}$.
  Let
  $\eta_0 := (\delta_{i_{\beta_1}}, \ldots, \delta_{i_{\beta_k}}, 2/u, \ldots, 2/u)$
  and $\eta_1 := (\delta_{i_{\beta_1}+1}, \ldots, \delta_{i_{\beta_k}+1}, 4L_{k,1} \rho_{\jb, k+1}, \ldots, 4L_{k,1} \rho_{\jb, d})$.
  For $\ve{l} \in \II_{m-1}$, let $\Gamma_{\bs l} := \sup_{ x \in G_{\ib,\jb}, f \in \cC^m_d(D_0,1)} |f^{(\ve l)}(x)|$.
  Then by
  Lemma~\ref{lem:m-mono-deriv-bound}
  we have
  \begin{equation*}
    \Gamma_{\ve l} \eta_1^{\ve l} \le c \frac{ \eta_1^{\ve l}}{ \eta_0^{\ve l}}
    \le c \max_{\beta \in \{ 1, \ldots, k\} } \frac{ \delta_{i_{\beta}+1}^{m-1}}{ \delta_{i_\beta}^{m-1}}.
  \end{equation*}
  % \begin{align*}
  %   \lv \sum_{{\ve l} \in \II_{m-1}} \Gamma_{{\ve l}} \eta_1^{{\ve l}} \rv
  %   = \lv \sum_{{\ve l} \in \II_{m-1}} \Gamma_{{\ve l}} \eta_0^{{\ve l}} \frac{\eta_1^{{\ve l}}}{\eta_0^{{\ve l}}} \rv
  %   \le \lv \sum_{{\ve l} \in \II_{m-1}} \Gamma_{{\ve l}} \eta_0^{{\ve l}}  \rv
  %   \max_{{\ve l} \in \II_{m-1}} \frac{\eta_1^{{\ve l}}}{\eta_0^{{\ve l}}}
  %   \le K \max_{{\ve l} \in \II_{m-1}} \frac{ \eta_1^{{\ve l}} }{ \eta_0^{{\ve l}}}
  %   \le c \max_{\beta \in \{ 1, \ldots, k\} } \frac{ \delta_{i_{\beta}+1}^{m-1}}{ \delta_{i_\beta}^{m-1}}.
  % \end{align*}
  Now we can apply Theorem~\ref{conj:m-mono-sup-bracket}. We use the fact
  that $G_{\ib, \jb}$ is embedded in a hyperrectangle $H$ with axes given by
  the orthonormal basis ${\ve e} := \lb e_\alpha \rb_{\alpha=1}^d$ specified
  by Proposition~\ref{prop:basis-lipschitz}.  Thus
  $\cC^{\mm}_d(D,1)|_{G_{\ib,\jb}}$ is contained in
  $\cC^{\mm}_d(H,1)|_{G_{\ib,\jb}}$.  Thus, letting
  $\Gb := \lb \Gb_{\bs l} : \bs l \in \II_{m-1} \rb$, we may apply
  Theorem~\ref{conj:m-mono-sup-bracket} to
  $\cC^{\mm}(H, 1, \Gb)|_{G_{\ib,\jb}}$ (after applying a rotation, see
  Remark~\ref{rem:2}).  By \eqref{eq:conj-k-mono-multivariate} (and the logic
  leading to \eqref{eq:Lip-bound-1}), the proof is complete.
\end{proof}

\medskip
\noindent We are now in a position to prove the following $\mm$-monotone version of Theorem~\ref{thm:main-thm}.

\medskip

\begin{theorem}
  \label{thm:main-extension-m-mono}
  Assume the setup and conclusions of \eqref{eq:conj-k-mono-multivariate} hold.  Fix $ p \ge 1$.  Then for all $\epsilon > 0$,
  \begin{equation}
    \label{eq:66}
    \log \Nb[\epsilon, \cC^{\mm}_d(D,B), L_p]
    \le C \epsilon^{-d/ \mm} \lp B \prod_{i=1}^d (b_i-a_i)^{1/p} \rp^{-d/ \mm}.
  \end{equation}
\end{theorem}

\begin{proof}
  The same scaling argument as given in (the beginning of) the proof of Theorem~\ref{thm:main-thm} applies here, since rescalings of $\mm$-monotone functions are still $\mm$-monotone.  Thus assume $D \subset [0,1]^d$ and $B = 1$.
  Now we define 
  \begin{equation}
    \label{eq:defn:u-k-mono}
    u \equiv u_D :=
    %% \exp \lp -2(p+1)^2  (p+2) \log 2 \rp
    r_D/2 \wedge
    2^{  - \mm (1 + p (\mm-1))^2 (2 + p (\mm -1 )) }
    \wedge
    \min_{k \in \lb 1,\ldots, d-1 \rb} \min_{ \jb \in J_k,
      e \in E_{\jb}
      %% \SPAN\lb  e_{\jb, k+1}, \ldots, e_{\jb,d} \rb
    }
    \frac{d^+(x_{\jb}, \relbd  G_{\jb}, e) }{ L_{k,2}}
  \end{equation}
  (where $L_{k,2}$ is still given by \eqref{eq:defn:Lk2}).
  We define $A$ and $\lb \delta_{i} \rb_{i=1}^A$ as in \eqref{eq:5}, by
  \begin{equation}
    \label{eq:67}
    \delta_i := \exp \lb p \lp \frac{p -1/(\mm-1)}{p+ 2/(\mm-1)}\rp^{i-1} \log \epsilon \rb
    \quad    \text{ for }  \quad
    i = 1, \ldots, A,
    \text{ and }
    \delta_0 = 0.
  \end{equation}
  For $k \in \lb 1, \ldots, d \rb$, $\ib \in I_k$ we let $a_{(i_1 \ldots i_k)} = 2$ if $i_\alpha = 0$ for any $\alpha \in \{1, \ldots, k\}$, and otherwise let
  \begin{equation}
    \label{eq:68}
    a_{(i_1, \ldots, i_k)}
    := \prod_{\beta=1}^k a_{i_\beta} := \prod_{\beta = 1}^k \epsilon^{1/k}
    \exp \lb - p \frac{(p+ 1/(\mm-1))^{i_{\beta}-2}}{(p+2/(\mm-1))^{i_\beta-1}} \log \epsilon \rb.
  \end{equation}
  When $k=0$, let $a_A := \epsilon / u$. Now define $a$ by \eqref{eq:13}, as before, and then
  \begin{equation}
    \label{eq:bracket-number-m-mono}
    \log \Nb[ a  , \cC^{\mm}_d(D, 1), L_p ]
    \le \sum_{k=0}^d \sum_{\jb \in J_k}
    \sum_{ \ib \in I_k} \log \Nb[a_{\ib} \vol_{d}(G_{\ib,\jb})^{1/p} , \cC^{\mm}_d(D,1)|_{G_{\ib,\jb}}, L_p].
  \end{equation}
  holds.
  We consider the case where $k \in \lb 1, \ldots,d  \rb$ (i.e., $k \ne 0$), and compute the sum above over $I_k$ for a fixed $\jb \in J_k$.
  We again use the trivial bracket $[-1,1]$ for any $G_{\ib,\jb}$ where $i_{\alpha} = 0$ for any $\alpha \in \lb 1,\ldots, k \rb$.
  By
  Lemma~\ref{lem:1}, the sum over the remaining terms is bounded above by
  \begin{align}
    %% \label{eq:bracketing-card-0}
    % \MoveEqLeft
    \sum_{i_1=1}^A \cdots \sum_{i_k=1}^A c_1  a_{\ib}^{-d/m}
    \lp
    \max_{\alpha=1,\ldots, k} \frac{ \delta_{i_\alpha +1}^{\mm-1}}{\delta_{i_\alpha}^{m-1}}   
    \rp^{d/\mm} % \\
  \end{align}
  which (using $\max_{\alpha \in \lb 1,\ldots, k\rb} 2 \delta_{i_{\alpha}+1} / \delta_{i_\alpha} \le \prod_{\alpha=1}^k 2 \delta_{i_\alpha+1} / \delta_{i_\alpha}$ as in the proof of Theorem~\ref{thm:main-thm}) is bounded above by
  \begin{equation}
    \label{eq:bracketing-card-_m-mono}
    c_2
    \sum_{i_1=1}^A \cdots \sum_{i_k=1}^A \prod_{\beta=1}^k \lp \frac{
      \delta_{i_\beta+1}^{\mm-1} }{\delta_{i_\beta}^{m-1}  a_{i_\beta} } \rp^{d/\mm}.
  \end{equation}
  %% (Recall, for the $k=d$ case,  we take the product over an empty set to be $1$.)
  The constants $c_1, c_2$ depend on $D$ and $d$.

  We now let
  \begin{equation}
    \label{eq:defn-zetai-m-mono}
    \zeta_i \equiv \zeta_{i,k,m}
    :=  \lp \epsilon^{1/k} \delta_{i+1}^{m-1} / (\delta_i^{m-1} a_i) \rp^{1/m}.
  \end{equation}
  By Lemma~\ref{lem:zetai-sum-m-mono} below,
  $\sum_{i=1}^A \zeta_i^d \le K_d < \infty$ for a constant $K_d$.
  Now,
  \begin{equation}
    \label{eq:72}
    \prod^d_{\alpha = k+1} \rho_{\jb, \alpha}^{\mm-1}
    \le C_d \vol_{d-k}(A_{\jb})^{\mm-1}
    \le C_d \vol_{d-k}(G_{\jb})^{\mm-1} .
  \end{equation}
  Thus we have shown that \eqref{eq:bracketing-card-_m-mono} is bounded above by $ K \epsilon^{-d/\mm}$.  where $K$ depends only on $D$, $d$, $\mm$, and $p$, but not on $\epsilon$.  Therefore, for a different constant $K$, we have shown that the right side of \eqref{eq:bracket-number-m-mono} is bounded above by $K \epsilon^{-d/ \mm}$ for all $\epsilon > 0$.

  We now bound the size of the brackets, $a$.  Define $I_k^+,$ $I_k^0$ as in the proof of
  Theorem~\ref{thm:main-thm}.
  Just as in  Theorem~\ref{thm:main-thm},
  \eqref{eq:ap} holds (using the current definitions of $\delta_{i_\alpha}$ and $a_{i_\alpha}$).
          %           We have
          %           \begin{equation}
          %           \label{eq:ap-m-mono}
          %           \begin{split}
          %           a^p  & \le
          %           a_A^p \vol_d(D)
          %           + \sum_{k=0}^d \sum_{\jb\in J_k, \ib \in I^0_k} 2^p \vol_{d}(G_{\ib,\jb}) \\
          %           & \quad +
          %           \sum_{k=1}^d (2L_{k,1})^{d-k} \sum_{\jb \in J_k} \vol_{d-k} (G_{\jb})
          %           \sum_{\ib \in I_k^+} a_{\ib}^p \prod_{\alpha=1}^k \frac{( \delta_{i_\alpha + 1} -
          %           \delta_{i_\alpha} )}{\la \tilde f_\alpha, v_{j_\alpha} \ra}.
          %         \end{split}
          %       \end{equation}
  For the middle term on the right side of \eqref{eq:ap}, recall \eqref{eq:74}, and for the first term recall $a_A := \epsilon / u$.  
  It remains only to bound the last term.
  We can check that $a_i^p \delta_{i+1} \le \epsilon^{p/k} \zeta_i^m $ (equality holding when $m=2$). Thus, 
  arguing as in the proof of   Theorem~\ref{thm:main-thm},
  we can see we need only bound
  $     \epsilon^{p/k}      \sum_{\alpha=1}^A \zeta_\alpha^{\mm} $ which
  % \begin{equation}
  %   \label{eq:75}
  %   \epsilon^{p/k}      \sum_{\alpha=1}^A \zeta_\alpha^{\mm} \le \epsilon^{p/k} A_u,
  % \end{equation}
  by Lemma~\ref{lem:zetai-sum-m-mono} is bounded above by
  $\epsilon^{p/k} A_u.$
  Thus $a \le C \epsilon$ for a constant $C$ not depending on $\epsilon$.  This completes the proof.
\end{proof}

\noindent In the $m$-monotone case, we did not relate the constants involved in the bound to the volumes of the faces of $D$ as explicitly as we did in the convex case.
The following lemma was used to bound both the cardinality and the size of the brackets in
Theorem~\ref{thm:main-extension-m-mono} above.

\begin{lemma}
  \label{lem:zetai-sum-m-mono}
  Define $A$, $u$, $\delta_i$, and $a_i$ by
  \eqref{eq:5},
  \eqref{eq:defn:u-k-mono}, \eqref{eq:67}, and \eqref{eq:68}, respectively.
  Assume $0 < \epsilon \le 1$.
  Let $\zeta_i$ be defined by \eqref{eq:defn-zetai-m-mono}.  Then for any $\gamma \ge 1$,
  \begin{equation*}
    \sum_{i=1}^A \zeta_i^\gamma \le 2^\gamma / (2^\gamma - 1).
  \end{equation*}
\end{lemma}
\begin{proof}
  Straightforward algebra shows
  \begin{equation*}
    \zeta_i = \exp \lb  \frac{(\mm-1) p
      \lp \frac{p + \inv{\mm-1}}{ p + \frac{2}{\mm-1}} \rp^{i}}{
      \mm(1 + (\mm-1)p)^2}
    \log \epsilon \rb.
  \end{equation*}
  %% See mathematica ('zeta-sum-v2.nb').
  Thus, for $\alpha = 1, \ldots, A-1$,  further algebra shows
  \begin{equation*}
    \frac{\zeta_\alpha^{\mm}}{\zeta_{\alpha+1}^{\mm}}
    = \exp \lb 
    (\mm-1) p \frac{(1 + p(\mm-1))^{\alpha-2}}{( 2 + p( \mm -1))^{\alpha+1}}
    \log \epsilon
    \rb,
  \end{equation*}
  \begin{mylongform}
    \begin{longform}
      which equals
      \begin{equation*}
        \exp \lb 
        \frac{(\mm-1) p}{(1 +p(\mm-1))^{2} (2 + p (\mm-1)) } 
        \frac{(1 + p(\mm-1))^{\alpha}}{( 2 + p( \mm -1))^{\alpha}}
        \log \epsilon
        \rb
      \end{equation*}
    \end{longform}
  \end{mylongform}
  which (since $\alpha \le A-1$) is bounded above by
  \begin{align*}
    \MoveEqLeft
    \exp \lb 
    \frac{(\mm-1) p}{(1 +p(\mm-1))^{2} (2 + p (\mm-1)) } 
    \frac{(1 + p(\mm-1))^{A-1}}{( 2 + p( \mm -1))^{A-1}}
    \log \epsilon
    \rb \\
    & \le \exp \lb \frac{\log u}{(1 +p(\mm-1))^{2} (2 + p (\mm-1))} \rb
      =: R^{-\mm}.
  \end{align*}
  Now, note that $\zeta_{A} \le 1$ (since $\epsilon \le 1$), and
  by its definition \eqref{eq:defn:u-k-mono}, %% as opposed to earlier def'n
  \begin{mylongform}
    \begin{longform}
      $u \le \exp ( - \mm (1 + p (\mm-1))^2 (2 + p (\mm -1 )) \log 2 )$,      so
    \end{longform}
  \end{mylongform}
  we see $R \ge 2$.
  The rest of the proof follows in fashion similar to the proof of
  Lemma~\ref{lem:zetasum}.
\end{proof}

%% \subsubsection{Extending and modifying results from \cite{Gao:2009hf}: $\mm$-monotonicity with
\subsubsection{Univariate $\mm$-monotonicity}
\label{sec:mm-mono-d=1}

In this subsection, we prove that when $d=1$,
\eqref{eq:82} holds, and so
the conclusion of Theorem~\ref{conj:m-mono-sup-bracket} holds with $\cC^*$ replaced by the full class $\cC$.
We point out that \cite{Gao:2009hf} provide bracketing entropy for classes of bounded (univariate) $\mm$-monotone functions on a compact interval in Hellinger distance, but their result does not immediately give a bound for $L_\infty$-bracketing entropy.  The methods of the previous section do give such a bound though.
% To get a bound in the strong $L_\infty$ distance, we must consider smaller classes of $\mm$-monotone functions than those considered in \cite{Gao:2009hf}. \cite{Gao:2009hf} consider $\mm$-monotone functions $f$ satisfying a uniform bound.
% We consider $\mm$-monotone $f$ satisfying a Lipschitz bound on  $f^{(\mm-2)}$.
% %% (which is a bound on $f^{(\mm-1)}$ when $f^{(\mm-1)}$ exists).
% Recall, as discussed earlier, in the convex case ($\mm=2$), such a Lipschitz assumption is necessary for $L_\infty$-entropy bounds.
%% necessary ? where ? 
% The proof of the following lemma is related to the proof of
% Lemma~1 in \cite{Gao:2009hf}.
% Recall $N(\cdot, \cdot, \cdot)$ is the
% (metric) covering number discussed and defined in the introduction.  Also,
Recall $\lambda$ is Lebesgue measure.

\begin{lemma}
  \label{lem:Gao-Wellner-extension}
  Let $\mc F$ be a class of functions on $[0,L]$ and let $\mc G : = \lb x\mapsto \int_0^x f d\lambda : f \in \mc F \rb$ be the class of primitives of $\mc F$ on $[0,L]$.  Assume $\log \Nb ( \epsilon, \mc F, L_\infty) \le \phi(\epsilon) < \infty$ for a function $\phi$ and $\epsilon > 0$.    Then there exists $0 < C< \infty$ such that
  \begin{equation}
    \label{eq:58}
    \log \Nb[\epsilon / \phi(\epsilon), L^{-1} \mc G, L_\infty ] \le C \phi(\epsilon).
  \end{equation}
\end{lemma}
\begin{proof}
  By \eqref{eq:81}, we will bound the metric covering number rather than the bracketing number, just for ease of notation.
  We take $L=1$, by rescaling: if $\mc F$ and $\mc G$ are classes of
  functions defined on $[0,1]$, let $\tilde{\mc F}$ be
  $\lb x\mapsto f(x L) : f \in \mc F \rb$ and define
  $\tilde{\mc G} := \lb x \mapsto L g(x L) : g \in \mc G\rb$.  Then
  $N(\epsilon, \mc F, L_\infty) = N(\epsilon, \tilde{\mc F}, L_\infty)$ and
  $\tilde {\mc G}$ is the class of primitives of $\tilde {\mc F}$.  We see that
  $\Nc[\epsilon, \mc G, L_\infty] = \Nc[\epsilon, L^{-1} \tilde{\mc G},
  L_\infty]$ so we can take $L = 1$.  Now, by the Fundamental Theorem of Calculus,
  \eqref{eq:82} holds, and so we can apply Lemma~\ref{lem:3}.
  This completes the proof.
\end{proof}

\medskip

\noindent
Now let $\mc G^1 \equiv \mc G^1(L,B)$ be the class of non-decreasing functions $f$ on $[0,L]$ satisfying $0 \le f \le B$, and let 
\begin{equation}
  \label{eq:59}
  \mc G^k \equiv
  \mc G^k(L,B)
  := \lb  g(x) = \int_0^x \int_0^{z_1} \cdots \int_0^{z_{k-2}} f(s) ds dz_{k-2} \cdots dz_1
  %% \,
  : 
  f \in \mc G^1
  \rb
\end{equation}
where $g$ is defined on $[0,L]$.
Note, when $f$ is continuous, then $0 \le  g^{(k-1)} \le B$.
\begin{lemma}
  \label{lem:k-mono-main-lemma}
  Fix $L, B>0$ and define $\mc G^k(L,B)$ by \eqref{eq:59}. Let $k \ge 2$ be an integer and $p \ge 1$. We have
  \begin{equation}
    \label{eq:61}
    %% \log N(\epsilon L , \mc G^k(L,B) , L_1)
    \log N(\epsilon L^{1/p} , \mc G^k(L,B) , L_p)
    \le \log \Nb[\epsilon, \mc G^k(L,B) , L_\infty]
    \le c_k \lp \frac{ B L^{k-1}}{\epsilon} \rp^{1/k}.
  \end{equation}
\end{lemma}
\begin{proof}
  This follows from using Lemma~\ref{lem:Gao-Wellner-extension} in the proof
  of Lemma~\ref{lem:2} (i.e., from the fact that \eqref{eq:82} is satisfied
  when $d = 1$).
\end{proof}
\begin{mylongform}
  \begin{longform}

    \commentC{(Here is the proof for the above lemma~\ref{lem:k-mono-main-lemma} that I originally wrote out.  Now it is subsumed by the multivariate proof.)}

    \begin{proof}
      The first inequality in \eqref{eq:61} is immediate.  For ease of notation,
      we write $\mc G^k$ for $\mc G^k(L,B)$ in this proof.  The proof is by
      induction starting from $k=2$.  For the base case of $k = 2$, by
      Theorem~\ref{thm:GS-extension}, we have
      \begin{equation}
        \label{eq:62}
        %% \log N( \epsilon L, \mc G^2, L_1)    \le
        \Nb[\epsilon, \mc G^2, L_\infty]
        \le c \lp \frac{BL}{\epsilon} \rp^{1/2}.
      \end{equation}
      (Note that the functions in $\mc G^2(L,B)$ are uniformly bounded by $BL$.)

      Now assume that \eqref{eq:61} holds with $\mm$ replaced by $\mm-1$.  We
      show it holds for $\mm$.  By the induction hypothesis,
      $\log N(\epsilon , \mc G^{\mm-1} , L_\infty) \le c_{\mm-1} ( BL^{\mm-2} /
      \epsilon)^{1/(\mm-1)} := \phi(\epsilon)$ for all $\epsilon > 0$.
      Thus, by Lemma~\ref{lem:Gao-Wellner-extension}, 
      \begin{equation*}
        \log \Nb[L \epsilon / \phi(\epsilon), \mc G^\mm, L_\infty ]
        \le \phi(\epsilon),
      \end{equation*}
      Substituting $\epsilon = ( c_{\mm-1} (B/L)^{1/\mm} u^{(\mm-1)/\mm}$, we get

      \commentC{(Some algebra:)
        \begin{equation*}
          c_{\mm-1} ( BL^{\mm-2} / (c_{\mm-1} u^{(\mm-1)/\mm} (B/L)^{1/\mm} )^{1/(\mm-1)}
          = c_\mm B^{1/\mm} (L^{\mm-2} L^{1/\mm} )^{1/(\mm-1)} u^{-1/\mm}
        \end{equation*}
        which equals $c_\mm (BL^{\mm-1} / u)^{1/\mm}$ as desired (since $(\mm-2 + 1/\mm)/(\mm-1) = \mm-1 - (\mm-1)/\mm) / \mm-1$).

        (Old algebra below
        \begin{align*}
          c_k ( BL^{k-1} /  (  (BL^{k-1})^{1/(k-1)} u)^{(k-1)/k} )^{1/(k-1)}
          = c_k ( (BL^{k-1})^{(k-1)/k} )^{1/(k-1)} u^{-1/k}.) 
        \end{align*}}
      \begin{equation}
        \label{eq:65}
        \log \Nb[u, \mc G^\mm, L_\infty ]
        \le c_\mm \lp \frac{B L^{\mm-1}}{u} \rp^{1/\mm}
      \end{equation}
      for all $u > 0$, as desired.
    \end{proof}
  \end{longform}
\end{mylongform}
%% We do not have any result for the $\mm = 1$ (monotone) case.
%% , for which we do not believe the result holds.

For $L, B, \Gamma > 0$, let $\mc C^\mm(L, B, \Gamma)$ be the class of $\mm$-monotone functions
(per Definition~\ref{defn:k-monotone-d=1}) $f $ on $[0,L]$ satisfying $0 \le f \le B$ and $f^{(\mm-2)}$ is Lipschitz with constant $\Gamma$.
%% Here $\mm$-monotone means $f^{(i)}$ exists for $1 \le i \le k-2$ and $(-1)^i f^{(i)} \ge 0$.
\begin{theorem}
  \label{thm:k-mono-entropy}
  Let $B, L$, $\Gamma > 0$.  Let $\mm \ge 2$ be an integer.  Assume
  $B < \Gamma L^{\mm-1}$.  Then there exists a constant $c_{\mm} > 0$ (not depending on $B, L , \Gamma,$ or $\epsilon$) such that for all $\epsilon > 0$,
  \begin{equation}
    \label{eq:64}
    \log \Nb[\epsilon, \mc C^{\mm}(L,B, \Gamma), L_\infty ]
    \le c_\mm \lp \frac{\Gamma L^{\mm-1}}{\epsilon} \rp^{1/\mm}.
  \end{equation}
\end{theorem}
\begin{proof}
  This follows from Theorem~\ref{conj:m-mono-sup-bracket} together with the
  fact that \eqref{eq:82} is satisfied when $d = 1$.
\end{proof}

\begin{mynotes}
  \begin{notes}
    This won't be quite right: don't have bound $B$ involved!  it has to show up somewhere., in log term i guess.

    The following is inspired by / modeled off section 2.1 of Gao and Wellner (2009).

    The point of confusion for me is that the dependence on $B$ here is
    different than what Guntuboyina and Sen get.  However there is no lower
    bound corresponding to their upper bound (or no lower bound with explicit
    dependence onc onstants) so their dependence could be wrong.
    
  \end{notes}
\end{mynotes}

\begin{mylongform}
  \begin{longform}

    \commentC{(Original proof of Theorem~\ref{thm:k-mono-entropy} below.  Now subsumed by multivariate case.)}

    \begin{proof}
      For notational ease, instead of $\mc C^{\mm}(L, B, \Gamma)$ we consider the class $\mc D^{\mm}$ (dropping dependence on $L,B,\Gamma$ for this proof) of functions
      $ f(u) = g(L-u)$ for any $g \in \mc C^{\mm}(L,B,\Gamma)$.  Clearly $\mc D^{\mm}$ and $\mc C^{\mm}(L,B,\Gamma)$ have the same bracketing entropy, and $f \in \mc D^{\mm}$ satisfies $f^{(i)} \ge 0$, whenever $f^{(i)}$ is defined.  Additionally, because infinitely differentiable functions are dense in $\mc D^{\mm}$ \citep{Gao:2008wf}, we will consider only functions $f$ that are infinitely differentiable.  Now, by iterated application of the Fundamental Theorem of Calculus, for any $f \in \mc D^{\mm}$ we can write
      \begin{align}
        \label{eq:FOTC-expansion}
        f(x) =
        f(0)  
        %% + f'(0) x
        + \cdots
        + \frac{f^{(\mm-2)}(0)}{(\mm-2)!} x^{\mm-2}
        + \int_0^x \int_0^{z_1} \cdots \int_0^{z_{\mm-2}} f^{(\mm-1)}(s) ds dz_{\mm-2} \cdots dz_1.
      \end{align}
      Note  that, by \eqref{eq:FOTC-expansion},
      $ f^{(i)}(0) L^i / i! \le f(L) \le B$ for all $1 \le i \le \mm-2$.
      Let
      \begin{equation*}
        \mc P_{\mm} := \lb x \mapsto a_0 + a_1 x + \cdots + a_{\mm-2} x^{\mm-2} :  0 \le a_i \le B/L^i \rb.
      \end{equation*}
      By assumption $0 \le f^{(\mm-1)} \le \Gamma$ and $f^{(\mm-1)}$ is non-decreasing.
      Thus, $\mc D^\mm \subseteq \mc P_\mm + \mc G^\mm(L, \Gamma)$ and so
      \begin{equation}
        \label{eq:bracketing-split}
        \Nc[\epsilon, \mc D^\mm, L_\infty]
        \le \Nb[ \epsilon / 2, \mc P_\mm, L_\infty] \cdot
        \Nc[ \epsilon / 2 , \mc G^\mm(L,\Gamma) , L_\infty].
      \end{equation}
      For an integer $N \ge 1$, note that
      \begin{equation}
        \label{eq:71}
        \lb a_0 
        %% + a_1 x
        + \cdots + a_{\mm-2} x^{\mm-2} :
        a_i \in \{B/L^i N, 2B/L^i N, \ldots, N B/L^i N \}, 0 \le i \le \mm-2 \rb 
      \end{equation}
      forms an $L_\infty$ cover for $\mc P_\mm$ of cardinality $N^{\mm-1}$ and of $L_\infty$ size bounded by $(\mm-1) B / N$.  Take $N$ to be $\lceil 2 \epsilon^{-1} \mm B \rceil$ (where $\lceil a \rceil$ is the least integer larger than $a$).  Then we have shown
      \begin{equation}
        \label{eq:70}
        \log N( \epsilon/2 , \mc P_\mm, L_\infty  ) \le \mm \log ( 1 + 2 \mm B / \epsilon).
      \end{equation}
      By Lemma~\ref{lem:k-mono-main-lemma}, $\log \Nc[\epsilon / 2, \mc G^{\mm}(L,\Gamma), L_\infty ] \le c (\Gamma L^{\mm-1} / \epsilon)^{1/\mm}$.
      Combining these two bounds, via  \eqref{eq:bracketing-split}, we get
      \begin{equation}
        \label{eq:bracketing-tight-bound}
        \log     \Nc[\epsilon, \mc D^\mm, L_\infty]
        \le \mm \log ( 1 + 2 \mm B / \epsilon) +
        c (\Gamma L^{\mm-1} / \epsilon)^{1/\mm}.
      \end{equation}
      Since $B \le \Gamma L^{\mm-1}$, %% $\mm \log ( 1 + 2 \mm B / \epsilon)$ is
      \eqref{eq:bracketing-tight-bound} can be bounded by
      $c_\mm (\Gamma L^{\mm-1} / \epsilon)^{1/\mm}$ where $c_\mm$ does not depend on $\epsilon$, $B$
      $\Gamma$, or $L$. Thus we are done, since
      $    \Nb[\epsilon, \mc D^\mm, L_\infty] \le 
      \Nc[\epsilon/2, \mc D^\mm, L_\infty].$
    \end{proof}

    \begin{remark}
      If $B > \Gamma L^{\mm-1}$ then we can replace $\Gamma L^{\mm-1}$ on the right side of \eqref{eq:65} by $B$ (i.e., by increasing $\Gamma$ so $\Gamma L^{\mm-1} = B$).  However, in this case the bound \eqref{eq:bracketing-tight-bound} gives better dependence on $B$ so can be used instead.
    \end{remark}

  \end{longform}
\end{mylongform}

\begin{mylongform}
  \begin{longform}
    Note: in Gao and Wellner (2009), they consider Hellinger bracketing for \mm-monotone classes.  The classes are bounded, but they do {\it not} bound the ($\mm$-1)-derivative.  I bound the latter, which is stronger (and I think is needed, in order to get sup-norm bounds).
  \end{longform}
\end{mylongform}

\subsection{Entropy of classes related to level set estimation}
\label{sec:entr-level-set}

% We have so far considered the (bracketing) entropy of classes of convex
% functions.
Now, we consider the entropy of certain classes of functions related to estimating the level sets of convex functions; we may consider, for instance, estimating the level set of a convex or concave regression function, or of a so-called log- or $s$-concave density.\footnote{As shown by
\cite{Doss-Wellner:2013}
in the univariate log- or $s$-concave cases 
and \cite{Kim:2014wa}  in the multivariate log-concave case,
bracketing entropies of log- or $s$-concave density classes are related to those of bounded concave function classes.}
We refer to \cite{Doss-Wellner:2013} for the definition of log- or $s$-concavity.
% A density is log-concave if it is of the form $e^{\varphi}$ for a concave function $\varphi$; the class of log-concave densities is a natural nonparametric alternative to the class of Gaussian densities, see the references given in the introduction.  The classes of $s$-concave densities are alternative classes that allow for heavier tails than the log-concave class does.
We are specifically concerned with the case when the level set is a polytope.

In the present paper, we are concerned with bracketing entropy bounds, rather than statistical methodological developments.  We provide here an extremely brief discussion of methodology to motivate the classes we are developing bounds for.  The methodology is to first pick a bandwidth $h > 0$, then to minimize an objective function $Q$, based on i.i.d.\ data points, over functions of the form
$f|_{\mc L(f) + B(0,h)}$
for convex $f$.
%% The objective function $Q$ needs to be such that if $Q(f,f_0)
The bandwidth $h$ will converge to $0$ as the sample size increases.
The model will need to satisfy some regularity conditions for M-estimation theory to apply
(it must be such that if
$H(p_f, p_{f_0}) \le \delta$
then $l_H(\mc L(f),D_0) \le C \delta$ where $H^2(p_f,p_{f_0}) := \int ( \sqrt{p_{f}} - \sqrt{p_{f_0}})^2 d\mu$ is Hellinger distance for some dominating measure $\mu$ on the sample space, and $p_f$ is the data generating density corresponding to convex function $f$).
% This is a very brief description and we do not have space here to develop methodology in greater detail. The focus here is just to develop bracketing bounds that may be useful in understanding rates of convergence in a variety of convex level set estimation settings.
% We do not have space to provide the full motivational background.  However, in general, in a nonparametric function estimation context (e.g., either the concave/convex regression setting or the log- or $s$-concave density estimation setting), one will have consistent estimators of the unknown function and this will allow one to restrict attention to a fixed set of support on which the convex/concave function is bounded.  Thus we consider this setup here.

We now proceed to find bracketing entropy bounds for a function class that
will govern rates of convergence for the above procedure, specifically when the level set is a polytope.  We operate under
the following basic setup or assumption.

\begin{assumption}
  \label{assm:LS-basic}
  Let $C_0$ be a closed, bounded convex set in $\RR^d$ with nonempty interior.
  Let $f_0 \in \cC(C_0, B)$ %%  satisfy $L_\infty(f_0) < \infty$, and
  satisfy
  $\inf_{x \in C_0} f_0(x) < \inf_{x \in \bd C_0} f_0(x)$.  Let $\lambda \in \RR$ satisfy
  $\inf_{x \in C_0} f_0(x) < \lambda < \sup_{x \in C_0} f_0(x)$.
  Assume further that $D_0 := \cL(f_0)$ is a polytope.
\end{assumption}

\noindent
The assumption restricts attention to functions which attain their minimum on the interior of $C_0$ and are strictly larger everywhere on their boundary than the minimum. (This is somewhat analogous to the assumption of so-called ``coercivity,'' except that we are restricting attention to a compact domain $C_0$.)
For a function $f$,  define $\mc L(f) \equiv \mc L_\lambda(f) := \lb x \in \RR^d :  f(x) = \lambda \rb$.
% and
%% let $D_0 := \cL(f_0)$.
%% Let $l_1(D,C)$ be the Lebesgue measure of the symmetric difference $D \Delta C := \lb  \rb.$
For two sets $C, D \subset \RR^d$, define the Hausdorff distance between them by
\begin{equation*}
  l_H(C,D) := \max \lp \sup_{x \in D} \inf_{y \in C}  \Vert x - y \Vert,
  \sup_{y \in C} \inf_{x \in D}  \Vert x - y \Vert \rp.
\end{equation*} 
Let $\mc S_\delta := \lb D : l_H(D, D_0) \le \delta \rb$.  Define set addition
$A_1+A_2 := \lb a_1+a_2 : a_1 \in A_1, a_2 \in A_2 \rb$ and recall that $f|_{A_1}$ is the restriction of a function $f$ to the set $A_1$.
For $h > 0$, the class of functions we consider is
\begin{equation}
  \label{eq:60}
  \mc C_{\delta,h}(C_0,B)
  := \lb f|_{\mc L(f) + B(0, h)} : f \in \cC(C_0,B), \mc L(f) \in S_\delta \rb;
\end{equation}
this is a class of bounded convex functions on $C_0$ restricted to a neighborhood about their $\lambda$-level set (which is generally not a convex set).

% Note that if $\delta $ and $h$ are small enough then functions in $\cCdh$ are automatically uniformly Lipschitz by the same basic argument used in the proof of Theorem~\ref{thm:main-thm}.

% \begin{assumption}
%   \label{assm:LS-polytope}
%   Assume that $D_0 := \cL(f_0)$ is a polytope.
% \end{assumption}
% \noindent

Let $F_{j}$, $j=1, \ldots, N$, be the facets of $D_0$. 
Let $T_j := F_j + B(0, \delta+h)$. Note that $T_j$ is a convex set.  Thus for any $f \in \cC(C_0,B)$ with
$l_H( \cL(f), D_0) \le \delta$, we have
$\cL(f) + B(0,h) \subseteq \cup_{j=1}^N T_j$.  Thus for $\epsilon>0$,
restating \eqref{eq:bracketing-cover}, we have
\begin{equation}
  \label{eq:54}
  \Nb[ N^{1/p} \epsilon, \cC_{\delta,h}(C_0,B), L_p ]
  \le  \prod_{j=1}^N \Nb[ \epsilon, \cC_{\delta,h}(C_0,B)|_{T_j}, L_p ].
\end{equation}
To bound the terms on the right side of the above display we will use
Theorem~\ref{thm:GS-extension}.  To do so, we need to compute $\Vol_d(T_j)$,
we need to find a hyperrectangle containing $T_j$, and we need to show
that $\cC_{\delta,h}(C_0,B)|_{T_j}$ is Lipschitz, which we will do by using
an idea from the proof of Theorem~\ref{thm:main-thm}.

By Assumption~\ref{assm:LS-basic}, $0 < l_H(D_0,C_0)$.  Let $\eta_0 := l_H(D_0,C_0)$.  Thus for any $f \in \cC_{\delta,h}(C_0,B)$, for $\delta, h$ small enough, for any $j \in \lb 1, \ldots, N \rb$,  $l_H(T_j, C_0) \ge \eta_0 - \delta - h > \eta_0 / 2$.  Thus
by Lemma~\ref{lem:basic-lemma_conv-lipschitz},
%% (by the logic used in the proof of Lemma~\ref{lem:Pij-bracketing-bound}, see \eqref{eq:C-inclusion-2}),
\begin{equation}
  \label{eq:56}
  \cC_{\delta,h}(C_0,B)|_{T_j} \subset \cC(T_j, B, \Gb)
\end{equation}
where $\Gb := (4B / \eta_0, \ldots, 4B / \eta_0).$
Now, let $V_{0,j} := \Vol_{d-1}(F_j)$.  Then
\begin{equation}
  \label{eq:57}
  \Vol_{d}(T_j) \le 2 V_{0,j} \cdot 2 (\delta + h).
\end{equation}
Next, note each facet $F_j$ is compact so can be embedded in a hyperrectangle.  Let $\prod_{i=1}^{d-1} [a_{j,i}, b_{j,i}]$ be a hyperrectangle of minimum volume containing $F_j$ (after an orthogonal rotation).  Then  by its definition, $T_j$ is (after rotation) contained in 
% \begin{equation}
%   \label{eq:55}
%   %%   T_j  \subset
%   \lp \prod_{i=1}^{d-1} [a_{j,i} - \delta - h, b_{j,i} + \delta + h] \rp
%   \times [ - \delta - h, \delta + h]
% \end{equation}
$  \lp \prod_{i=1}^{d-1} [a_{j,i} - \delta - h, b_{j,i} + \delta + h] \rp
\times [ - \delta - h, \delta + h],$
for $\delta, h > 0$ small enough.
Thus by Theorem~\ref{thm:GS-extension}, for $1 \le p < \infty$, for any $j \in \{ 1, \ldots, N \}$, and for $\delta, h$ small enough,
$  \log \Nb[\epsilon, \cC(T_j, B, \Gb), L_p]$ is bounded above by
\begin{align*}
  \MoveEqLeft c \lp \frac{ \Vol_d(T_j)^{1/p}}{\epsilon} \rp^{d/2}
  \lp B + \frac{ 4 B}{\eta_0} \lp 2 (\delta+h) + \sum_{i=1}^{d-1}  (2( \delta + h) + b_{j,i} 
  - a_{j,i} ) \rp \rp^{d/2} \\
  & \le    c_2 \lp \frac{ V_{0,j}^{1/p} (\delta+h) }{\epsilon} \rp^{d/2}
    \lp B + \frac{ 4 B}{\eta_0} \lp  \sum_{i=1}^{d-1} 2 (b_{j,i}-a_{j,i}) \rp \rp^{d/2}.
\end{align*}
Let $U_{0,j} := \sum_{i=1}^{d-1} (b_{j,i}-a_{j,i})$.  Then the right side of
the above display is bounded above by
$ c_3 \lp V_{0,j}^{1/p} (\delta+h) / \epsilon \rp^{d/2} \lp B +
B U_{0,j} / \eta_0 \rp^{d/2}.$ Thus the log of the right of
\eqref{eq:54} is bounded above by
$ \sum_{j=1}^N c_3 \lp  V_{0,j}^{1/p} (\delta+h) /\epsilon \rp^{d/2}
\lp B + B U_{0,j} / \eta_0 \rp^{d/2}.$ Thus we have shown the
following.

\begin{theorem}
  \label{thm:applications:level-set}
  Let Assumption~\ref{assm:LS-basic} hold.
  Fix $\delta>0,$ $h>0$, and let $\mc{C}_{\delta,h}(C_0,B)$ be defined as in \eqref{eq:60}.  Then
  \begin{equation*}
    \log   \Nb[ \epsilon, \cC_{\delta,h}(C_0, B), L_p ]
    \le S_0 \lp \frac{\delta + h}{\epsilon} \rp^{d/2},
  \end{equation*}
  where $S_0$ is a constant depending on $B$, $f_0$, and $C_0$.
\end{theorem}

\noindent This result suggests that
%%% adaptation to
fast rates of convergence may be possible when estimating polytopal level sets of convex functions.  As mentioned above, we do not here develop a full estimation procedure.  To do so will require studying the optimal choice of the bandwidth $h$.  If we can take $h = \delta$ for instance, then the bound is of order $(\delta / \epsilon)^{d/2}$.  When $d=2$ or $3$, one can then compute the entropy integral (see, e.g., \cite{vanderVaart:1996tf}, \cite{Dudley:1999dc}, \cite{vandeGeer:2009tk}), $ \int_{0}^{\delta} \sqrt{ \log \Nb[ \epsilon, \cC_{\delta,h}(C_0, B), L_2] } d\epsilon = \delta^{d/4} \int_{0}^\delta \epsilon^{-d/4} \, d\epsilon$ and see that it is of order $\delta^{d/4} \delta^{1 - d/4} = \delta$.  This corresponds, at least heuristically, to a $\sqrt{n}$ rate of convergence (the rate $\sqrt{n}$ arises from combining, e.g., Lemma 3.4.2 (p.\ 324) and Theorem~3.2.5 (p.\ 289) of \cite{vanderVaart:1996tf}).  These calculations are only suggestive in nature (and indeed we have not formally proposed an estimator in a specific model!).  They are presented to explain potential repercussions of Theorem~\ref{thm:applications:level-set}.

% \section{Discussion}
% \label{sec:discussion}

%% \subsection*{Appendix: Technical Tools}
\subsection*{Appendix}
\label{sec:appendix_technical-tools}

% \begin{theorem}[John's theorem, \cite{John:1948vs}; \cite{Ball:1997ud}, Theorem 3.1]
%   \label{thm:FritzJohn}
%   Each convex body $K \subset \RR^d$ contains a unique ellipsoid of maximal volume.  This ellipsoid is the Euclidean unit ball  $B_d(0,1)$
%   if and only if the following conditions are satisfied: $B_d(0,1) \subset K$ and (for some $m$) there are Euclidean unit vector $(u_i)_{i=1}^m$ on the boundary of $K$ and positive numbers $(c_i)_{i=1}^m$ satisfying $\sum_i c_i u_i = 0$ and $\sum_i c_i \la x, u_i \ra^2 = \| x \|^2$ for all $x \in \RR^d$.  If the conditions are satisfied, then $\|x \| \le d$ for all $x \in K$.
% \end{theorem}
\begin{theorem}[John's theorem, \cite{John:1948vs}; Theorem 13.4.1 \cite{Matouvsek:2002df}]
  \label{thm:FritzJohn}
  Let $K \subset \RR^d$ be a bounded closed convex body with nonempty
  interior.  Then there exists an ellipsoid $E$ of maximal volume such that
  $E \subseteq K \subseteq n E$.
  % Each convex body $K \subset \RR^d$ contains a unique ellipsoid of maximal volume.  This ellipsoid is the Euclidean unit ball  $B_d(0,1)$
  % if and only if the following conditions are satisfied: $B_d(0,1) \subset K$ and (for some $m$) there are Euclidean unit vector $(u_i)_{i=1}^m$ on the boundary of $K$ and positive numbers $(c_i)_{i=1}^m$ satisfying $\sum_i c_i u_i = 0$ and $\sum_i c_i \la x, u_i \ra^2 = \| x \|^2$ for all $x \in \RR^d$.  If the conditions are satisfied, then $\|x \| \le d$ for all $x \in K$.
\end{theorem}

\begin{theorem}[Theorem~2.7.1 of \cite{vanderVaart:1996tf}]
  \label{thm:1}
  Let $\mc L$ be a class of functions on $\prod_{i=1}^d [0,L_i],$ $0 < L_i < \infty$, such that for all $f \in \mc L$, we have $L_\infty(f) \le B < \infty $ and $f$ has Lipschitz constant in the direction $x_i$ given by $\Gamma_i < \infty$. Then
  \begin{equation}
    \label{eq:7}
    \log N( \epsilon, \mc L, L_\infty)
    \le K \lp \frac{B + \sum_{i=1}^d \Gamma_i L_i}{\epsilon} \rp^{d}.
  \end{equation}
\end{theorem}
\begin{proof}
  The theorem is given by \cite{vanderVaart:1996tf} when the domain is $[0,1]^d$ and the sup and Lipschitz bounds are all $1$.  A scaling argument gives the general form.
\end{proof}
\begin{mylongform}
  \begin{longform}
    They include the volume of domain ... 
  \end{longform}
\end{mylongform}

%%% Local Variables:
%%% mode: latex
%%% TeX-master: "bracketing"
%%% End:

\markboth{  List of Notation \hfill}{List of Notation \hfill} %%need
%% pagestyle myheadings?
\subsection*{Notation Index}
\label{sec:notation}
\begin{tabbing}
%%For some mysterious reason you need to do the first by hand?  confused.
%% But the number of "~"'s for the first dictates the number for the rest.
% "=" and ">" (see latex/lyx preamble) seem to have some significance?
% should perhaps understand the 'tabbing' environment...
%$(x)_+$~~~~~~~~~~~~~~~~~~~~\=\parbox{5.7in}{Class of log-concave densities
%\dotfill \pageref{symbol:plus}}\\
  $\Nb(\epsilon, {\cal F}, d) $~~~~~~~~~\=\parbox{4.4in}{
    bracketing number \dotfill  \pageref{symbol:bracketing-number}}\\
  \addnotation L_p: {$L_p(f)$ is $\lp \int f^p(x) dx \rp^{1/p}$ if $1 \le p < \infty$ or $\sup_x |f(x)|$ if $p=\infty$}{symbol:Lp}
  % \addnotation \cC\lp D, B, {\Gamma_1, \ldots, \Gamma_d, v_1,\ldots, v_d}\rp: {}{symbol:convex-fns-B-Gamma}
  % \addnotation \cC\lp D, B\rp: {}{symbol:convex-fns-B}
  \addnotation \mathcal{C(\cdot, \ldots, \cdot)}: {convex function
    classes}{symbol:convex-fns-B-Gamma}
  \addnotation \Vert \cdot \Vert: {Euclidean distance:  $\Vert z \Vert^2 :=
    \sum_{i=1}^d z_i^2$ for $z \in \RR^d$}{symbol:Euclidean-norm}
  \addnotationtwo D: {$d$-dimensional convex
    set; $d$-dimensional convex polytope}{symbol:D-general}{symbol:D-polytope}
  \addnotation \bd; \relbd: {boundary; relative boundary}{symbol:relbd}
  \addnotation d^+(x,G,v): {}{symbol:d+}
  \addnotation d(x,H,v): {}{symbol:d(x,H,v)}
  \addnotation d(E,H,v): {}{symbol:d(E,H,v)}
  \addnotation N: {number of halfspaces defining $D$}{symbol:N}
  \addnotation E_i: {halfspaces}{symbol:Ei}
  \addnotation H_i: {hyperplanes}{symbol:Hi}
  \addnotation F_i: {facets of convex polytope $D$}{symbol:Fi}
  \addnotation v_i: {(inner) normal vectors for $H_i$}{symbol:vi}
  \addnotation G_{\jb}: {faces of polytope $D$ }{symbol:vi}
  \addnotation G_{\ib,\jb}: {}{symbol:vi}
  \addnotation A_{\jb}, x_{\jb}: {}{symbol:Ajxj}
  \addnotation \gamma_{\jb}: {}{symbol:gammaj}
  \addnotation \delta_i, A, u : {}{symbol:deltai-u}
  \addnotation J_k: {}{symbol:Jk}
  \addnotation I_k: {}{symbol:Ik}
  \addnotationtwo E_{\jb}; e \equiv e_{\ib,\jb}: {See Proposition~\ref{prop:basis-lipschitz}}{symbol:e-axes-def1}{symbol:e-basis}
  \addnotation h(D,x): {support function}{symbol:h}
  \addnotation w(D): {width}{symbol:w}
  \addnotation \II_k: {multi-indices $\ib$ with $|\ib|=k$}{symbol:IIk}
\end{tabbing}

%% this doesn't seem to work

%%% Local Variables:
%%% mode: latex
%%% TeX-master: "bracketing"
%%% End:
 \clearpage

\bibliography{bracketing}

\end{document}